%% file: ms.tex
\documentclass[sigconf]{acmart}

\input{preamble}
\input{macros}

\copyrightyear{2023} 
\acmYear{2023} 
\setcopyright{acmlicensed}
\acmConference[PODS '23]{Proceedings of the 42nd ACM SIGMOD-SIGACT-SIGAI Symposium on Principles of Database Systems}{June 18--23, 2023}{Seattle, WA, USA} 
\acmBooktitle{Proceedings of the 42nd ACM SIGMOD-SIGACT-SIGAI Symposium on Principles of Database Systems (PODS '23), June 18--23, 2023, Seattle, WA, USA}
\acmPrice{15.00}
\acmDOI{10.1145/3584372.3588675}
\acmISBN{979-8-4007-0127-6/23/06}

\usepackage{tcolorbox}% http://ctan.org/pkg/tcolorbox
\usepackage{enumitem}
\usepackage{wrapfig}
\usepackage{balance} 
\usepackage{framed}   % for framing
{\endMakeFramed}
\definecolor{shadecolor}{gray}{0.75}
\begin{document}
\title{Space-Time Tradeoffs for Conjunctive Queries with Access Patterns}

%% remove the page numbering from the bottom of each page
\settopmatter{printfolios=true}

\author{Hangdong Zhao}
\affiliation{%
\institution{University of Wisconsin-Madison}
\city{Madison}
\state{WI}
\country{USA}
}
\email{hangdong@cs.wisc.edu}

\author{Shaleen Deep}
\affiliation{%
\institution{Microsoft Gray Systems Lab}
\city{Madison}
\state{WI}
\country{USA}}
\email{shaleen.deep@microsoft.com}

\author{Paraschos Koutris}
\affiliation{%
\institution{University of Wisconsin-Madison}
\city{Madison}
\state{WI}
\country{USA}}
\email{paris@cs.wisc.edu}
\renewcommand{\shortauthors}{Hangdong Zhao, Shaleen Deep, \& Paraschos Koutris}
\sloppy
 \begin{abstract}
In this paper, we investigate space-time tradeoffs for answering conjunctive queries with access patterns (CQAPs). The
goal is to create a space-efficient data structure in an initial preprocessing phase and use it for answering (multiple)
queries in an online phase. Previous work has developed data structures that trades off space usage for answering time
for queries of practical interest, such as the path and triangle query. However, these approaches lack a comprehensive
framework and are not generalizable. 
Our main contribution is a general algorithmic framework for obtaining space-time tradeoffs for any CQAP. Our framework builds upon the $\PANDA$ algorithm and tree decomposition techniques. We demonstrate that our framework captures all state-of-the-art tradeoffs that were independently produced for various queries. Further, we show surprising improvements over the state-of-the-art tradeoffs known in the existing literature for reachability queries.
 \end{abstract}

 \maketitle
 
\input{introduction}
\input{background}

\input{treeDecompositions}
\input{framework}

\input{2pd}
\input{results}

\input{related}
\input{conclusion}

\bibliographystyle{plain}  
\balance
\bibliography{ref}

\newpage
\onecolumn
\appendix
\input{appendix}

\input{hierarchical}

\end{document}

%% file: preamble.tex
\usepackage[linesnumbered,ruled,noend]{algorithm2e}

\usepackage{tikz}
\usetikzlibrary{positioning,calc}
\usetikzlibrary{decorations.text}
\usetikzlibrary{decorations.pathmorphing,decorations.pathreplacing}
\usetikzlibrary{arrows,petri, topaths,fit}

\usepackage{listings}
\usepackage{wrapfig}
\usepackage{framed}   % for framing
\renewenvironment{shaded}{%
  \MakeFramed{\advance\hsize-\width \FrameRestore\FrameRestore}}%
 {\endMakeFramed}
\definecolor{shadecolor}{gray}{0.75}
\usepackage{thmtools, thm-restate}

\usepackage{multirow}
\usepackage{subcaption}
\usepackage{url}
\usepackage{graphicx}
\usepackage{tcolorbox}% http://ctan.org/pkg/tcolorbox
\usepackage{enumitem}

\definecolor{mycolor}{rgb}{0.122, 0.435, 0.698}% Rule colour
\makeatletter
\newcommand{\mybox}[1]{%
  \setbox0=\hbox{#1}%
  \setlength{\@tempdima}{\dimexpr\wd0+13pt}%
  \begin{tcolorbox}[colframe=mycolor,boxrule=0.5pt,arc=4pt,
      left=6pt,right=6pt,top=3pt,bottom=3pt,boxsep=0pt,width=\@tempdima]
    #1
  \end{tcolorbox}
}

\newcommand{\hangdong}[1]{{\color{red}{\noindent\textsf{\textbf{Hangdong}: #1}}}}
\newcommand{\highlight}[1]{{\color{black}{\noindent #1}}}

\newcommand{\mC}{\mathcal{C}}

\newcommand{\bU}{\mathbf{U}}
\newcommand{\bZ}{\mathbf{Z}}
\newcommand{\bu}{\mathbf{u}}

\newcommand{\slack}{\alpha}

\newcommand{\anc}{\mathsf{anc}}

\newcommand{\vars}[1]{\mathsf{vars}(#1)}

%% file: macros.tex
\usepackage{mdframed}
\usepackage{graphicx}   % need for figures
\usepackage{hyperref}   % use for hypertext links, including those to external documents and URLs
\usepackage{bbm}
\hypersetup{
	colorlinks = true,
	citecolor=blue
}
\usepackage{thm-restate}
\usepackage{mdwlist}
\usepackage{mathtools}
\usepackage{xspace}
\usepackage{verbatim}   % useful for program listings
\usepackage{xcolor}
\usepackage{makecell}
\usepackage{tikz}
\usepackage{relsize}
\usepackage{booktabs}
\usepackage{tikz-qtree}
\usepackage{subcaption}
\usepackage{multirow}
\usetikzlibrary{arrows.meta}
\usetikzlibrary{fit,shapes,trees,shapes.geometric}
\usetikzlibrary{patterns,decorations.pathreplacing,calc}
\usetikzlibrary{matrix, positioning, arrows}
\usetikzlibrary{chains,shapes.multipart}
\usetikzlibrary{shapes,calc}
\usetikzlibrary{automata}
\usepackage[normalem]{ulem}

\newcommand{\shaleen}[1]{{\color{magenta}{\noindent\textsf{\textbf{Shaleen}: #1}}}}

\allowdisplaybreaks

\newcommand{\cut}[1]{}   % remove things inside the cut
\makeatletter
\newcommand{\mytag}[2]{%
	\text{#1}%
	\@bsphack
	\protected@write\@auxout{}%
	{\string\newlabel{#2}{{#1}{\thepage}}}%
	\@esphack
}
\makeatother

\newcommand{\setdisj}{\textsf{Set Disjointness }}

\newcommand{\introparagraph}[1]{\smallskip \noindent {{\bf \em #1.}}}  % define own new subsection type: noindent, bold (textsc)

\usepackage{aliascnt}

\newtheorem{observation}{Observation}[section]

\providecommand{\bx}[0]{\mathbf{x}}
\providecommand{\by}[0]{\mathbf{y}}

\providecommand{\bz}[0]{\mathbf{z}}

\providecommand{\bu}[0]{\mathbf{u}}

\providecommand{\ba}[0]{\mathbf{a}}

\providecommand{\mD}[0]{\mathcal{D}}

\providecommand{\mH}[0]{\mathcal{H}}

\providecommand{\mJ}[0]{\mathcal{J}}

\providecommand{\mP}[0]{\mathcal{P}}

\providecommand{\OBJ}[0]{\mathsf{OBJ}}

\providecommand{\edges}[0]{\mathcal{E}}

\providecommand{\vars}[1]{\textsf{vars}(#1)}

\providecommand{\poly}[0]{poly}

\providecommand{\tree}[0]{\mathcal{T}}

\providecommand{\htree}[0]{\mathcal{T}}

\providecommand{\mw}[0]{\mathsf{w}}
\providecommand{\mF}[0]{\mathcal{F}}
\providecommand{\bU}[0]{\boldsymbol{U}}

\providecommand{\bt}[0]{\mathbf{t}}

\providecommand{\eat}[1]{}
\providecommand{\mF}[0]{\mathsf{F}}

\newcommand{\parent}{\mathit{parent}}

\newcommand{\inv}{\raisebox{0.15ex}{$\scriptscriptstyle-\!1$}}

\providecommand{\blambda}[0]{\boldsymbol{\lambda}} %% Shannon-flow
\providecommand{\bdelta}[0]{\boldsymbol{\delta}} %%
\providecommand{\btheta}[0]{\boldsymbol{\theta}} %%
\providecommand{\bmu}[0]{\boldsymbol{\mu}} %% 
\providecommand{\bsigma}[0]{\boldsymbol{\sigma}} %% 
\providecommand{\bgamma}[0]{\boldsymbol{\gamma}} %% 

\providecommand{\DC}[0]{\mathsf{DC}}
\providecommand{\AC}[0]{\mathsf{AC}}

\providecommand{\SC}[0]{\mathsf{SC}}
\providecommand{\HDC}[0]{\mathsf{HDC}}
\providecommand{\HWC}[0]{\mathsf{HAC}}
\providecommand{\HSC}[0]{\mathsf{HSC}}
\providecommand{\PANDA}[0]{\mathsf{PANDA}}

\providecommand{\inflow}[0]{\mathsf{inflow}}

\providecommand{\bA}[0]{\boldsymbol{A}} %% matrices
\providecommand{\vectorb}[0]{\boldsymbol{b}}
\providecommand{\bD}[0]{\boldsymbol{D}}
\providecommand{\bd}[0]{\boldsymbol{d}}
\providecommand{\bC}[0]{\boldsymbol{C}}
\providecommand{\bc}[0]{\boldsymbol{c}}

\providecommand{\bg}[0]{\boldsymbol{g}}
\providecommand{\bzero}[0]{\mathbf{0}}
\providecommand{\bone}[0]{\mathbf{1}}

\providecommand{\polyO}[0]{\widetilde{O}}

\providecommand{\sfB}[0]{\mathsf{B}}
\providecommand{\sfBS}[0]{\mathsf{BS}}
\providecommand{\sfBT}[0]{\mathsf{BT}}

\providecommand{\bR}[0]{\mathbb{R}}  %% input database
\providecommand{\bQ}[0]{\mathbb{Q}}
\newcommand*{\defeq}{\stackrel{\text{def}}{=}}
\providecommand{\bh}[0]{\boldsymbol{h}} %% entropic vector
\providecommand{\bs}[0]{\boldsymbol{s}} %% submodularity step
\providecommand{\bm}[0]{\boldsymbol{m}} %% monotone step
\providecommand{\bc}[0]{\boldsymbol{c}} %% composition step
\providecommand{\bd}[0]{\boldsymbol{d}} %% decomposition step
\providecommand{\bL}[0]{\mathbf{L}} %% LPs

\providecommand{\hS}[0]{\textcolor{blue}{h_S}} 
\providecommand{\bhS}[0]{\textcolor{blue}{\boldsymbol{h}_S}} 
\providecommand{\hT}[0]{\textcolor{red}{h_T}} 
\providecommand{\hA}[0]{\textcolor{red}{h_T}} %%% can change back to pink
\providecommand{\bhT}[0]{\textcolor{red}{\boldsymbol{h}_T}} 

\providecommand{\proofstep}[0]{\boldsymbol{f}} %% any proof step
\providecommand{\proofseq}[0]{\mathsf{ProofSeq}} %% any proof sequence
\providecommand{\LogSizeBound}[0]{\mathsf{LogSizeBound}}
\providecommand{\btime}[0]{T} %% 
\providecommand{\bspace}[0]{S} %% 

\AtBeginDocument{%
}

%% file: introduction.tex
\section{Introduction}
\label{sec:intro}

We study a class of problems that splits an algorithmic task into two phases: the {\em preprocessing phase}, which computes a space-efficient data structure from the input, and the {\em online phase}, which uses the data structure to answer requests of a specific form over the input as fast as possible. Many important algorithmic tasks such as set intersection problems~\cite{Cohen2010,goldstein2017conditional}, reachability in directed graphs~\cite{agarwal2011approximate, agarwal2014space, cohen2010hardness}, histogram indexing~\cite{chan2015clustered, kociumaka2013efficient}, and problems related to document retrieval~\cite{afshani2016data, larsen2015hardness} can be expressed in this way. The fundamental algorithmic question related to these problems is to find {\em the tradeoff between the space $S$ necessary for storing the data structures and the time $T$ for answering a request.} 

Let us look at one of the simplest tasks in this setup. Consider the $2$-\setdisj problem: given a universe of elements $U$ and a collection of $m$ sets $S_1, \dots, S_m \subseteq U$, we want to create a data structure such that for any pair of integers $1 \leq i,j \leq m$, we can efficiently decide whether $S_i \cap S_j$ is empty or not. It is well-known that the space-time tradeoff for $2$-\setdisj is captured by the equation $S \cdot T^2 = O(N^2)$, where $N$ is the total size of all sets~\cite{Cohen2010, goldstein2017conditional}. Similar tradeoffs have also been established for other data structure problems. In the $k$-\textsf{Reachability} problem~\cite{goldstein2017conditional, Cohen2010} we are given as input a directed  graph $G = (V,E)$, an arbitrary pair of  vertices $u, v$, and the goal is to decide whether there exists a path of length $k$ between $u$ and $v$.  The data structure obtained was conjectured to be optimal by~\cite{goldstein2017conditional}, and the conjectured optimality was used to develop conditional lower bounds for other problems, such as approximate distance oracles~\cite{agarwal2011approximate, agarwal2014space} where no progress has been made in improving the upper bounds in the last decade. In the \textsf{edge triangle detection} problem~\cite{goldstein2017conditional}, we are given as input a graph $G = (V,E)$, and the goal is to develop a data structure that can answer whether a given edge $e \in E$ participates in a triangle or not. Each of these problems has been studied in isolation and therefore, the algorithmic insights are not readily generalizable into a comprehensive framework. 

In this paper, we cast many of the above problems into answering {\em Conjunctive Queries with Access Patterns (CQAPs)} over a relational database. For example, by using the relation $R(x,y)$ to encode that element $x$ belongs to set $y$, $2$-\setdisj can be captured by  the following CQAP: $\varphi (\mid y_1, y_2) \leftarrow R(x,y_1) \land R(x,y_2)$. The expression $\varphi (\mid y_1, y_2)$ can be interpreted as follows: given values for $y_1, y_2$, compute whether the query returns true or not. Different access patterns capture different ways of accessing the result of the CQ and result in different tradeoffs. %As we will see later, $k$-\textsf{Reachability} can also be easily captured by a CQAP.  

Tradeoffs for enumerating Conjunctive Query results under static and dynamic settings have been a subject of previous research~\cite{olteanu2016factorized, greco2013structural, deep2018compressed, CQAP, kara19, kara2019counting}. However, previous work either focuses on the tradeoff between preprocessing time and answering time~\cite{CQAP, kara19, kara2019counting}, or the tradeoff between space and delay in enumeration~\cite{olteanu2016factorized,deep2018compressed}. In this paper, we focus explicitly on the tradeoff between space and answering time, without optimizing for preprocessing time.  Most closely related to our setting is the problem of answering Boolean CQs~\cite{deep2021space}. In that work, the authors slightly improve upon the data structure proposed in~\cite{deep2018compressed} and adapt it for Boolean CQ answering. Further,~\cite{deep2021space} identified that the conjectured tradeoff for the $k$-reachability problem is suboptimal by showing slightly improved tradeoffs for all $k \geq 3$. The techniques used in this paper are quite different and a vast generalization of the techniques used in~\cite{deep2021space}. The proposed improvements in~\cite{deep2021space} for $k$-reachability are already captured in this work and in many cases, surpass the ones from~\cite{deep2021space}.

\introparagraph{Our Contribution}
Our key contribution is a general algorithmic framework for obtaining space-time tradeoffs for any CQAP. Our framework builds upon the \textsf{PANDA} algorithm~\cite{DBLP:conf/pods/Khamis0S17} and tree decomposition techniques~\cite{gottlob2014treewidth,Marx13}. Given any CQAP, it calculates a tradeoff that can find the best possible time for a given space budget. To achieve this goal, we need two key technical contributions. 

First, we introduce the novel notion of \emph{partially-materialized tree decompositions} (PMTDs) that allow us to capture different possible materialization strategies on a given tree decomposition (\autoref{sec:pmtd}). At a high level, a PMTD augments a tree decomposition with information on which bags should be materialized and which should be computed online.  To use a PMTD, we propose a variant of the Yannakakis algorithm (\autoref{Yannakakis}) such that during the online phase we incur only the cost of visiting the non-materialized bags. 

The second key ingredient is an extension of the \textsf{PANDA} algorithm~\cite{DBLP:conf/pods/Khamis0S17} that computes a disjunctive rule in two phases. The computation of a disjunctive rule allows placing an answer to any of the targets in the head of the rule. A key technical component in the \textsf{PANDA} algorithm is the notion of a \emph{Shannon-flow inequality}. For any Shannon-flow inequality, one can construct a proof sequence that has a direct correspondence with relational operators. Consequently, a proof sequence can be transformed into a join algorithm.
The disjunctive rules we consider are computed in two phases: in the first phase (preprocessing), we can place an answer only to targets that will be materialized during the preprocessing phase. In the second phase (online), we place an answer to the remaining targets. We call these rules {\em 2-phase disjunctive rules} (\autoref{sec:TPDR}). To achieve this 2-phase computation, we introduce a type of Shannon-flow inequalities, called \emph{joint Shannon-flow inequalities} (\autoref{sec:2pd}), such that each inequality gives rise to a space-time tradeoff. The joint Shannon-flow inequality generates two parallel proof sequences, one proof sequence for the preprocessing phase and another proof sequence for the answering phase. This transformation allows us to use the \textsf{PANDA} algorithm as a blackbox on each of the proof sequences independently and is instrumental in achieving  space-time tradeoffs. 

We demonstrate the versatility of our framework by recovering state-of-the-art space-time tradeoffs for Boolean CQAPs, $2$-\setdisj as well as its generalization $k$-\setdisj, and $k$-\textsf{Reachability} (\autoref{sec:results}). We also apply our framework to the previously unstudied setting of space-time tradeoffs (in the static setting) for access patterns over a subset of \emph{hierarchical queries}, a fragment of acyclic CQs that is of great interest~\cite{dalvi2009probabilistic, berkholz2017answering, kara19, idris2017dynamic, deep2021enumeration, bonifati2020analytical}. Interestingly, we can recover strategies that are very similar to how specialized enumeration algorithms with provable guarantees work for this class of CQs~\cite{kara19, deep2021enumeration}. More importantly, we improve state-of-the-art tradeoffs. Our most interesting finding is that we can obtain complex tradeoffs for $k$-\textsf{Reachability}  that exhibit different behavior for different regimes of $S$. For the $3$-\textsf{Reachability} problem, we show how to improve the tradeoff for a significant part of the spectrum. For the $4$-\textsf{Reachability} problem, we are able to show (via a rather involved analysis) that the space-time tradeoff can be improved \emph{everywhere} when compared to the conjectured optimal! These results falsify the proposed optimal tradeoff of $S \cdot T^{2/(k-1)} = \polyO(|E|^2)$ for $k$-\textsf{Reachability} for regimes that are even larger than what was shown in~\cite{deep2021space}.

\smallskip
\introparagraph{Organization} We introduce the basic terminology and problem definition in~\autoref{sec:background}. In~\autoref{sec:pmtd}, we describe the augmented tree decompositions that are necessary for our framework.~\autoref{sec:framework} introduces the general framework while~\autoref{sec:2pd} presents the algorithms used in our framework. We present the applications of the framework in~\autoref{sec:results}. The related work is described in~\autoref{sec:related} and we conclude with a list of open problems in~\autoref{sec:conclusion}.

%% file: background.tex
\section{Background} 
\label{sec:background}

 \introparagraph{Conjunctive Query} We associate a Conjunctive Query (CQ) $\varphi$ with a hypergraph $\mH = ([n], \edges)$, where $[n] = \{1, \dots, n\}$ and $\edges \subseteq 2^{[n]}$. The body of the query has atoms $R_F$, where $F \in \edges$. To each node $i \in [n]$, we associate a variable $x_i$. The CQ is then
 $$
 \varphi(\bx_H) \leftarrow \bigwedge_{F \in \edges} R_F(\bx_F),
 $$
 where $\bx_I$ denotes the tuple $(x_i)_{i \in I}$ for any $I \subseteq [n]$.
The variables in $\bx_H$ are called the {\em head variables} of the CQ.  The CQ is \textit{full} if $H = [n]$ and \textit{Boolean} if $H = \emptyset$. We use $\varphi$ to denote the output of the CQ $\varphi$.
% The output size of $\varphi$ is denoted as $|\varphi|$.
%is the head variables and we denote by $\vars{R_F}$ the variables contained in atom $R_F$.
 
 \introparagraph{Degree Constraints}
  A \textit{degree constraint} is a triple $(X, Y, N_{Y|X})$ where $X \subset Y \subseteq [n]$ and $N_{Y|X}$ is a natural number. A relation $R_F$ is said to \textit{guard} the degree constraint $(X, Y, N_{Y|X})$ if $X \subset Y \subseteq F$ and for every tuple $\bt_X$ (over $X$),
 $ \max_{\bt_X} \deg_F (Y|\bt_X) \leq N_{Y|X},$
 where $\deg_F(Y|\bt_X) = \left|\Pi_Y(\sigma_{X=\bt_X} R_F)\right|$. We use $\DC$ to denote a set of degree constraints and say that $\DC$ is guarded by a database instance $\mD$ if every $(X, Y, N_{Y|X}) \in \DC$ is guarded by some relation in $\mD$. A degree constraint $(X, Y, N_{Y|X})$ is a \textit{cardinality constraint} if $X = \emptyset$. Throughout this work, we make the following assumptions on $\DC$ guarded by a database instance $\mD$:
 \begin{itemize}
     \item \textit{(best constraints assumption)} w.l.o.g, for any $X \subset Y \subseteq [n]$,
     there is at most one $(X, Y, N_{Y|X}) \in \DC$. This assumption can be maintained by only keeping the minimum $N_{Y|X}$ if there is more than one.
     \item for every relation $R_F \in \mD$, there is a \textit{cardinality constraint} $(\emptyset, F, |R_F| \defeq N_{F|\emptyset}) \in \DC$. The \textit{size} of the database $\mD$ is denoted as $|\mD| \defeq \max_{R_F \in \mD} |R_F|$.
 \end{itemize}
 In this work, we use degree constraints to measure data complexity. All logs are in base $2$, unless otherwise stated.
%   The \textit{size} of a database instance $\mD$ is denoted as $|\mD| \defeq \max_{F \in \edges} N_{F|\emptyset}$.

\subsection{CQs with Access Patterns}

We define CQs with access patterns following the definition from~\cite{CQAP}:

\begin{definition}[CQ with access patterns] 
A Conjunctive Query with {\em Access Patterns} (CQAP) is an expression of the form 
$$ 
\varphi(\bx_H \mid \bx_A) \leftarrow \bigwedge_{F \in \edges} R_F(\bx_F),
$$
where $A \subseteq [n]$ is called the {\em access pattern} of the query.
\end{definition}

The access pattern tells us how a user accesses the result of the CQ. In particular, the user will provide an instance of a relation $Q_A(\bx_A)$, which we call an {\em access request}. The task is then to return the result of the following CQ, denoted as $\varphi$, where
$$ \varphi(\bx_H) \leftarrow Q_A(\bx_A) \wedge \bigwedge_{F \in \edges} R_F(\bx_F).$$ 
We call $\varphi$ the \textit{access CQ}. The most natural access request is one where $|Q_A|=1$; in other words, the user provides only one fixed value for every variable $x_i, i \in A$. This can be thought of as using the CQ result as an index with search key $\bx_A$. By allowing the access request $Q_A$ to consist of more tuples, we can capture other scenarios. For example, one can take a stream of access requests of size $1$ and batch them together to obtain a (possibly faster) answer for all of them at once. Prior work~\cite{deep2018compressed,CQAP} has only considered the case where $|Q_A|=1$. 

\eat{
\introparagraph{Discussion} In~\cite{CQAP}, the authors define the head of a CQAP to be $ \varphi(\bx_{H \cup A})$ instead. This makes no difference when $|Q_A|=1$, but is a more restrictive definition when $|Q_A| > 1$. 
}

\subsection{Problem Statement} \label{ps}

Let $\varphi(\bx_H \mid \bx_A)$ be a CQAP under degree constraints $\DC$ guarded by the input relations. In addition, we denote by $\AC$ another set of degree constraints known in prior, guarded by any access request $Q_A$. Similar to $\DC$, we that assume there is a cardinality constraint $(\emptyset, A, |Q_A| = N_{A|\emptyset}) \in \AC$ guarded by $Q_A$. For example, the case where $|Q_A| = 1$ can be interpreted as a cardinality constraint $(\emptyset, A, 1) \in \AC$. 
Given a database instance $\mD$ guarding $\DC$, our goal is to construct a data structure, such that we can answer any access request as fast as possible. More formally, we split query processing into two phases:
\begin{description}
 \item [{\em Preprocessing phase:}] it constructs a data structure in space $\polyO(\bspace)$\footnote{The notation $\polyO$ hides a polylogarithmic factor in $|\mD|$.}. The overall space cost takes the form $\polyO(\bspace + |\mD|)$, where $\bspace$ is called the \textit{intrinsic space cost} of the data structure and $|\mD|$ is the (unavoidable) space cost for storing the database.
     \item [ {\em Online phase:}]
     given an access request $Q_A$ (guarding $\AC$), it returns the results of the access CQ $\varphi$ using the data structure built in the preprocessing phase. The (worst-case) answering time is then $\polyO(\btime + |Q_A|) +  O(|\varphi|)$, where $\btime$ is called the \textit{intrinsic time cost} and $|Q_A| + |\varphi|$ is the (unavoidable) time cost of reading the access request $Q_A$ and enumerating the output. For the Boolean case and when $|Q_A| = 1$, the answering time simply becomes $\polyO(\btime)$. 
 \end{description}
In this work, we study the tradeoffs between the two intrinsic quantities, $\bspace$ and $\btime$, which we will call an {\em intrinsic tradeoff}.
 At one extreme, the algorithm stores nothing, thus $\bspace = O(1)$, and we answer each access request from scratch. At the other extreme, the algorithm stores the results of the CQ $\varphi_M(\bx_{H \cup A}) \leftarrow \bigwedge_{F \in \edges} R_F(\bx_F)$
 as a hash table with index key $\bx_A$. For any access request $Q_A$, we simply evaluate the query $\varphi(\bx_H) \leftarrow Q_A \wedge \varphi_M$ in the online phase by probing each tuple of $Q_A$ in the hash table. If $H \supseteq A$, then any access request can be answered in (instance-optimal) time $O(|Q_A| + |\varphi|)$, in which case $\btime = O(1)$.
 
%  To depict the relationships between the two intrinsic quantities $\bspace$ and $\btime$, we introduce the following definition:
% \begin{definition}[Tradeoff]
%  A \textit{(smooth space-time) tradeoff} is an asymptotic expression of $\bspace$ and $\btime$ of the form
%  \begin{align}\label{tradeoff}
%       \bspace^{\alpha} \cdot \btime^{\beta} \cong \prod_{(X, Y, N_{Y|X}) \in \DC \cup \AC} N_{Y|X}^{w_{Y|X}},
%  \end{align}
%  where the exponents $\alpha, \beta, w_{Y|X} \in \bQ_+$ and we use $\cong$ to hide a poly-logarithmic factor at the right-hand side.
% \end{definition}
 
%  An algorithm is said to attain a (smooth space-time) trade-off of \eqref{tradeoff} if it follows the two-phase paradiam and can obtain any value of $\btime$ from tuning $\bspace$, as predicated by \eqref{tradeoff}. 

%\introparagraph{Examples}
%We give below a few examples of how space-time tradeoff problems can be cast within our problem setup.

\begin{example}[$k$-Set Disjointness]
In this problem, we are given sets $S_1, \dots, S_m$ with elements drawn from the same universe $U$. Each access request asks whether the intersection between $k$ sets is empty or not. By encoding the family of sets as a binary relation $R(y, x)$ such that element $y$ belongs to set $x$, we can express the problem as the following CQAP:
\begin{align}\label{k-disjoint-boolean}
    \varphi(  \mid \bx_{[k]}) \leftarrow \bigwedge_{i \in [k]} R(y, x_i).
\end{align}
If we also want to enumerate the elements in their intersection, we would instead use the non-Boolean version:
\begin{align}\label{k-disjoint-nonboolean}
    \varphi( y  \mid \bx_{[k]}) \leftarrow \bigwedge_{i \in [k]} R(y, x_i).
\end{align}

\end{example}

\begin{example}[$k$-Reachability]\label{example:k-reachability}
Given a direct graph $G$ , the $k$-reachability problem asks, given a pair vertices $(u, v)$, to check whether they are connected by a path of length $k$. Representing the graph as a binary relation $R(x, y)$, we can model this problem through the following CQAP (the $k$-path query):
$$ \phi_k(  \mid x_1, x_{k+1}) \leftarrow \bigwedge_{i \in [k]} R(x_i, x_{i+1}).$$
We can also check whether there is a path of length at most $k$ by combining
the results of $k$ such queries (one for each $1, \dots, k$).
\end{example}

%\begin{example}[Edge Triangle Detection]
%Given a graph $G = (V, E)$, this problem asks, given an edge $(u, v)$ as the access request, whether $(u, v)$ participates in a triangle or not. This task can be expressed as the following CQAP:
%%
%$$ \varphi(  \mid x,z) \leftarrow R(x, y) \wedge R(y, z) \wedge R(x, z).$$
%\end{example}

In this work, we focus on the CQAP such that $H \supseteq A$. If we are given a CQAP where $H \nsupseteq A$, we replace the head of the CQAP with $\varphi(\bx_{H \cup A} \mid \bx_A)$, and simply project on the desired results in the end.

%% file: treeDecompositions.tex
\section{Partially Materialized Tree Decompositions}
\label{sec:pmtd}

In this section, we introduce a type of tree decomposition that augments a decomposition with information about what bags we want to materialize.

\begin{definition}[Tree Decomposition] A {\em tree decomposition} of a hypergraph $\mH = ([n], \edges)$ is a pair 
    $(\htree, \chi)$ where $(i)$ $\htree$ is an undirected tree, and $(ii)$ $\chi: V(\htree) \rightarrow 2^{[n]}$ is a mapping that assigns to every node $t \in V(\htree)$ a subset of $[n]$,  called the {\em bag} of $t$, such that 
    	\begin{enumerate}
    	    \item [(1)] For every hyperedge $F \in \edges$,  the set $F$ is contained in some bag; and
    	    \item [(2)] For each vertex $x \in [n]$, the set of nodes $\{t \mid x \in \chi(t) \}$ forms a (non-empty) connected subtree of $\htree$.
    	\end{enumerate}
\end{definition}

Take a tree decomposition $(\htree, \chi)$ and a node $r \in V(\htree)$. We define $\textsf{TOP}_r(x)$ as the highest node in $\htree$ containing $x$ in its bag if we root the tree at $r$. We now say that $(\htree, \chi)$ is \textit{free-connex w.r.t. $r $} if for any $x \in H$ and $y \in [n] \setminus H$, $\textsf{TOP}_r(y)$ is not an ancestor of $\textsf{TOP}_r(x)$~\cite{Secure}. We say that  $(\htree, \chi)$ is free-connex if it is free-connex w.r.t. some  $r \in V(\htree)$.

% \begin{definition}[$C$-connex Tree Decomposition]
% Let $\mH = ([n], \edges)$ be a hypergraph. Given a set of nodes $C \subseteq [n]$,  we say that a tree decomposition  $(\htree,\chi)$ of $\mH$ is $C$-connex if there is a connected subset $U$ of $V(\htree)$ such that the union of the bags in $U$ is exactly $C$. It is  a {\em free-connex tree decomposition} if $C$ is the head variables of $\varphi$, i.e., $C = H$.
% \end{definition}

  We can now introduce our key concept of a partially materialized tree decomposition, tailored for CQAPs. Let $\varphi( \bx_H \mid \bx_A)$ be a CQAP such that $A \subseteq H$. Let $\mH$ be the hypergraph associated with $\varphi(\bx_H)$, the access CQ.

 \begin{definition}[PMTD]
 A {\em Partially Materialized Tree Decomposition (PMTD)} of the CQAP $\varphi( \bx_H \mid \bx_A)$ with $H \supseteq A$ is a tuple $(\htree,\chi, M, r)$ such that the following properties hold:
\begin{enumerate}
\item $(\htree, \chi)$ is a free-connex tree decomposition of $\mH$ w.r.t. node $r$, called the root ; and
\item  $A \subseteq \chi(r)$ ; and
\item  $M \subseteq V(\htree)$ such that whenever $t \in M$ then all the nodes of its subtree (w.r.t. orienting the tree away from $r$) are in $M$.
\end{enumerate}
\end{definition}

%The root $r$ defines a rooted directed tree from $\htree$ by orienting each edge away from $r$.

Given a PMTD $(\htree,\chi, M, r)$, we call $M$ the {\em materialization set}. We also associate with each node $t \in V(\htree)$ a {\em view} with variables $\bx_{\nu(t)}$, where the mapping $\nu: V(\htree) \rightarrow 2^{[n]}$ is defined as follows. If the node $t \notin M$, then $\nu(t) \defeq \chi(t)$ and the view is of the form $T_{\nu(t)}(\bx_{\nu(t)})$, called a \textit{$T$-view}. Otherwise, $t \in M$. If $t = r \in M$, define $\nu(r) \defeq \chi(t) \cap H$. Let $p$ be the parent node of a non-root node $t \in M$ and define
 \begin{align*}
      \nu(t) & \defeq \begin{cases} 
      \chi(t) \cap (H \cup \chi(p)) 
      & \text { if } p \notin M  \\
      \chi(t) \cap H & \text { if } p \in M \text { and } \chi(t) \cap H \nsubseteq \chi(p) \cap H  \\
      \emptyset & \text { if } p \in M \text { and } \chi(t) \cap H \subseteq \chi(p) \cap H.
       \end{cases}
 \end{align*}
 %
% \begin{enumerate}
%     \item [(1)] if the node $t \in M \cap U$, then $\nu(t) = \chi(t)$; and
%     \item [(2)] if the node $t \in M \setminus U$ and no ancestor of $t$ is in $M$, $\nu(t)$ is the set of variables in $\chi(t)$ that are either common with its parent node or are in $H$; and
%     \item [(3)] if the node $t \in M \setminus U$ and there is an ancestor of $t$ in $M$, $\nu(t) = \emptyset$ (its $S$-view is always empty).
% \end{enumerate}
The view (for each $t \in M$) then is of the form $S_{\nu(t)}(\bx_{\nu(t)})$, called the \textit{$S$-view}. This definition of $S$-views corresponds to running a bottom-up semijoin-reduce pass of the Yannakakis algorithm in the materialization set $M$ of the free-connex tree decomposition $(\htree, \chi)$. Indeed, any variables in $\chi(t) \setminus \nu(t)$ are safely projected out after the semijoin-reduce. 
 
 On a high level, $M$ specifies the type of views associated with each bag ($S$-view or $T$-view), and $\nu(\cdot)$ pinpoints the schema of that view (possibly empty). A PMTD appoints its $S$-views to be materialized in the preprocessing phase and its $T$-views to be computed in the online phase. In the case where $M = \emptyset$, every view in the decomposition is obtained in the online phase. When $H = A$ or $H = [n]$, the free-connex property does not put any additional restrictions on the tree decompositions for a PMTD.

\begin{example}\label{example:3path}
We use the CQAP for 3-reachability as an example: 
$$\phi_3(x_1, x_4 \mid x_1, x_4) \leftarrow R_{1}(x_1, x_2) \wedge R_{2}(x_2, x_3) \wedge R_{3}(x_3, x_4).$$
Here, $(x_1, x_4)$ is the access pattern.~\autoref{fig:cfhw} shows three PMTDs for the above query, along with the associated views of each bag in each PMTD. The leftmost PMTD has an empty materialization set. The middle PMTD materializes the bag $\{x_1, x_2, x_3\}$ but the associated view $S_{13}$ projects out $x_2$. The rightmost PMTD materializes the only bag $\{x_1, x_2, x_3, x_4\}$ but the view $S_{14}$ keeps only the variables $x_1, x_4$. 
%We use the left tree decomposition in ~\autoref{example:3path} to show how every tree decompositions of the $3$-path Boolean adorned CQ $\varphi^{[(x_1, x_4)]}$ is constructed. There are three possible preprocessing sets, each generates a tree decomposition for $\varphi^{[(x_1, x_4)]}$ following the compression steps described in \autoref{sec:compress}.
\end{example}
 
\introparagraph{Redundancy \& Domination}
  We say that a tree decomposition is {\em non-redundant} if no bag is a subset of another bag. We say that a tree decomposition $(\htree_1, \chi_1)$ is {{\em dominated}} by another tree decomposition $(\htree_2, \chi_2)$ if every bag of $(\htree_1, \chi_1)$ is a subset of some bag of $(\htree_2, \chi_2)$. Here, we will generalize both notions to PMTDs.
  
  \begin{definition}[PMTD Redundancy]
  A PMTD $(\htree,\chi, M, r)$ is \textit{non-redundant} if $(1)$ for $t \in M$, $\nu(t) \neq \emptyset$ and no $\nu(t)$ is a subset of another; and $(2)$ for $t \notin M$, no $\nu(t)$ is a subset of another.
%   \hangdong{the root can be redundant actually, to contain just A, so no $\nu(t)$ (other than the root) can be a subset of another bag}
% \hangdong{take the bag $t$ contains $A$, there is no non-head vars above it by free-connex, build a root $r$ of $A$ above $t$, make $t$'s parent become $r$'s child}
 \end{definition}

\begin{definition}[PMTD Domination]
 A PMTD $(\htree_1,\chi_1, M_1, r_1)$ is dominated by another PMTD $(\htree_2,\chi_2, M_2, r_2)$ if $(1)$ for every node $t_1 \in M_1$, there is some node $t_2 \in M_2$ such that $\nu(t_1) \subseteq \nu(t_2)$, and $(2)$ for every node $t_1 \in V(\htree_1) \setminus M_1$, there is some node $t_2 \in V(\htree_2) \setminus M_2$ such that $\nu(t_1) \subseteq \nu(t_2)$.
\end{definition}

For PMTDs, both redundancy and domination are defined using the materialization set and views instead of the bags. For PMTDs with $M = \emptyset$, both PMTD redundancy and domination become equivalent to the standard definition. 

\begin{example}
Continuing Example~\ref{example:3path}, suppose we consider a PMTD with that takes the same tree decomposition as the left PMTD, but with both bags in the materialization set. The $S$-view associated with the root bag is $S_{14}$, and $S_{\emptyset}$ for the child bag; thus, this PMTD is redundant. Moreover, suppose we consider a PMTD with one bag $\{x_1, x_2, x_3,x_4\}$ which is the root, but is not in $M$. The $T$-view associated with this bag is $T_{1234}$; thus, this PMTD dominates the left PMTD in~\autoref{fig:cfhw}. On the other hand, all PMTDs in~\autoref{fig:cfhw} are non-redundant and non-dominant to each other.
\end{example}

 As we later suggest in our general framework, we {mostly} focus on sets of non-redundant and non-dominant PMTDs. Note that a non-redundant PMTD  $(\htree,\chi, M, r)$ satisfies $\nu(t) \neq \emptyset$, for any $t \in V(\htree)$, thus we can safely assume that all views are non-empty.

	\begin{figure}[t]
	    \begin{subfigure}{0.3\linewidth}
				%\vspace{2em}
					\centering
				\scalebox{1}{\begin{tikzpicture}
				\tikzset{edge/.style = {->,> = latex'},
					vertex/.style={circle, thick, minimum size=5mm}}
				\def\x{0.25}
				
				\begin{scope}[fill opacity=1]
				
				%\draw[] (0,-2) ellipse (0.8cm and 0.33cm) node {\small ${x_1, x_4}$};
				%\node[vertex]  at (1.1,-2) {$r$};
				\draw[] (0,-3.5) ellipse (0.8cm and 0.33cm) node {\small \small ${x_1, x_3, x_4}$};
				\node[vertex]  at (1.2,-3.5) {$T_{134}$};				
				\draw[] (0,-5) ellipse (0.8cm and 0.33cm) node {\small \small ${x_1, x_2, x_3}$};						
				\node[vertex]  at (1.2,-5) {$T_{123}$};								
				%\draw[edge] (0,-2.33) -- (0,-3.2);
				\draw[edge] (0,-3.83) -- (0,-4.65);			
				\end{scope}	
				\end{tikzpicture} 
				}
		\end{subfigure}
		\begin{subfigure}{0.3\linewidth}
				%\vspace{2em}
					\centering
			\scalebox{1}{\begin{tikzpicture}
				\tikzset{edge/.style = {->,> = latex'},
					vertex/.style={circle, thick, minimum size=5mm}}
				\def\x{0.25}
				
				\begin{scope}[fill opacity=1]
				
				%\draw[] (5,-2) ellipse (0.8cm and 0.33cm) node {\small ${x_1, x_4}$};
				%\node[vertex]  at (6.1,-2) {$r$};
				\draw[] (5,-3.5) ellipse (0.8cm and 0.33cm) node {\small \small ${x_1, x_3, x_4}$};
				\node[vertex]  at (6.2,-3.5) {$T_{134}$};				
				\draw[fill=black!10] (5,-5) ellipse (0.8cm and 0.33cm) node {\small \small ${x_1, x_2, x_3}$};						
				\node[vertex]  at (6.2,-5) {$S_{13}$};								
				%\draw[edge] (5,-2.33) -- (5,-3.2);
				\draw[edge] (5,-3.83) -- (5,-4.65);			
				\end{scope}	
				\end{tikzpicture} 
				}
		\end{subfigure}
		\begin{subfigure}{0.3\linewidth}
				%\vspace{2em}
					\centering
								\scalebox{1}{\begin{tikzpicture}
				\tikzset{edge/.style = {->,> = latex'},
					vertex/.style={circle, thick, minimum size=5mm}}
				\def\x{0.25}
				
				\begin{scope}[fill opacity=1]
				
				%\draw[] (0,0) ellipse (0.8cm and 0.33cm) node {\small ${x_1, x_4}$};
				%\node[vertex]  at (1.1,0) {$r$};
				\draw[fill=black!10] (0,-3.5) ellipse (1cm and 0.33cm) node {\small \small ${x_1, x_2, x_3, x_4}$};
				\node[vertex]  at (1.4,-3.5) {$S_{14}$};	\node[vertex]  at (1.4,-5.1) { };	
				%\draw[edge] (0,-0.33) -- (0,-1.2);
				\end{scope}	
				\end{tikzpicture} 
				}
		\end{subfigure}
		
		\caption{Three PMTDs for the 3-reachability CQAP. The materialized nodes are shaded and labeled as $S$-views.}
		\label{fig:cfhw}	
	\end{figure}
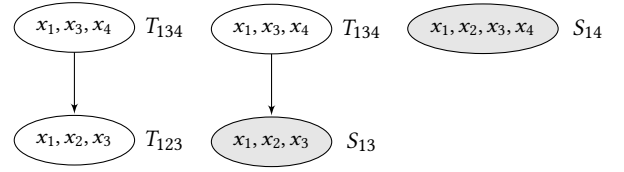

 \subsection{Online Yannakakis for PMTDs} 
 \label{Yannakakis}
 
We introduce an adaptation of the Yannakakis algorithm~\cite{Yannakakis} for a non-redundant PMTD (so no empty views), called \textit{Online Yannakakis}. Recall that for a non-redundant PMTD, the $S$-views, one for each $t \in M$, are stored in the preprocessing phase, while the $T$-views, one for each $t \in \htree \setminus M$, and the access request $Q_A$ are accessible only in the online phase. 
%Assume that we have obtained $S$-views with size $\polyO(S)$ and $T$-views with size $\polyO(T)$ (the attainment of views is discussed later). 
%Additionally, we assume that every $S$-view is semijoin-reduced with $\bigwedge_{F \in \edges} R_F$.
  
 \begin{theorem}\label{lem:Yannakakis}
 Consider a PMTD $(\htree, \chi, M, r)$ and its view $\nu(\cdot)$. Given $S$-views, we can preprocess them in space linear in their size such that we can compute the free-connex acyclic CQ 
  \begin{align}\label{acyclicjoin}
    \psi (\bx_H) \leftarrow Q_A \wedge \bigwedge_{t \in M} S_{\nu(t)}  \wedge \bigwedge_{t \in V(\htree) \setminus M} T_{\nu(t)}
 \end{align}
for any $T$-view and $Q_A$ in time $O(\max_{t \in V(\htree) \setminus M} |T_{\nu(t)}| + |Q_A| +  |\psi|)$, where $|\psi|$ is the output size of \eqref{acyclicjoin}.
 \end{theorem}

Note that the time cost has no dependence on the size of $S$-views, because throughout Online Yannakakis, $S$-views will be only used {for hash probing} in semijoin operations. We defer the details of the algorithm and the proof of its correctness to Appendix~\ref{sec:pmtd:appendix}.

%% file: framework.tex
 \section{General Framework}
\label{sec:framework}

Consider a CQAP $\varphi(\bx_H \mid \bx_A)  \leftarrow \bigwedge_{F \in \edges} R_F(\bx_F) $ with $H \supseteq A$. Recall that our goal is to find the best space-time tradeoffs under degree constraints $\DC$ (guarded by input relations) and $\AC$ (guarded by any access requests $Q_A$), as specified in \autoref{ps}. Our main algorithm is parameterized by:
\begin{itemize}
\item  $\mP = \{P_i \}_{i \in I}$, a (finite) indexed set of non-redundant and non-dominant PMTDs such that $P_i = (\htree_i, \chi_i, M_i, r_i)$ for every $i \in I$. Including all such PMTDs in $\mP$ (which are finite) will result in the best possible tradeoff. However, as we will see later, it is meaningful to consider smaller sets of PMTDs that result in more interpretable space-time tradeoffs.
\item $S$, the space budget.
\end{itemize}

 \subsection{2-Phase Disjunctive Rules} \label{sec:TPDR}
 
In this section, we define a specific type of disjunctive rule that will be necessary to acquire the $S$-views and $T$-views for PMTDs. We start by recalling the notion of a disjunctive rule. A disjunctive rule has the exact body of a CQ, while the head is a disjunction of output relations $T_B(\bx_B)$, which we call {{\em targets}}. Let $\sfBT \subseteq 2^{[n]}$ be a non-empty set, then a {\em disjunctive rule} $\rho$ takes the form:
 \begin{equation}\label{def:disjunctiveRule}
     \rho: \quad \bigvee_{B \in \sfBT} T_{B}(\bx_{B}) \leftarrow \bigwedge_{F \in \edges} R_F(\bx_F).
 \end{equation}
 Given a database instance $\mD$, a {{\em model}} of $\rho$ is a tuple $(T_B)_{B \in \sfBT}$ of relations, one for each target, such that the logical implication indicated by \eqref{def:disjunctiveRule} holds. More precisely, for any tuple $\ba$ that satisfies the body, there is a target $T_B \in (T_B)_{B \in \sfB}$ such that $\Pi_{B}(\ba) \in T_B$. The {{\em size}} of a model is defined as the maximum size of its output relations and the {\em output size} of a disjunctive rule $\rho$, denoted as $|\rho|$, is defined as the minimum size over all models.

% We write $(T_B)_{B \in \mM} \models \rho$ to denote the fact that $(T_B)_{B \in \mM}$ is a model of $\rho$. The {{\em size}} of a model is defined as the maximum size of output relations and the {\em output size} of a disjunctive rule $\rho$ is defined as the minimum size over all models:
%  $$ |\rho(\mD)| \defeq \min_{(T_B)_{B \in \mB} \models \rho} \max_{B \in \mB} |T_B|
%  $$
%   If the disjunctive rule has only one target, then it becomes a CQ, where its model is any superset of the answer of the CQ, and its output size is the output size of the CQ. 
 \smallskip   
 For our purposes, we define a type of disjunctive rules, called \textit{2-phase disjunctive rules}.
 \begin{definition}[2-phase Disjunctive Rules] A {{\em 2-phase disjunctive rule}} $\rho$ defined by a CQAP $\varphi(\bx_H \mid \bx_A)$ is a single disjunctive rule that takes the body of the access CQ $\varphi$, while the head has two sets of output relations. In other words, $\rho$ takes the form
     \begin{equation}\label{def:partitionDisjunctiveRule}
          \rho: \quad \bigvee_{B \in \sfBS} S_{B}(\bx_{B}) \vee \bigvee_{B \in \sfBT} T_{B}(\bx_{B}) \leftarrow Q_A(\bx_A) \wedge \bigwedge_{F \in \edges} R_F(\bx_F), 
     \end{equation}
     where $\sfBS, \sfBT \subseteq 2^{[n]}$ and at most one can be empty. A \textit{model} of $\rho$ thus consists of two sets of output relations,
     i.e. the \textit{$S$-targets} $(S_{B})_{B \in \sfBS}$ and the \textit{$T$-targets} $(T_B)_{B \in \sfBT}$.  
    %  We implicitly assume that in $\sfBS$ (and in $\sfBT$), no bag is a subset of another, if not, we can safely remove the larger bag from $\sfBS$ (or $\sfBT$). 
 \end{definition}
 
 As the name suggests, a model of a 2-phase disjunctive rule $\rho$ is computed in two phases, the preprocessing and online phase:
  \begin{description}
     \item [Preprocessing phase:] we obtain the $S$-targets $(S_B)_{B \in \sfBS}$ using a {\textit{preprocessing disjunctive rule}}
  \begin{align} \label{rule:preprocess}
     \rho_S & : \quad \bigvee_{B \in \sfBS} S_{B}(\bx_{B}) \leftarrow \bigwedge_{F \in \edges} R_F(\bx_F),
 \end{align}
  The space cost for storing the $S$-targets is $\polyO(\bspace_\rho)$, and the overall space cost is $\polyO(\bspace_\rho + |\mD|)$. The preprocessing phase has no knowledge of $Q_A$ except for the degree constraints $\AC$, so as to explicitly force the $S$-targets to be universal for any instance of access request.
     \item [Online phase:] given an access request $Q_A$ (under $\AC$), we obtain the $T$-targets $(T_B)_{B \in \sfBT}$ using an {\textit{online disjunctive rule}}
 \begin{align} \label{rule:online}
       \rho_T & : \quad \bigvee_{B \in \sfBT} T_{B}(\bx_{B}) \leftarrow Q_A(\bx_A) \wedge \bigwedge_{F \in \edges} R_F(\bx_F)
 \end{align}
     in time and space $\polyO(\btime_\rho)$. The overall time is $\polyO(\btime_\rho + |Q_A|)$.
   \end{description}
  
If $\sfBS = \emptyset$, then the model is computed from scratch in the online phase (and vice versa).  
As in \autoref{ps}, our focus is on analyzing the space-time tradeoffs between the two intrinsic quantities, $\bspace_\rho$ and $\btime_\rho$.

For the next part, assume that we have a 2-phase algorithm (called $\twoPP$) that, given a space budget $S$, has a preprocessing procedure $\twoPPp$ using space $S_\rho \leq S$ and an online procedure $\twoPPo$ using time (and space) $T_\rho$. We will discuss this algorithm in the next section.

% \subsection{Any Proper Adorned CQ ($\bx_A \neq \emptyset$)}
% (1) $\bx_H \subseteq \bx_A$, we keep the tree decompositions as all of them are free-connex by adding a fictional bag containing $\bx_H$ atop as root; for $\bx_H = \emptyset$, we recover boolean, $\bx_H = \bx_A$, we are in batch processing. In the online phase, after finishing the bottom-up semi-join sequence, we output the remaining tuples of $\bx_H$ at the root, so the answering time is $O(T + |\OUT|)$ \\ \\
% (2) $\bx_H = [n]$, we keep the tree decompositions as all of them are free-connex. In the online phase, after finishing the bottom-up semi-join sequence, we compute joins via a top-down traversal, costing time $O(|\OUT|)$. Tthe answering time overall is $O(T + |\OUT|)$ \\ \\
% (3) $\bx_A = \emptyset$, which indicates for normal CQ. In this case,
% \begin{enumerate}
%     \item [(i)] if $x_H = \emptyset$ or $x_H = [n]$, we root at any node and label subtrees accordingly; for full query, no compression, but for boolean query, some trees are compressed to an empty tree (empty root node only). $3$-cycle and $4$-cycle are good examples for this.
%     \item [(ii)] if $x_H \subset [n]$, we only look at free-connex tree decompositions and root at one of the node with only head variables. Then, we label its subtrees. The following compression compresses the subtree to only its root node, containing only join variables and head variables. In the online phase, $T \ltimes S$ is a full join. Then we degernerate to the free-connex Yannakakis algorithm.
% \end{enumerate}

 \subsection{Preprocessing Phase}
 
As a first step, we construct from $\mP$ a set of 2-phase disjunctive rules as follows. Let $\nu_i$ be the mapping for associated views of $P_i$. Let us define the cartesian product $\mathbf{A} = \times_{i \in I} \{V(\htree_i)\}$ and let $M = |\mathbf{A}|$.  Informally, every element $\mathbf{a} \in \mathbf{A}$ picks one view from every PMTD in the indexed set.  For every $\mathbf{a} \in \mathbf{A}$, we construct the following 2-phase disjunctive rule {(recall that $M_i$ is the materialization set of PMTD $(\htree_i, \chi_i, M_i, r_i)$)}:
 \begin{align*}%\label{PMTDtoDR}
      \bigvee_{\mathbf{a}_i \in M_i} S_{\nu_i(\mathbf{a}_i)}(\bx_{\nu_i(\mathbf{a}_i)}) \vee \bigvee_{\mathbf{a}_i \notin M_i}  T_{\nu_i(\mathbf{a}_i)}& (\bx_{\nu_i(\mathbf{a}_i)})  \leftarrow 
     Q_A(\bx_A) \wedge \bigwedge_{F \in \edges} R_F(\bx_F)
 \end{align*}

The body of the rule is the same independent of $\mathbf{a} \in \mathbf{A}$. The head of the rule introduces an $S$-target whenever the corresponding bag is in the materialization set of the PMTD (using the corresponding view); otherwise, it introduces a $T$-target. There are exactly $M$ 2-phase disjunctive rules constructed from the given set of PMTDs, which is a query-complexity quantity.

\begin{example}
Continuing our running example, consider the three PMTDs in Figure~\ref{fig:cfhw}. These result in four 2-phase disjunctive rules (after removing redundant $T$-targets and $S$-targets):
\begin{align*}
T_{134}(x_1, x_3, x_4) \vee S_{14}(x_1, x_4) & \leftarrow  \textsf{body}  \\
T_{134}(x_1, x_3, x_4) \vee S_{13}(x_1, x_3)  \vee S_{14}(x_1, x_4) & \leftarrow \textsf{body}  \\
T_{123}(x_1, x_2, x_3) \vee T_{134}(x_1, x_3, x_4) \vee S_{14}(x_1, x_4) & \leftarrow \textsf{body}  \\
T_{123}(x_1, x_2, x_3) \vee S_{13}(x_1, x_3) \vee S_{14}(x_1, x_4)  & \leftarrow \textsf{body} 
\end{align*}
where 
\begin{align*}
 \textsf{body} = Q_{14}(x_1, x_4)  \wedge R_{1}(x_1, x_2) \wedge R_{2}(x_2, x_3) \wedge R_{3}(x_3, x_4)
\end{align*}
%
%In other words, we take the conjunction of the heads as a logical expression, then apply the distributivity law, i.e.
% \begin{align*}
% & (T_{134} \wedge T_{123}) \vee (T_{134} \wedge S_{13}) \vee S_{14} = (T_{134} \vee S_{14}) \wedge  \\
% & (T_{134} \vee S_{13} \vee S_{14}) \wedge (T_{123} \vee T_{134} \vee S_{14}) \wedge (T_{123} \vee S_{13} 
% \vee S_{14}) 
% \end{align*}
% This provides an intuition of the fact that any tuple $\ba$ satisfying \textsf{body} (of the access CQ) gets placed into at least one of the PMTDs.
\end{example} 
 
 %Following \autoref{sec:PMTDtoRule}, we construct a set of 2-phase disjunctive rules (exactly $M$ of them), where each rule takes the form \eqref{PMTDtoDR}, i.e. bags selected from the materialized set of a PMTD are grouped as $S$-targets, whereas bags selected outside the materialized set of a PMTD are grouped as $T$-targets. 
 
 For each 2-phase disjunctive rule $\rho_k$, where $k \in [M]$, we run $\twoPPp$ with the space budget $S$. $\twoPPp$ generates the $S$-targets for $\rho_k$. Next, we compute each $S$-view of a PMTD $P_i$ by unioning all $S$-targets with the same schema as the $S$-view (possibly from outputs of different disjunctive rules).
Then, we semijoin-reduce every $S$-view with the full join $\bowtie_{F \in \edges} R_F$. This semijoin-reduce can be accomplished by tentatively storing $\bowtie_{F \in \edges} R_F$ as an intermediate truth table and remove it after the semijoin-reduce of all $S$-views. This step guarantees that any tuple in a $S$-view participates in $\bowtie_{F \in \edges} R_F$. Finally, we preprocess the $S$-views as described in Theorem~\ref{lem:Yannakakis}.
%Taking each 2-phase disjunctive rule $\rho_k$, where $k \in [M]$, we set the preprocessing space budget for storing $S$-views to be $\polyO(S)$ and find the optimal $(\blambda^*_k, \btheta^*_k)$ following \autoref{lem:maxminLP}. Using them as inputs, we run the $\twoPPp$ algorithm (the preprocessing phase of $\twoPP$). $\twoPP$ stores $S$-views and a list of indexes $\mJ_1, \ldots, \mJ_{M}$ in space $\polyO(S)$, where $\mJ_k$ denotes the index for tracking aborted $\PANDA$ instances from $\rho_k$. Note that for $S$-views having the same schema (possibly from outputs of different rules), we union them as \textit{the} $S$-view associated with that exact schema.

%At the end of the preprocessing phase, we store:
% \begin{enumerate}
%     \item [(1)] for each $\rho_k$, its necessary inputs for $\twoPPo$ (including the index $\mJ_k$, the participating Shannon-flow inequalities and proof sequences), i.e. the online phase of $\twoPP$; and
%     \item [(2)] all reduced $S$-views (having distinct schema)
% \end{enumerate}
% 
% Later on we will construct the $S$-targets (and $T$-targets) of 2-phase disjunctive rules from the $S$-views (and $T$-views) of PMTDs, so the terms views/targets are used interchangeably. 
 
 \subsection{Online Phase}
Recall that upon receiving an instance of access request $Q_A$, we need to return the results of the CQ, $\varphi(\bx_H)$. We obtain $\varphi(\bx_H)$ as follows. First, we apply $\twoPPo$ for every $\rho_k$ to get its $T$-targets (of size $\polyO(\btime_{\rho_k})$) in time $\polyO(\btime_{\rho_k} + |Q_A|)$.  Let $T_{\textsf{max}} = \max_{k \in [M]} T_{\rho_k}$.
 
 %By \autoref{thm:main}, the online phase of $\twoPP$ obtains a model of $\rho_k$ in time as given by \eqref{tradeOff:disjunctive}. In particular, by supplying the optimal $(\blambda^*_k, \btheta^*_k)$ in the preprocessing phase, $\twoPPo$ runs in time $\polyO(T_k)$, where $T_k$ coincides with the optimal objective value of $\eqref{OPT-maximin}$. The overall time required for running the online phase of $\twoPP$ on all $\rho_k$ is simply $\polyO(\btime)$, where $\btime \defeq \max_{k \in [M]} T_k$.
 
 Next, we compute each $T$-view of a PMTD $P_i$ by unioning all $T$-targets with the same schema as the $T$-view (possibly from outputs of different disjunctive rules). We semijoin-reduce every $T$-view (of size $\polyO(T_{\rho})$) with every input relation and $Q_A$. Then, for every PMTD in $\{P_i\}_{i \in I}$, we compute the free-connex acyclic CQ
 \begin{align}\label{free-connex-i}
    {\psi_{i}(\bx_H)} \leftarrow Q_A \wedge \bigwedge_{t \in M_i} S_{\nu_i(t)}  \wedge \bigwedge_{t \in V_i(\htree_i) \setminus M_i} T_{\nu_i(t)}
 \end{align}
 by applying Online Yannakakis as described in Section~\ref{Yannakakis} in time $\polyO(T_{\textsf{max}}) + O(|Q_A| + |\phi_{i} \cap \varphi|)$.
We obtain the final result by unioning the outputs across all PMTDs in our set, {$\varphi = \bigcup_{i \in I} \psi_i$}. In total, we answer the access request $Q_A$ in time $\polyO(T_{\textsf{max}}  + |Q_A|) + O(|\varphi|)$.

%% file: 2pd.tex
\section{Constructing the Tradeoffs}
\label{sec:2pd}

Let $\rho$ be a 2-phase disjunctive rule taking the form \eqref{def:partitionDisjunctiveRule}, under degree constraints $\DC$ (guarded by input relations) and degree constraints $\AC$ (guarded by $Q_A$). In this section, we will discuss how we can obtain a model of $\rho$ in two phases using $\PANDA$, and the resulting space-time tradeoff. Due to limited space, we will keep the presentation informal and introduce the key ideas through an example. The full details and proofs are deferred to Appendix~\ref{sec:panda} and~\ref{sec:flow}. We will use the following rule as our running example, where $|R_1| = |R_2| = |\mD|$:
$$ T_{123} \vee S_{13} \leftarrow Q_{13}(x_1, x_3), R_{1}(x_1, x_2), R_{2}(x_2, x_3).$$
This rule is the only rule we obtain from considering two PMTDs for the 2-reachability query.
To compute a disjunctive rule, $\PANDA$ starts with a {\em Shannon-flow inequality}, which is an inequality over set functions $h: 2^{[n]} \rightarrow \bR_{+}$ that must hold for any set function that is a polymatroid\footnote{A polymatroid is a set function $h: 2^{[n]} \rightarrow \bR_{+}$ that is non-negative, monotone, and submodular, with $h(\emptyset)=0$.}. For our purposes, we need a {\em joint Shannon-flow inequality}, which holds over two set functions $\hS, \hT$ that must be polymatroids. Intuitively, $\hS$ governs the preprocessing phase, while $\hT$ governs the online phase. The joint Shannon-flow inequality for our example is:
  \begin{align*}
    \underbrace{\hS(1) + \hT( 2 | 1)}_{R_1} + \underbrace{\hS(3) + \hT(2| 3)}_{R_2} + 2 \underbrace{\hT(1 3)}_{Q_{13}} \geq \underbrace{\hS(1 3)}_{S_{13}} + 2 \underbrace{\hT(1 2 3)}_{T_{123}} 
\end{align*}
where $h(Y|X) =h(Y)-h(X)$. The right-hand side includes terms of $\hS$ that correspond to $S$-targets and terms of $\hT$ that correspond to $T$-targets. The left-hand side includes a term of $\hT$ that corresponds to the access request $Q_A$, and possibly terms of $\hS$ ($\hT$) that encode the degree constraints $\DC$ ($\DC \cup \AC$). More importantly, it contains terms that correlate the two polymatroids by splitting an input relation with attributes $Y$ into two parts, either $(i)$ $\hS(X) + \hT(Y|X)$, or  $(ii)$ $\hT(X) + \hS(Y|X)$, where $X \subseteq Y$. Intuitively, the first split materializes the heavy $X$-values and sends everything else to the online phase, while the second split preprocesses the light $X$-values and sends the heavy $X$-values to the online phase. In our example, relation $R_1$ is split into $\hS(1) + \hT( 2 | 1)$, and each part is sent to a different polymatroid. 
Using the coefficients of the above joint Shannon-flow inequality, we get the following intrinsic space-time tradeoff:
$$ S \cdot T^2 \cong |Q_{13}|^2 \cdot |\mD|^2$$
We will use the $\cong$ notation to mean that $S \cdot T^2 = \polyO(|Q_{13}|^2 \cdot |\mD|^2)$. Generally, we show (for a formal definition, see~\autoref{thm:main}):

\begin{theorem}[Informal] \label{thm:informal}
Every joint Shannon-flow inequality for a 2-phase disjunctive rule implies a space-time tradeoff computed by reading the coefficients of the inequality. 
\end{theorem}

The above theorem requires that we are given a joint Shannon-flow inequality to obtain a space-time tradeoff. We additionally show that, given a space budget $S$, we can also compute via a linear program the optimal inequality that will result in the best possible answering time.

\introparagraph{The  $\twoPP$ algorithm} We now present how our main algorithm works (see Appendix~\ref{sec:flow} for a detailed description). For the running example, we take $|Q_{13}| = 1$, and $S$ is a fixed space budget.

As a first step, $\twoPP$ scans the joint Shannon-flow inequality and partitions $R_{1}(x_1, x_2)$ (on $x_1$) into $R^H_{1}$ and $R^L_{1}$, where $R^H_1$ contains all $(x_1, x_2)$ tuples where $|\sigma_{x_1 = t}(R_{12})| \geq |\mD|/\sqrt{S}$, and $R^L_1$ contains the tuples that satisfy $\deg_{12}(x_2|x_1) \leq |\mD|/\sqrt{S}$. $R_2$ is partitioned symmetrically (on $x_3$) into $R^H_{2}$ and $R^L_{2}$.
%Similarly, $\twoPP$ partitions $R_{23}(x_2, x_3)$ into $R'_{23}(x_3)$ and $R'_{23}(x_2, x_3)$ where $R'_3(x_3)$ contains all values of $x_1$ such that $|\sigma_{x_3 = t}(R_{23})| \geq |\mD|/ \sqrt{S}$, and $R'_{23}(x_2, x_3)$ contains all pairs $(x_2, x_3)$ where $\deg_{23}(x_2|x_3) \leq |\mD|/\sqrt{S}$.
This creates four subproblems, $\{R^H_1, R^H_2\}$, $\{R^H_1, R^L_2\}$, $\{R^L_1, R^H_2\}$ and $\{R^L_1, R^L_2\}$. In general, these splits will be done according to the correlated terms in the joint flow.

The preprocessing phase ($\twoPPp$) is governed by the Shannon-flow inequality for $\hS$, which is $\hS(1)+\hS(3) \geq \hS(13)$. We now follow $\PANDA$ and construct a {\em proof sequence} for this inequality. A proof sequence proves the inequality via a sequence of smaller steps, such that each step can be interpreted as a relational operator. The proof sequence for our case is:
      \begin{align*}
     \textcolor{black}{\hS(1)} + \textcolor{black}{\hS(3)} & \geq \textcolor{black}{\hS(13|3)} + \textcolor{black}{\hS(3)} && (submodularity) \\
     & = \textcolor{black}{\hS(13)} && (composition)
    \end{align*}
In this case, $\PANDA$ attempts to join the two relations in each subproblem. However, we allow this to happen only if the resulting space is at most $S$. Because $R^H_1$ and $R^H_2$ have size at most $|\mD| / (|\mD|/\sqrt{S}) = \sqrt{S}$ values for $x_1, x_3$ respectively, the subproblem $\{R^H_1, R^H_2\}$ can be stored in $S_{13}$ in space at most $\sqrt{S} \cdot \sqrt{S} = S$.

The online phase ($\twoPPo$) takes an access request $Q_{13}(x_1, x_3)$ that contains one tuple. Now, $\twoPPo$ follows the second  proof sequence for the polymatroid $\hT$:
    \begin{align*}
      \textcolor{black}{\hT(2 | 1)} + \textcolor{black}{\hT(2| 3)}+ 2 \textcolor{black}{\hT(1 3)} & \geq  2\textcolor{black}{\hT(2|13)} + 2 \textcolor{black}{\hT(1 3)}   && (submod.)\\
      &  = 2 \textcolor{black}{\hT(123)}  && (comp.) 
 \end{align*}
    For the other $3$ subproblems, $\twoPPo$ computes 
    the following $3$ joins: $Q_{13}(x_1, x_3)  \bowtie R^L_2(x_2, x_3)$, $Q_{13}(x_1, x_3) \bowtie R^L_1(x_1, x_2)$ and $Q_{13}(x_1, x_3) \bowtie R^L_1(x_1, x_2)$. In the submodularity step, $\twoPPo$ identifies that for the first join, $\deg_{23}(x_2|x_3) \leq |\mD|/\sqrt{S}$, so this join takes time $|Q_{13}| \cdot |\mD|/\sqrt{S} \leq |\mD|/\sqrt{S}$; and since $\deg_{12}(x_2|x_1) \leq |\mD|/\sqrt{S}$, the last two {identical} joins take time $|Q_{13}| \cdot \deg_{12}(x_2|x_1) \leq  |\mD|/\sqrt{S}$. Therefore, the overall online computing time is $|\mD|/\sqrt{S}$.

\begin{figure}[t]
	    \begin{subfigure}{0.4\linewidth}
				\vspace{2em}
					\centering
				\scalebox{1}{\begin{tikzpicture}
				\tikzset{edge/.style = {->,> = latex'},
					vertex/.style={circle, thick, minimum size=5mm}}
				\def\x{0.25}
				
				\begin{scope}[fill opacity=1]
				
				%\draw[] (0,-2) ellipse (0.8cm and 0.33cm) node {\small ${x_1, x_4}$};
				%\node[vertex]  at (1.1,-2) {$r$};
				\draw[] (0,-3.5) ellipse (0.8cm and 0.33cm) node {\small \small ${x_1, x_3, x_4}$};
				\node[vertex]  at (1.2,-3.5) {$T_{134}$};				
				\draw[] (0,-5) ellipse (0.8cm and 0.33cm) node {\small \small ${x_1, x_2, x_3}$};						
				\node[vertex]  at (1.2,-5) {$T_{123}$};								
				%\draw[edge] (0,-2.33) -- (0,-3.2);
				\draw[edge] (0,-3.83) -- (0,-4.65);			
				\end{scope}	
				\end{tikzpicture} 
				}
		\end{subfigure}
		\begin{subfigure}{0.4\linewidth}
				\vspace{2em}
					\centering
								\scalebox{1}{\begin{tikzpicture}
				\tikzset{edge/.style = {->,> = latex'},
					vertex/.style={circle, thick, minimum size=5mm}}
				\def\x{0.25}
				
				\begin{scope}[fill opacity=1]
				
				%\draw[] (0,0) ellipse (0.8cm and 0.33cm) node {\small ${x_1, x_4}$};
				%\node[vertex]  at (1.1,0) {$r$};
				\draw[fill=black!10] (0,-3.5) ellipse (1cm and 0.33cm) node {\small \small ${x_1, x_2, x_3, x_4}$};
				\node[vertex]  at (1.4,-3.5) {$S_{14}$};	\node[vertex]  at (1.4,-5.1) { };	
				%\draw[edge] (0,-0.33) -- (0,-1.2);
				\end{scope}	
				\end{tikzpicture} 
				}
		\end{subfigure}
		
		\caption{Two PMTDs for the square CQAP. The materialized nodes are shaded and labeled as $S$-views.}
		\label{fig:sqaure}	
\end{figure}
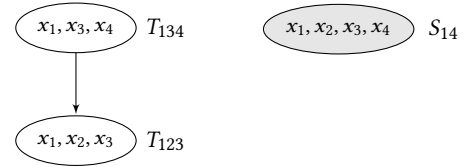

 \begin{example}[The square query]
We now give a comprehensive example of how to construct tradeoffs for the following CQAP:
 $$   \varphi(x_1, x_3 \mid x_1, x_3 ) \leftarrow R_1(x_1, x_2) \wedge R_2(x_2, x_3) \wedge R_3(x_3, x_4) \wedge R_4(x_4, x_1). $$ 
This captures the following task: given two vertices of a graph, decide whether they occur in two opposite corners of a square. We consider two PMTDs. The first PMTD has a root bag $\{1, 3, 4\}$ associated with a $T$-view $T_{134}$, and a bag $\{1, 3, 2\}$ associated with a $T$-view $T_{132}$. The second PMTD has one bag $\{1, 2, 3, 4\}$ associated with an $S$-view $S_{13}$. The two PMTDs are depicted in \autoref{fig:sqaure}. This in turn generates two disjunctive rules:
\begin{align*}
T_{134} \vee S_{13} \leftarrow  \textsf{body}, \quad \quad T_{132} \vee S_{13} \leftarrow \textsf{body}
\end{align*}
where $
 \textsf{body} = Q_{13}(x_1, x_3) \wedge R_1(x_1, x_2) \wedge R_2(x_2, x_3) \wedge R_3(x_3, x_4) \wedge R_4(x_4, x_1).
$
We can construct the following joint Shannon-flow
inequality (and its proof sequence) for the first rule:
\begin{align*}
    \underbrace{\hS(1) + {\hT( 4 | 1)}}_{R_4} & +  \underbrace{{\hS(3)} + {\hT(4| 3)}}_{R_3} + 2\cdot \underbrace{{\hT(1 3)}}_{Q_{13}} \\
    & \geq {\hS(1 3)} + {\hT(4 | 1)} + {\hT(4|3)} + 2\cdot {\hT(1 3)} \\ 
    & \geq {\hS(1 3)} + {\hT(4 | 1 3)} + {\hT(1 3)}  + {\hT(4| 1 3)} + {\hT(1 3)} \\ 
    & = \underbrace{\hS(1 3)}_{S_{13}} + 2 \cdot \underbrace{\hT(1 3 4)}_{T_{134}}.
\end{align*}
For the second rule, we symmetrically construct a proof sequence for $ 2 \log |\mD| + 2 \log |Q_{13}| \geq  {\hS(1 3)} + 2 \cdot {\hT(1 3 2)}$. Hence, reading the coefficients of the above joint
Shannon-flow inequalities, we obtain the following intrinsic space-time tradeoff $S \cdot T^2 \cong |\mD|^2 \cdot |Q_{13}|^2 $ for the given square CQAP.

%This tradeoff (when $|Q_A| = 1$) recovers the improved one obtained in Example 15 of \cite{deep2021space}.
\end{example}

%% file: results.tex
\section{Applications}
\label{sec:results}

In this section, we apply our framework to obtain state-of-the-art space-time tradeoffs for several specific problems, as well as obtain new tradeoff results. We defer the discussion for hierarchical CQAPs to the full version of the paper~\cite{full}.

\subsection{Tradeoffs for $k$-Set Intersection}

We will first study the CQAP \eqref{k-disjoint-nonboolean} that corresponds to the non-Boolean \textsf{$k$-Set Disjointness} problem (set $y = x_{k+1}$)
 $$\varphi(\bx_{[k+1]} \mid \bx_{[k]}) \leftarrow 
 \bigwedge_{i \in [k]} R(x_{k+1}, x_i)
 $$ 
 From the decomposition with a single node $t$ with $\chi(t) = [k+1]$, we construct two PMTDs, one with $M_1 = \emptyset$, another with $M_2 = \{t\}$. Thus, $\nu_1(t) = \nu_2(t) = [k+1]$. This gives rise to the following (only) two-phase disjunctive rule:
$$ T_{[k+1]} \vee S_{[k+1]} \leftarrow 
 Q_{[k]}(\bx_{[k]}) \wedge \bigwedge_{i \in [k]} R(x_{k+1}, x_i)
$$
For this rule, we have the following joint Shannon-flow inequality:
 \begin{align*}
      \hS(k, k+1) &+ \sum_{i \in [k - 1]} \{\hS(i |{k+1}) + \hT({k+1})\}  + (k-1) \cdot  {\hT([k])} \\
%     &  \geq {\hS(k, k+1)} + \sum_{i \in [k - 1]} {\hS(i | {[i+1:k + 1]} )} + (k-1) \cdot  {\hT([k+1])} \\
     & \geq {\hS({[k+1]})} + (k-1) \cdot  {\hT([k+1])}.
 \end{align*}
 By \autoref{thm:informal}, we get the tradeoff $S \cdot T^{k - 1} \cong |\mD|^{k} \cdot |Q_A|^{k - 1}$.
  
	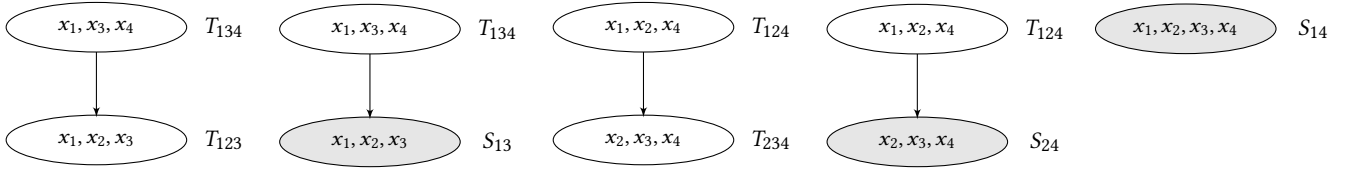
\begin{figure*}[t]
		\begin{subfigure}{0.2\linewidth}
				% \vspace{2em}
				\centering
				\scalebox{1}{\begin{tikzpicture}
				\tikzset{edge/.style = {->,> = latex'},
					vertex/.style={circle, thick, minimum size=5mm}}
				\def\x{0.25}
				\begin{scope}[fill opacity=1]
				\draw[] (0,0) ellipse (1.2cm and 0.33cm) node {\small ${x_1, x_3, x_4}$};
				\node[vertex]  at (1.7,0) {$T_{134}$};
				\draw[] (0,-1.5) ellipse (1.2cm and 0.33cm) node {\small \small ${x_1,  x_2 , x_3}$};
				\node[vertex]  at (1.7,-1.5) {$T_{123}$};				
				% \draw[] (0,-5) ellipse (1.2cm and 0.33cm) node {\small \small ${x_1, x_2, x_3}$};
				% \node[vertex]  at (1.7,-5) {$T_{123}$};								
				\draw[edge] (0,-0.33) -- (0,-1.2);
				% \draw[edge] (0,-3.83) -- (0,-4.65);			
				\end{scope}	
				\end{tikzpicture} 
				}
		\end{subfigure}
		\begin{subfigure}{0.2\linewidth}
				% \vspace{2em}
				\centering
				\scalebox{1}{\begin{tikzpicture}
				\tikzset{edge/.style = {->,> = latex'},
					vertex/.style={circle, thick, minimum size=5mm}}
				\def\x{0.25}
				\begin{scope}[fill opacity=1]
				\draw[] (0,0) ellipse (1.2cm and 0.33cm) node {\small ${x_1, x_3, x_4}$};
				\node[vertex]  at (1.7,0) {$T_{134}$};
				\draw[fill=black!10] (0,-1.5) ellipse (1.2cm and 0.33cm) node {\small \small ${x_1,  x_2 , x_3}$};
				\node[vertex]  at (1.7,-1.5) {$S_{13}$};				
				% \draw[] (0,-5) ellipse (1.2cm and 0.33cm) node {\small \small ${x_1, x_2, x_3}$};
				% \node[vertex]  at (1.7,-5) {$T_{123}$};								
				\draw[edge] (0,-0.33) -- (0,-1.2);
				% \draw[edge] (0,-3.83) -- (0,-4.65);			
				\end{scope}	
				\end{tikzpicture} 
				}
		\end{subfigure}		
		\begin{subfigure}{0.2\linewidth}
					\centering
								\scalebox{1}{\begin{tikzpicture}
				\tikzset{edge/.style = {->,> = latex'},
					vertex/.style={circle, thick, minimum size=5mm}}
				\def\x{0.25}
				\begin{scope}[fill opacity=1]
				\draw[] (0,0) ellipse (1.2cm and 0.33cm) node {\small ${x_1, x_2, x_4}$};
				\node[vertex]  at (1.7,0) {$T_{124}$};
				\draw[] (0,-1.5) ellipse (1.2cm and 0.33cm) node {\small \small ${x_2, x_3, x_4}$};
				\node[vertex]  at (1.7,-1.5) {$T_{234}$};				
				% \draw[] (5,-5) ellipse (1.2cm and 0.33cm) node {\small \small ${x_2, x_3, x_4}$};						
				% \node[vertex]  at (6.7,-5) {$T_{234}$};								
				\draw[edge] (0,-.33) -- (0,-1.2);
				% \draw[edge] (5,-3.83) -- (5,-4.65);			
				\end{scope}	
				\end{tikzpicture} 
				}
		\end{subfigure}
		\begin{subfigure}{0.2\linewidth}
					\centering
								\scalebox{1}{\begin{tikzpicture}
				\tikzset{edge/.style = {->,> = latex'},
					vertex/.style={circle, thick, minimum size=5mm}}
				\def\x{0.25}
				\begin{scope}[fill opacity=1]
				\draw[] (0,0) ellipse (1.2cm and 0.33cm) node {\small ${x_1, x_2, x_4}$};
				\node[vertex]  at (1.7,0) {$T_{124}$};
				\draw[fill=black!10] (0,-1.5) ellipse (1.2cm and 0.33cm) node {\small \small ${x_2, x_3, x_4}$};
				\node[vertex]  at (1.7,-1.5) {$S_{24}$};				
				% \draw[] (5,-5) ellipse (1.2cm and 0.33cm) node {\small \small ${x_2, x_3, x_4}$};						
				% \node[vertex]  at (6.7,-5) {$T_{234}$};								
				\draw[edge] (0,-.33) -- (0,-1.2);
				% \draw[edge] (5,-3.83) -- (5,-4.65);			
				\end{scope}	
				\end{tikzpicture} 
				}
		\end{subfigure}
		\begin{subfigure}{0.18\linewidth}
					\centering
				\scalebox{1}{\begin{tikzpicture}
				\tikzset{edge/.style = {->,> = latex'},
					vertex/.style={circle, thick, minimum size=5mm}}
				\def\x{0.25}
				\begin{scope}[fill opacity=1]
				\draw[fill=black!10] (0,0) ellipse (1.2cm and 0.33cm) node {\small ${x_1, x_2, x_3, x_4}$};
				\node[vertex]  at (1.7,0) {$S_{14}$};		
				\node[vertex]  at (1.4,-1.6) { };
%				 \draw[] (0,-1.5) ellipse (1.2cm and 0.33cm) node {\small \small ${x_1, x_2, x_3, x_4}$};						
%				 \node[vertex]  at (1.7,-1.5) {$T_{1224}$};								
				%\draw[edge] (0,-.33) -- (0,-1.2);
				% \draw[edge] (5,-3.83) -- (5,-4.65);			
				\end{scope}	
				\end{tikzpicture} 
				}
		\end{subfigure}
		\caption{The PMTDs for the 3-reachability CQAP.}
		\label{fig:example3}
	\end{figure*}  
  
\subsection{Tradeoffs via Fractional Edge Covers} \label{sec:edgecover}

Let $\varphi(\bx_A \mid \bx_A)$ be a CQAP with hypergraph $([n],\edges)$ of $\varphi$. A fractional edge cover of $S \subseteq [n]$ is an assignment $\bu = (u_F)_{F \in \edges}$ such that $(i)$ $u_F \geq 0$, and $(ii)$ for every $i \in S$, $\sum_{F: i \in F} u_F \geq 1$. For any fractional edge cover $\bu$ of $[n]$, we define the {\em slack} of $\bu$ w.r.t. $A \subseteq [n]$:
$$ \slack(\bu,A) \defeq \min_{i \notin A}  \sum_{F \in \edges: i \in F} u_F. $$
In other words, the slack is the maximum factor by which we can scale down the fractional cover $\bu$ so that it remains a valid edge cover of the variables not in $A$. Hence $(u_F/\slack(\bu,A))_{F \in \edges}$ is a fractional edge cover of $[n] \setminus A$. We always have $\slack(\bu,A) \geq 1$.

\begin{theorem}\label{thm:main1}
Let $\varphi(\bx_A \mid \bx_A)$ be a CQAP. Let $\bu$ be any fractional edge cover of the hypergraph of $\varphi$. 
Then, for any input database  $\mD$, and any access request, the following intrinsic tradeoff holds:
$$S \cdot T ^{\slack(\bu,A)} \cong |Q_A|^{\slack(\bu,A)} \cdot \prod_{F \in \edges} |R_F|^{u_F} $$
\end{theorem}

The above theorem can also be shown as a corollary of Theorem 1 in~\cite{deep2018compressed}. However, the data structure used in~\cite{deep2018compressed} is much more involved, since its goal is to also bound the delay during enumeration (while we are interested in total time instead). A simpler construction with the same tradeoff was shown in~\cite{deep2021space}. Our framework recovers the same result using a simple materialization strategy with two PMTDs.

\begin{example}
Consider $\varphi(\bx_{[k]} \mid \bx_{[k]}) \leftarrow \bigwedge_{i \in [k]} R(y, x_i)$ (corresponds to the $k$-\setdisj problem) with the fractional edge cover $\bu$, where $u_j=1$ for $j \in \{1, \dots, k\}$. The slack w.r.t. $[k]$ is $k$, since the fractional edge cover $\hat{\bu}$, where $\hat{u}_i = u_i/k = 1/k$ covers $x$.	
 Applying Theorem~\ref{thm:main1}, we obtain a tradeoff of $S \cdot T^k  \cong |Q_A|^k \cdot |\mD|^k$. When $|Q_A|=1$, this matches the best-known space-time tradeoff for the $k$-Set Disjointness problem.
\end{example}

%\introparagraph{Edge Triangle Detection}
%Consider the CQAP $ \varphi( x,z \mid x,z) \leftarrow R(x, y) \wedge R(y, z) \wedge R(x, z)$. The fractional edge cover that assigns a weight $1/2$ to each edge has slack $\slack=1$. Thus, Theorem~\ref{thm:main} gives a tradeoff of $S \cdot T \cong |Q_A| \cdot |R|^{3/2}$. We note here that the authors in~\cite{goldstein2017conditional} claim that -- conditioned on the strong set disjointness conjecture -- any data structure that achieves answering time $T$ for the above CQAP when $|Q_A|=1$ needs space $S = \Omega(|R|^2/T^2)$. Our construction is always better when $T \leq |R|^{1/2}$, thus refuting the conditional lower bound in~\cite{goldstein2017conditional}. We should note that this does not imply that the strong set disjointness conjecture is false, as we have observed an error in the reduction used in~\cite{goldstein2017conditional}.

\subsection{Tradeoffs via Tree Decompositions} \label{sec:decompostion}

Let $\varphi(\bx_A \mid \bx_A)$ be a CQAP. In the previous section, we recovered a space-time tradeoff using two trivial PMTDs. Here, we will show how our framework recovers a better space-time tradeoff by considering a larger set of PMTDs that corresponds to one decomposition.

Pick any arbitrary non-redundant free-connex decomposition $(\htree, \chi, r)$. We start by taking any set of nodes that are not ancestors of each other in the decomposition as a materialization set. Then, for each node $t$ in the materialization set, we merge all bags in the subtree of $t$ into the bag of $t$ (and truncate the subtree). By ranging over all such materialization sets, we construct a {fixed (finite)} set of PMTDs. We say that this set of PMTDs is \textit{induced} from $(\htree, \chi, r)$.
We now input the induced set of PMTDs to our general framework. To discuss the obtained space-time tradeoff, take any assignment of a fractional edge cover $\bu_t$ to each node $t \in V(\htree)$ and let $u^*_t$ be its total weight. Let $A_t$ denote the common variables between node $t$ and its parent (for the root, $A_r = A$), and define $\alpha_t = \alpha(\bu_t, A_t)$ to be the slack in node $t$ w.r.t. $A_t$. Now, take the nodes $P$ of any root-to-leaf path in $\htree$. We can show that any such path $P$ generates the following intrinsic tradeoff:
$$ S^{\sum_{t \in P} 1/\alpha_t} \cdot T \cong  |Q_A| \cdot  |\mD|^{\sum_{t \in P} u^*_t  / \alpha_t}$$
To obtain the final space-time tradeoff, we take the worst such tradeoff across all root-to-leaf paths. We show in Appendix~\ref{sec:tree-tradeoff} how to obtain this tradeoff, and also show why it recovers prior results~\cite{deep2021space}. Our framework guarantees that adding PMTDs to the induced set we considered here can only make this tradeoff better.

% We start by taking all PMTDs $(\htree, \chi, M, r)$ using the fixed decomposition $(\htree, \chi, r)$. Note that some PMTDs can be redundant, but an easy fix to make them non-redundant is to merge the nodes associated with redundant views into one node (and union the corresponding bags into one bag). Obviously, it leads to a non-redundant PMTD (with a truncated tree structure). 

\begin{example}
Consider the 4-reachability CQAP.
Here, $H = A = \{x_1, x_5\}$. We will consider the tree decomposition with bags $t_1 = \{x_1, x_2, x_4, x_5 \} \rightarrow t_2 = \{x_2, x_3, x_4\}$.

Take the edge cover $u_1=1, u_4=1$ for the bag $t_1$, and the edge cover $u_2=1, u_3=1$ for the bag $t_2$. The first bag has slack $\alpha_1 = 1$ (w.r.t. $x_1, x_5$), while the second has slack $\alpha_2 = 2$ (w.r.t. $x_2, x_4$). Here we have one root-to-leaf path, hence we get the tradeoff $S^{1+1/2} \cdot T \cong |Q_A| \cdot |\mD|^{2/1 + 2/2}$, or equivalently $S^{3/2} \cdot T \cong |Q_A| \cdot |\mD|^{3}$.
\end{example}

\eat{Note the proof sequence is valid for any disjunctive rule of the form given by~\autoref{eq:disjunctive}. The last piece is to show that it is sufficient to consider only the rules generated by~\autoref{eq:disjunctive} and that the dominating tradeoff happens for $\ell = k$. We proceed to show this inductively on the value of $\ell$. First, observe that the S-view $S_{\bx_A}(\bx_A)$ is always picked. Consider the base case $\ell = 1$. For the first PMTD, we can either pick $T_{\bx}(\bx)$ (which corresponds to root bag $t_1$) or $S(\bx_\nu(t_2))$. If we pick $T_{\bx}(\bx)$, it does not matter what we pick from the other PMTDs because the proof sequence for~\autoref{eq:disjunctive} with $\ell=1$ already generates entropy terms for $T_{\bx}(\bx)$ and $S_{\bx_A}(\bx_A)$ in the RHS. If we pick $S(\bx_\nu(t_2))$, we consider the second PMTD, where we have three choices: $T_{\bx}(\bx), T(\bx_{\nu(t_2)}), S(\bx_{\nu(t_3)})$. Picking the first term takes us back to the base case. Picking the second term corresponds to~\autoref{eq:disjunctive} with $\ell=2$ and the choice of the terms from other PMTDs is again immaterial. If the third term is picked, we move on to the third PMTD and continue the process. At the end of the process, the set of disjunctive rules generated from~\autoref{eq:disjunctive} cover all possible disjunctive rules. \shaleen{this goes away; newer argument by Paris}}

%\begin{figure}[t]
%	    \begin{tikzpicture}
%				\tikzset{edge/.style = {->,> = latex'},
%					vertex/.style={circle, thick, minimum size=5mm}}
%				\def\x{0.25}
%				
%				\begin{scope}[fill opacity=1]
%				
%				%\draw[] (0,-2) ellipse (0.8cm and 0.33cm) node {\small ${x_1, x_4}$};
%				%\node[vertex]  at (1.1,-2) {$r$};
%				\draw[] (0,-2) ellipse (1.2cm and 0.33cm) node {\small \small ${x_1, x_2, x_4, x_5}$};
%				\node[vertex]  at (1.7,-2) {$T_{1245}$};
%				\draw[] (0,-3.5) ellipse (1.2cm and 0.33cm) node {\small \small ${x_2, x_3, x_4}$};
%				\node[vertex]  at (1.7,-3.5) {$T_{234}$};									
%				%\draw[edge] (0,-2.33) -- (0,-3.2);
%				%\draw[edge] (0,-3.83) -- (0,-4.65);	
%                \draw[edge] (0,-2.33) -- (0,-3.2);	
%				\end{scope}	
%				\end{tikzpicture} 
%				
%		
%		\caption{Query decomposition for the 4-reachability query.}
%		\label{fig:qd}	
%	\end{figure}

\eat{\begin{figure}[t]
	    \begin{subfigure}{0.24\linewidth}
				\vspace{2em}
					\centering
				\scalebox{.9}{\begin{tikzpicture}
				\tikzset{edge/.style = {->,> = latex'},
					vertex/.style={circle, thick, minimum size=5mm}}
				\def\x{0.25}
				
				\begin{scope}[fill opacity=1]
				
				%\draw[] (0,-2) ellipse (0.8cm and 0.33cm) node {\small ${x_1, x_4}$};
				%\node[vertex]  at (1.1,-2) {$r$};
				\draw[] (0,-2) ellipse (1.2cm and 0.33cm) node {\small \small ${x_1, x_4, x_5}$};
				\node[vertex]  at (1.7,-2) {$T_{145}$};
				\draw[] (0,-3.5) ellipse (1.2cm and 0.33cm) node {\small \small ${x_1, x_3, x_4}$};
				\node[vertex]  at (1.7,-3.5) {$T_{134}$};				
				\draw[] (0,-5) ellipse (1.2cm and 0.33cm) node {\small \small ${x_1, x_2, x_3}$};						
				\node[vertex]  at (1.7,-5) {$T_{123}$};								
				%\draw[edge] (0,-2.33) -- (0,-3.2);
				\draw[edge] (0,-3.83) -- (0,-4.65);		
				\draw[edge] (0,-2.33) -- (0,-3.2);		
				\end{scope}	
				\end{tikzpicture} 
				}
		\end{subfigure}
		\begin{subfigure}{0.24\linewidth}
				\vspace{2em}
					\centering
				\scalebox{.9}{\begin{tikzpicture}
				\tikzset{edge/.style = {->,> = latex'},
					vertex/.style={circle, thick, minimum size=5mm}}
				\def\x{0.25}
				
				\begin{scope}[fill opacity=1]
				
				%\draw[] (0,-2) ellipse (0.8cm and 0.33cm) node {\small ${x_1, x_4}$};
				%\node[vertex]  at (1.1,-2) {$r$};
				\draw[] (0,-2) ellipse (1.2cm and 0.33cm) node {\small \small ${x_1, x_4, x_5}$};
				\node[vertex]  at (1.7,-2) {$T_{145}$};
				\draw[] (0,-3.5) ellipse (1.2cm and 0.33cm) node {\small \small ${x_1, x_3, x_4}$};
				\node[vertex]  at (1.7,-3.5) {$T_{134}$};				
				\draw[fill=black!10] (0,-5) ellipse (1.2cm and 0.33cm) node {\small \small ${x_1, x_2, x_3}$};						
				\node[vertex]  at (1.7,-5) {$S_{123}$};								
				%\draw[edge] (0,-2.33) -- (0,-3.2);
				\draw[edge] (0,-3.83) -- (0,-4.65);		
				\draw[edge] (0,-2.33) -- (0,-3.2);		
				\end{scope}	
				\end{tikzpicture} 
				}
		\end{subfigure}
		\begin{subfigure}{0.24\linewidth}
				\vspace{2em}
					\centering
			\scalebox{.9}{\begin{tikzpicture}
				\tikzset{edge/.style = {->,> = latex'},
					vertex/.style={circle, thick, minimum size=5mm}}
				\def\x{0.25}
				
				\begin{scope}[fill opacity=1]
				
				%\draw[] (5,-2) ellipse (0.8cm and 0.33cm) node {\small ${x_1, x_4}$};
				%\node[vertex]  at (6.1,-2) {$r$};
				\draw[] (5,-2) ellipse (1.2cm and 0.33cm) node {\small \small ${x_1, x_4, x_5}$};
				\node[vertex]  at (6.7,-2) {$T_{145}$};
				\draw[fill=black!10] (5,-3.5) ellipse (1.2cm and 0.33cm) node {\small \small ${x_1, x_3, x_4}$};
				\node[vertex]  at (6.7,-3.5) {$S_{134}$};				
				%\draw[] (5,-5) ellipse (1.2cm and 0.33cm) node {\small \small ${x_1, x_2, x_3}$};						
				%\node[vertex]  at (6.7,-5) {$S_{13}$};								
				%\draw[edge] (5,-2.33) -- (5,-3.2);
				%\draw[edge] (5,-3.83) -- (5,-4.65);	
				\draw[edge] (5,-2.33) -- (5,-3.2);
				\end{scope}	
				\end{tikzpicture} 
				}
		\end{subfigure}
		\begin{subfigure}{0.24\linewidth}
				\vspace{2em}
					\centering
								\scalebox{.9}{\begin{tikzpicture}
				\tikzset{edge/.style = {->,> = latex'},
					vertex/.style={circle, thick, minimum size=5mm}}
				\def\x{0.25}
				
				\begin{scope}[fill opacity=1]
				
				%\draw[] (0,0) ellipse (0.8cm and 0.33cm) node {\small ${x_1, x_4}$};
				%\node[vertex]  at (1.1,0) {$r$};
				\draw[fill=black!10] (0,-1.5) ellipse (1.2cm and 0.33cm) node {\small \small ${x_1, x_4}$};
				\node[vertex]  at (1.7,-1.5) {$S_{14}$};									
				%\draw[edge] (0,-0.33) -- (0,-1.2);
				\end{scope}	
				\end{tikzpicture} 
				}
		\end{subfigure}
		
		\caption{Four PMTDs for the 5-cycle query. The materialized nodes are shaded. Note that the tree in the chosen decomposition is a path.}
		\label{fig:qd}	
	\end{figure}}

 \def\ra{1.5}
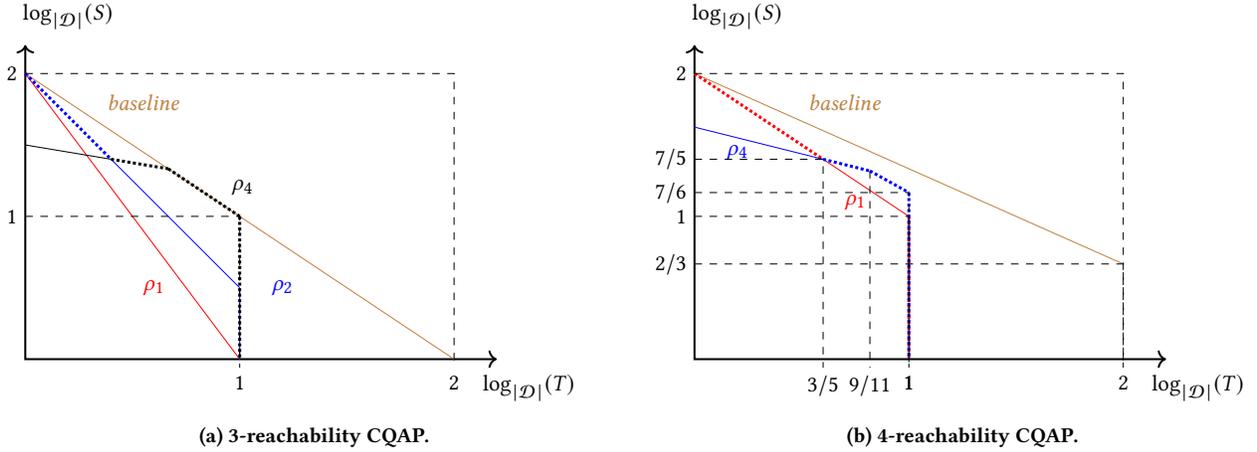
\begin{figure*}[!htp]
    \centering
    \begin{subfigure}[b]{0.48\linewidth} 
    \scalebox{0.95}{
    \begin{tikzpicture}[scale=2]
    % Draw axes
    \draw [<->,thick] (0,2.2) node (yaxis) [above] {}
        |- (2.2*\ra,0) node (xaxis) [below] {};
     \node at (0.3,2.4) {$\log_{|\mD|}(S)$};   
     \node at (2.35*\ra,-0.2) { $\log_{|\mD|}(T)$};
    % prior tradeoff    
    \draw[brown] (0,2) coordinate (A) -- (2*\ra,0) coordinate (B);
    \node at (0.55*\ra,1.8) (r4) {{\color{brown} \textit{baseline}}};
      % rho1
    \draw[red] (0,2) coordinate (b_1) -- (1*\ra,0) coordinate (b_2);
    \node at (0.6*\ra,0.5) (r1) {{\color{red} $\rho_1$}};
      % rho2/rho3
    \draw[densely dotted,very thick,blue] (A) -- (2/5*\ra,7/5) coordinate (C);
    \draw[blue] (C) -- (1*\ra,1/2) coordinate (D);
    \draw[blue] (D) -- (1*\ra,0);
    \node at (1.2*\ra,0.5) (r2) {{\color{blue} $\rho_2$}};
      % rho4
    \draw (0,1.5) coordinate (c_1) --  (C);
    \draw[densely dotted,very thick] (C) -- (2/3*\ra,4/3) coordinate (D);
    \draw[densely dotted,very thick] (D) -- (1*\ra,1);
    \draw[densely dotted,very thick] (1*\ra,1) -- (1*\ra,0);
    \node at (1*\ra,1.2) (r4) {{ $\rho_4$}};
    \coordinate (c) at (2*\ra,2);
    \draw[dashed] (yaxis |- c) node[left] {$2$}
        -| (xaxis -| c) node[below] {$2$};
      \coordinate (u) at (1*\ra,1);
     \draw[dashed] (yaxis |- u) node[left] {$1$}
        -| (xaxis -| u) node[below] {$1$};       
    %\draw[dashed] (1,2) -- (1,0) node[below] {$1$};    
    \end{tikzpicture}}
    \caption{3-reachability CQAP.}
 \label{fig:3path}
\end{subfigure}
\begin{subfigure}[b]{0.48\linewidth}
\scalebox{0.95}{
    \begin{tikzpicture}[scale=2]
    % Draw axes
    \draw [<->,thick] (0*\ra,2.2) node (yaxis) [above] {}
        |- (2.2*\ra,0) node (xaxis) [below] {};
          \node at (0.3,2.4) {$\log_{|\mD|}(S)$};      
          \node at (2.35*\ra,-0.2) { $\log_{|\mD|}(T)$};
    % prior tradeoff    
    \draw[brown] (0,2) coordinate (A) -- (2*\ra,2/3) coordinate (B);
    \node at (0.7*\ra,1.8) (r4) {{\color{brown} \textit{baseline}}};

      % rho1/rho2
    \draw[densely dotted, very thick, red] (0,2) coordinate (b_1) -- (3/5*\ra,7/5) coordinate (b_2);
    \draw[red]  (3/5*\ra,7/5) -- (1*\ra,1);
    \draw[densely dotted, very thick, red] (1*\ra, 1) -- (1*\ra, 0);
    \node at (0.75*\ra,1.1) (r1) {{\color{red} $\rho_1$}};
    
      % rho4
    \draw[blue] (0*\ra, 13/8) -- (3/5*\ra, 7/5);
    \draw[densely dotted, very thick,blue] (3/5*\ra, 7/5) -- (9/11*\ra, 29/22);
    \draw[densely dotted, very thick,blue] (9/11*\ra, 29/22) -- (1*\ra,7/6) coordinate (C);
    \draw[densely dotted, very thick, blue] (C) -- (1*\ra, 0);
    \node at (0.2*\ra, 1.45) (r2) {{\color{blue} $\rho_4$}};
    \coordinate (c) at (2*\ra,2);
    \draw[dashed] (yaxis |- c) node[left] {$2$}
        -| (xaxis -| c) node[below] {$2$};
    \coordinate (u) at (1*\ra,1);
     \draw[dashed] (yaxis |- u) node[left] {$1$}
        -| (xaxis -| u) node[below] {$1$};       
        
    \coordinate (u) at (2*\ra,2/3);
     \draw[dashed] (yaxis |- u) node[left] {$2/3$}
        -| (xaxis -| u) node[below] {};       
        
    \coordinate (u) at (3/5*\ra,7/5);
     \draw[dashed] (yaxis |- u) node[left] {$7/5$}
        -| (xaxis -| u) node[below] {$3/5$};       
        
    \coordinate (u) at (1*\ra,7/6);
     \draw[dashed] (yaxis |- u) node[left] {$7/6$}
        -| (xaxis -| u) node[below] {$1$};       
    \draw[dashed] (9/11*\ra, 29/22) -- (9/11*\ra, -0.06) node[below] {$9/11$};    
    \end{tikzpicture}}
\caption{4-reachability CQAP.}
 \label{fig:4path}
\end{subfigure}  
\caption{Space-time tradeoffs for the 3- and 4-reachability CQAP. The new tradeoffs obtained from our framework are depicted via the dotted segments. The brown lines (baseline) show the previous state-of-the-art tradeoffs.}
\end{figure*}	

\subsection{Tradeoffs for $k$-Reachability}
\label{sec:path}

In this part, we will consider the CQAP that corresponds to the {\textsf{$k$-reachability} problem described in \autoref{example:k-reachability}:
$$ \phi_k(x_1, x_{k+1} \mid x_1, x_{k+1}) \leftarrow \bigwedge_{i \in [k]} R(x_i, x_{i+1}).$$
Prior work~\cite{goldstein2017conditional} has shown the following tradeoff for a input $\mD$, which was conjectured to be asymptotically optimal for $|Q_A|=1$:
$$S \cdot T^{2/(k-1)} \cong |\mD|^2 \cdot |Q_A|^{2/(k-1)}.$$

%\begin{theorem}[\.]\label{thm:kpath}
%Let $\phi_k$ as above. Then, for any input database $D$, we can construct a data structure that answers any access request of size $\leq W$ with the following space/time tradeoff:
%\end{theorem}

We will show that the above tradeoff can be significantly improved for $k \geq 3$ by applying our framework. 

\introparagraph{3-reachability} 
%We will study the following CQAP:
%$$ \phi_3(x_1, x_{4} \mid x_1, x_{4}) \leftarrow R(x_1, x_2), R(x_2, x_3), R(x_3, x_{4}).$$
As a first step, we consider the set of all non-redundant and non-dominant PMTDs (five in total), as seen in Figure~\ref{fig:example3}.
The five PMTDs will lead to $2^4 = 16$ disjunctive rules, but we can reduce the number of rules we analyze by discarding rules with strictly more targets than other rules. For example, the disjunctive rule with head $T_{134} \vee S_{13} \vee T_{124} \vee S_{14}$ can be ignored, since there is another disjunctive rule which has a strict subset of targets, i.e., $T_{134} \vee T_{124} \vee S_{14}$. We list out the heads of the two-phase disjunctive rules we need to consider (we omit the variables for simplicity), along with the intrinsic tradeoffs for each rule in~\autoref{table:3reachability}.

\begin{table}
\caption{2-phase disjunctive rules for $3$-reachability} \label{table:3reachability}
\begin{center}
\begin{tabular}{c | c | m{3cm} } 
 \toprule
 rule & head & tradeoff \\ 
 \midrule
 $\rho_1$ & $T_{134} \vee T_{124} \vee S_{14} $  & $S \cdot T^2  \cong |\mD|^2 \cdot |Q_A|^2$  \\ 
 \midrule
 $\rho_2$ & $T_{123}  \vee S_{13} \vee T_{124} \vee S_{14}$ &  $S^2 \cdot T^3 \cong |\mD|^4 \cdot |Q_A|^3$ $T \cong |\mD|\cdot |Q_A|$  \\
 \midrule
 $\rho_3$ & $ T_{134} \vee T_{234} \vee S_{24} \vee S_{14} $  & $S^2 \cdot T^3 \cong |\mD|^4 \cdot |Q_A|^3$ $T \cong |\mD| \cdot |Q_A|$ \\
 \midrule
 $\rho_4$ & $T_{123} \vee S_{13} \vee T_{234} \vee S_{24} \vee S_{14}$ & $S \cdot T \cong |\mD|^2 \cdot |Q_A|$  $S^4 \cdot T \cong |\mD|^6 \cdot |Q_A|$ $T \cong |\mD|\cdot |Q_A|$ \\ [1ex] 
 \bottomrule
\end{tabular}
\end{center}
\end{table}
Note that rules can admit two (or more) tradeoffs that do not dominate each other; hence, we need to pick the best tradeoff depending on the regime we consider (see Appendix~\ref{sec:missing} for how we prove each tradeoff).  To understand the combined tradeoff we obtain from our analysis, we plot in Figure~\ref{fig:3path} each tradeoff curve by taking $\log_{|\mD|}$ and then taking the $x$-axis as $\log_{|\mD|}(T)$ and the $y$-axis as $\log_{|\mD|}(S)$ (fixing $|Q_A|=1$). The dotted line in the figure shows the resulting tradeoff, which is a piecewise linear function. Note that each linear segment denotes a different strategy that is optimal for that regime of space. Note that Figure~\ref{fig:3path} is not necessarily optimized for $|Q_A| > 1$. Suppose that $S = |\mD|$ and we receive $|\mD|$ single-tuple access requests in the online phase. Answering them one-by-one costs time $\polyO(|\mD|^{2})$ using the above tradeoffs. However, one better strategy is to batch the $|\mD|$ tuples (into a $4$-cycle query) and use $\PANDA$ to answer it from scratch, which costs time  $\polyO(|\mD|^{3/2})$.

\introparagraph{4-reachability} We also study the CQAP for the $4$-reachability problem. We leave the (quite complex) calculations to Appendix~\ref{sec:missing}, but we include here a plot (Figure~\ref{fig:4path}) similar to the one in Figure~\ref{fig:3path}. One surprising observation is that the new space-time tradeoff is better than the prior state-of-the-art for {\em every regime of space}. We should also point out that the tradeoff can possibly be further improved by including even more PMTDs (our analysis involved 11 PMTDs!), but the calculations were beyond the scope of this work.

\introparagraph{General reachability} Analyzing the best possible tradeoff for $k \geq 5$ becomes a very complex proposition. However, from ~\autoref{sec:decompostion} and the analysis of~\cite{deep2021space}, our framework can at least obtain the $S \cdot T^{2/(k-1)} \cong |\mD|^{2}$ tradeoff, and can likely strictly improve it.

%% file: related.tex
\section{Related Work}
\label{sec:related}

\noindent {{\bf \em Set Intersection and Distance Oracles.}} 
Space-time tradeoffs for query answering (exact and approximate) has been an active area of research across multiple communities in the last decade~\cite{kociumaka2013efficient,afshani2016data, larsen2015hardness,Cohen2010}. Cohen and Porat~\cite{Cohen2010} introduced the fast intersection problem and presented a data structure to enumerate the intersection of two sets with guarantees on the total answering time. Goldstein et. al~\cite{goldstein2017conditional} formulated the $k$-reachability problem on graphs, and showed a simple recursive data structure which achieves the $S \cdot T^{2/(k-1)} = O(|\mD|^2)$ tradeoff. They also conjectured that the tradeoff is optimal and used it to justify the optimality of an approximate distance oracle proposed by~\cite{agarwal2014space}. The study of (approximate) distance oracles over graphs was initiated by Patrascu and Roditty~\cite{patrascu2010distance}, where lower bounds are shown on the size of a distance oracle for sparse graphs based on a conjecture about the best possible data structure for a set intersection problem. Cohen and Porat~\cite{cohen2010hardness} also connected set intersection to distance oracles. Agarwal et al.~\cite{agarwal2011approximate, agarwal2014space} introduced the notion of \emph{stretch} of an oracle that controls the error allowed in the answer. Further, for stretch-2 and stretch-3 oracles, we can achieve tradeoffs $S \cdot T = O(|\mD|^2)$ and $S \cdot T^2 = O(|\mD|^2)$ respectively, and for any integer $k>0$, a stretch-$(1+1/k)$ oracle exhibits an $S \cdot T^{1/k} = O(|\mD|^2)$ tradeoff. Unfortunately, no lower bounds are known for non-constant query time. 

\smallskip
\noindent {{\bf \em Space/Delay Tradeoffs.}} A different line of work considers the problem of enumerating query results of a non-Boolean query, with the goal of minimizing the {\em delay} between consecutive tuples of the output. In constant delay enumeration~\cite{Segoufin13,BaganDG07}, the goal is to achieve constant delay for a CQ after a preprocessing step of linear time (and space); however, only a subset of CQs can achieve such a tradeoff. Factorized databases~\cite{OlteanuZ15} achieve constant delay enumeration after a more expensive super-linear preprocessing step for any CQ. If we want to reduce preprocessing time further, it is necessary to increase the delay. Kara et. al~\cite{kara19} presented a tradeoff between preprocessing time and delay for enumerating the results of any hierarchical CQ under static (and dynamic) settings. {Deng et.al \cite{DBLP:conf/icdt/DengL023} initiates the study of the space-query tradeoffs for range subgraph counting and range subgraph listing problems}. The problem of CQs with access patterns was first introduced by Deep and Koutris~\cite{deep2018compressed} (under the restriction $|Q_A|=1$), but the authors only consider full CQAPs. Previous work~\cite{X17} considered the problem of constructing space-efficient views of graphs to perform graph analytics, but did not offer any theoretical guarantees. More recently, Kara et. al~\cite{CQAP} studied tradeoffs between preprocessing time, delay, and update time for CQAPs. They characterized the class of CQAPs that admit linear preprocessing time, constant delay enumeration, and constant update time.  
All of the above results concern the tradeoff between space (or preprocessing time) and delay, while our work focuses on the total time to answer the query. Our work is most closely related to the non-peer-reviewed work in~\cite{deep2021space}. There, the authors also study the problem of building tradeoffs for Boolean CQs. The authors propose two results that slightly improve upon~\cite{deep2018compressed}. They were also the first to recognize that the $k$-reachability tradeoff is not optimal by proposing a small improvement for $k \geq 3$. The results in our work are a vast generalization that is achieved using a more comprehensive framework. For the dynamic setting,~\cite{berkholz2017answering} initiated the study of answering CQs under updates. Recently,~\cite{kara2019counting} presented an algorithm for counting the number of triangles under updates.~\cite{CQAP} proposed dynamic algorithms for CQAPs and provided  a syntactic characterization of queries that admit constant time per single-tuple update and whose output tuples can be enumerated with constant delay.

\smallskip
\noindent {{\bf \em CQ Evaluation.}} Our proposed framework is based on recent advances in efficient CQ evaluation, and in particular the \textsf{PANDA} algorithm~\cite{DBLP:conf/pods/Khamis0S17}. This powerful algorithmic result follows a long line of work on query decompositions~\cite{gottlob2014treewidth,Marx13,Marx10}, worst-case optimal algorithms~\cite{NgoRR13}, and connections between CQ evaluation and information theory~\cite{KhamisNR16,KhamisK0S20}. 

%% file: conclusion.tex
\section{Conclusion}
\label{sec:conclusion}

In this paper, we present a framework for computing general space-time tradeoffs for answering CQs with access patterns. We show the versatility of our framework by demonstrating how it can capture state-of-the-art tradeoffs for problems that have been studied separately. The application of our framework has also uncovered previously unknown tradeoffs. Many open problems remain, among which are obtaining (conditional) lower bounds that match our upper bounds, and investigating how to make our approach practical. 

Many open problems remain that merit further work. In particular, there are no known lower bounds to prove the optimality of the space-time tradeoffs. The optimality of existing space-time tradeoffs for approximate distance oracles is also now an open problem again and we believe our proposed framework should be able to improve the upper bounds. \eat{We would also like to extend the framework to incorporate general semirings to support operations such as counting etc. using the proposed variants of \textsf{PANDA} in the existing literature.} It would also be very interesting to see the applicability of this framework in practice. In particular, our framework can be extended to also include views that have been precomputed, which is a common setting. In this regard, challenges remain to optimize the constants in the time complexity to ensure implementation feasibility.

%% file: appendix.tex
\section{Missing Details from Section~\ref{sec:pmtd}}
\label{sec:pmtd:appendix}

We will adapt the Yannakakis algorithm for a free-connex tree decomposition that works into two passes: the bottom-up semijoin-reduce pass and the top-down join pass.
 The algorithm first groups all edges by whether the edge connects two $S$-views, two $T$-views or one $S$-view and one $T$-view. We name them the $SS$-edges, $TT$-edges and $ST$-edges, respectively. One key observation of a PMTD is that the bottom-up order of edges on each branch is always: first some $SS$-edges, then at most one $ST$-edges followed by some $TT$-edges. 
 
 To preprocess the $S$-views, we first run a bottom-up semijoin pass on the $SS$-edges. Then, for each $S$-view, we create a hash index with search key the common variables with its (unique) parent. We now illustrate the two passes of Online Yannakakis. This allows the semijoin of a parent with a child that is an $S$-view to be done in time linear to the size of the parent view.
 
% The Old Pre-processing Phase Yannakakis, replaced by the Adapted Yannakakis
%  We apply a bottom-up semijoin-reduce pass and stop going upwards for a branch if we hit a $ST$-edge on that branch. More precisely, for each branch, if we visit a $SS$-edge, we semijoin-reduce the $S$-view of the parent with the $S$-view of the child, else if we hit a $ST$-edge, we stop for that branch. 
 
 \introparagraph{Bottom-up Semijoin-Reduce Pass}
  We first apply a bottom-up semijoin-reduce pass to remove all non-head variables in the tree by semijoins and projections. There are three scenarios as we go upwards, depending on the type of the edge we visit.
  Let $(t, p) \in E(\htree)$ be the edge we are visiting, where $t$ is the child node and $p$ is the parent of $t$. We distinguish the following cases:
 \begin{enumerate}
	\item $(t, p)$ is an $SS$-edge: we skip the edge (recall we have handled this edge during the bottom-up semijoin-reduce pass in $M$). 
    \item $(t, p)$ is an $ST$-edge: we update the view $T_{\nu(p)}  \leftarrow T_{\nu(p)} \ltimes S_{\nu(t)}$. If every head variable in $\nu(t)$ is also in $\nu(p)$, we remove $S_{\nu(t)}$ from the tree.
    \item $(t, p)$ is a $TT$-edge: we update $T_{\nu(p)}  \leftarrow  T_{\nu(p)} \ltimes T_{\nu(t)}$. If every head variable in $\nu(t)$ is also in $\nu(p)$, we remove $T_{\nu(t)}$ from the tree; otherwise, we update $\nu(t) \leftarrow \nu(t) \cap H$ and $T_{\nu(t)} \leftarrow \Pi_{\nu(t) \cap H}(T_{\nu(t)})$.
\end{enumerate}

 At the end of the bottom-up pass, for the root node $r$, if $r \in M$, we update $Q_A \leftarrow Q_A \ltimes S_{\nu(r)}$; or if  $r \notin M$, we update $\nu(r) \leftarrow \nu(r) \cap H$, $T_{\nu(r)}  \leftarrow \Pi_{\nu(t) \cap H}(T_{\nu(r)})$ and $Q_A  \leftarrow Q_A \ltimes T_{\nu(r)}$. Now, a bottom-up semi-join reducer is accomplished on the reduced tree. We prove that this reduced tree contains only head variables in \autoref{lem:reducedTree}.

 \introparagraph{Top-down Join Pass}
 If $r \in M$, we compute $Q_A \bowtie S_{\nu(r)}$, or if $r \notin M$, we compute $Q_A \bowtie T_{\nu(r)}$. From here, we apply the exact top-down full-join pass of Yannakakis on the reduced tree (from $r$, use the parent view to probe the child view) to get the output $\psi$.

 \begin{example}
 We use the non-redundant PMTD shown in \autoref{fig:Yannakakis} of a CQAP $\varphi(x_1, x_2, x_3, x_4, x_7, x_8 \mid x_1, x_2)$, where $(x_1, x_2)$
 is the access pattern. We use the following free-connex acyclic CQ to demonstrate Online Yannakakis:
\begin{align*} 
\psi(\bx_H) \leftarrow  Q_{12}(x_1, x_2) \wedge  T_{12}(x_1, x_2) \wedge T_{13}(x_1, x_3) \wedge T_{345}(x_3, x_4, x_5) \wedge  S_{45}(x_4, x_5) \wedge S_{37}(x_3, x_7) \wedge S_{78}(x_7, x_8),
 \end{align*}
 where $\bx_H = (x_1, x_2, x_3, x_4, x_7, x_8)$. The following is the sequence of 
 semijoin-reduces in the bottom-up semijoin-reduce pass (the $SS$-edge $(S_{37}, S_{78})$ is skipped)
 \begin{align*}
    TS\text{-edge } (T_{345}, S_{45}): &   \quad T_{345}^{(1)} \leftarrow T_{345} \ltimes S_{45}, \quad \text{remove } S_{45}    \\
    TS\text{-edge } (T_{13}, S_{37}):  &   \quad T_{13}^{(1)} \leftarrow T_{13}  \ltimes S_{37}    \\
    TT\text{-edge } (T_{13}, T_{345}): &   \quad T_{13}^{(2)} \leftarrow T_{13}  \ltimes T_{345}^{(1)}, \quad T_{34}^{(1)}  \leftarrow \Pi_{34}(T_{345}^{(1)})    \\
    TT\text{-edge } (T_{12}, T_{13}):  &   \quad T_{12}^{(1)} \leftarrow T_{12} \ltimes T_{13}^{(2)} \\
    \text{root } :  &   \quad Q_A^{(1)}  \leftarrow Q_A \ltimes T_{12}^{(1)}
 \end{align*}
  In the top-down pass, to get the result of $\psi$, Online Yannakakis computes the following joins from the root to the leaves, starting from $ Q_A^{(1)}$
  \begin{align*}
      Q_A^{(1)} \bowtie T_{12}^{(1)} \bowtie T_{13}^{(2)} \bowtie T_{34}^{(1)}  \bowtie S_{37} \bowtie S_{78}.
  \end{align*}

 \end{example}
 
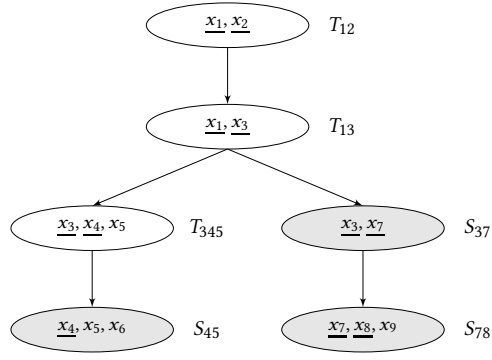
\begin{figure}[t]
			\vspace{2em}
			\centering
			\scalebox{.9}{\begin{tikzpicture}
			\tikzset{edge/.style = {->,> = latex'},
				vertex/.style={circle, thick, minimum size=5mm}}
			\def\x{0.25}
			
			\begin{scope}[fill opacity=1]
			
			\draw[] (0,-2) ellipse (1.2cm and 0.33cm) node {\small \small ${\underline{x_1}, \underline{x_2}}$};
			\node[vertex]  at (1.7,-2) {$T_{12}$};
			\draw[] (0,-3.5) ellipse (1.2cm and 0.33cm) node {\small \small ${\underline{x_1}, \underline{x_3}}$};
			\node[vertex]  at (1.7,-3.5) {$T_{13}$};				
			\draw[] (-2,-5) ellipse (1.2cm and 0.33cm) node {\small \small ${\underline{x_3}, \underline{x_4}, x_5}$};		
			\node[vertex]  at  ( -0.3,-5) {$T_{345}$};	
			
			\draw[fill=black!10] ( 2,-5) ellipse (1.2cm and 0.33cm) node {\small \small ${\underline{x_3}, \underline{x_7}}$};		
			\node[vertex]  at  ( 3.7,-5) {$S_{37}$};	
            
            \draw[fill=black!10] ( -2,-6.5) ellipse (1.2cm and 0.33cm) node {\small \small ${\underline{x_4}, x_5, x_6}$};	\node[vertex]  at  ( -0.3, -6.5) {$S_{45}$};	
            
            \draw[fill=black!10] ( 2,-6.5) ellipse (1.2cm and 0.33cm) node {\small \small ${\underline{x_7}, \underline{x_8}, x_9}$};	\node[vertex]  at  ( 3.7, -6.5) {$S_{78}$};	
            
			\draw[edge] (0,-2.33) -- (0,-3.17);	
			\draw[edge] (0,-3.83) -- (-2,-4.67);	
			\draw[edge] (0,-3.83) -- ( 2,-4.67);
			\draw[edge] (-2, -5.33) -- ( -2, -6.17);
			
			\draw[edge] ( 2, -5.33) -- ( 2, -6.17);
			
			\end{scope}	
			\end{tikzpicture} 
			}
		\caption{A non-redundant PMTD as an example for Online Yannakakis. The materialization set is shaded (and marked as S-views) and the head variables are underlined.}
		\label{fig:Yannakakis}	
	\end{figure}

\begin{lemma}\label{lem:reducedTree}
  The reduced tree after the bottom-up semijoin-reduce pass of the Online Yannakakis contains only head variables, i.e. $\bx_H$. 
 \end{lemma} 
 \begin{proof}
  Obviously, $T$-views in the reduced tree contain only head variables. Therefore, we only need to show the property for every $S$-view. This is obvious for a root $S$-view or a non-root $S$-view that has a parent $S$-view by definition of $\nu(\cdot)$. We are left to show for a $ST$-edge $(t, p)$ such that there is a head vertex $y \in \nu(t) \setminus \nu(p)$, which indicates that $t$ is the top-most bag to contain $y$. Suppose the $S$-view $S_{\nu(t)}$ contains some $z \notin H$, then $z$ must be in $\nu(p) = \chi(p)$ by definition of $\nu(\cdot)$. This contradicts the free-connex property since $\textsf{TOP}_r(z)$ is an ancestor of $\textsf{TOP}_r(y)$, i.e. a non-head vertex $z \in \nu(p)$ is above the head vertex $y \in \nu(t)$.
 \end{proof}
 
 \begin{proof}[Proof of Theorem~\ref{lem:Yannakakis}]
Define $T =  \max_{t \in V(\htree) \setminus M} |T_{\nu(t)}|$. The bottom-up pass costs time $O(T + |Q_A|)$ as only $T$-views and $Q_A$ are semijoin-reduced ($S$-views, by the index construction, are also bottom-up semijoin-reduced). Moreover, by \autoref{lem:reducedTree}, the reduced tree after the bottom-up pass contains only variables in $H$. So, joining top-down from the root circumvents any intermediate dangling tuples and costs time $O(|\psi|)$. The overall time cost is $O(T+|Q_A| +  |\psi|)$. 
 \end{proof}
 
\section{Missing Details from Section~\ref{sec:framework}}

  We show that the general framework, for any access request $Q_A$, returns the correct results $\varphi$, i.e. $\varphi = \bigcup_{i \in I} \psi_i$. As a benefit from semijoin-reduces of every view in a PMTD, it is obvious that $\bigcup_{i \in I} \psi_i \subseteq \varphi$. For the opposite inclusion, we prove the following two  claims, following the proof of Corollary 7.13 in \cite{DBLP:conf/pods/Khamis0S17}.
   
\introparagraph{Claim 1} Let us pick out one target (either a $S$-target or a $T$-target), denote the target as $U_k$, from the head of every rule $\rho_k$, where $k \in [M]$. We call the tuple $(U_k)_{k \in [M]}$ that consists of one picked target per rule a \textit{full-choice}. Then, for any \textit{full choice} $(U_k)_{k \in [M]}$, there is a PMTD $P \in \{P_i \}_{i \in I}$ such that for every tree node $t \in M$, the $S$-target $S_{\nu(t)}$ is in the full-choice and for every tree node $t \notin M$, the $T$-target $T_{\nu(t)}$ is in the full-choice. Breaking ties arbitrarily, we call this PMTD \textit{the} PMTD associated with the full choice $(U_k)_{k \in [M]}$.
    
     \begin{proof}[Proof of Claim 1]
  Fix a full-choice $(U_k)_{k \in [M]}$. Suppose to the contrary that for every PMTD $P_i \in \mP$, there is a tree node such that its corresponding target (either a $T$-target or a $S$-target), denoted as $U_i^*$, is not in the full-choice. Then we examine the 2-phase disjunctive rule $\rho^*$ that takes $\bigvee_{i \in I} U_i^* $ as its head. By definition of a full-choice, one target must be picked from $\rho^*$. However, this is a contradiction because every head $U_i^*$ in $\rho^*$ does not show up in the fixed full-choice $(U_k)_{k \in [M]}$.
 \end{proof}
 
  \introparagraph{Claim 2} For any full-choice $\boldsymbol{U} \defeq (U_k)_{k \in [M]}$ with its associated PMTD $(\htree_i, \chi_i, M_i, r_i)$, we define a CQ
 $$ \varphi_{\boldsymbol{U}}(\bx_H) \leftarrow Q_A \wedge \bigwedge_{k \in [M]} U_k
 $$
 Let $\mathcal{U}$ denote the set of all full-choices. Then, 
 \begin{align}\label{claim2}
     \varphi \subseteq \bigcup_{\boldsymbol{U}  \in \mathcal{U}}   \varphi_{\boldsymbol{U}}  \subseteq \bigcup_{i \in I} \psi_{i}
 \end{align}
 
  \begin{proof}[Proof of Claim 2]
  Take any output tuple $\ba_{H} \in \varphi$. Then, there is some tuple $\ba$ satisfying the body of $\varphi$ such that $\ba_{H} = \Pi_H(\ba)$. For each rule $\rho_k$, where $k \in [M]$, let $U_k^*$ be the target (either $S$-target or $T$-target) associated with the view ${\nu}_k^*$ such that $\Pi_{{\nu}_k^*}(\ba) \in U_k^*$. Therefore, $\boldsymbol{U}^* = (U_k^*)_{k \in M}$ is a full-choice and
  $\ba$ satisfies the body of $\varphi_{\boldsymbol{U}^*}$. Hence, $\ba_H \in \varphi_{\boldsymbol{U}^*}$ and we have shown the first inclusion in \eqref{claim2}. The second inclusion follows by dropping the atoms in the body of $\varphi_{\boldsymbol{U}^*}$ that is not any view of $\boldsymbol{U}^*$'s associated PMTD.
%   as $\phi_i$ simply drops some atoms in the body of $\varphi_{\boldsymbol{U}^*}$.
 \end{proof}

%%%%%%%%%% the 2PP SECTIONS are all here, SEC C & D
 \input{disjunctiveRules}

\newpage
\section{Missing Details from Section~\ref{sec:results}}
\label{sec:missing}

\subsection{Tradeoffs via Fractional Edge Cover (Section~\ref{sec:edgecover})}
The following lemma is a generalization of Shearer's lemma (Lemma D.1 in \cite{NgoRR13}).

\begin{lemma}\label{lem:conditional:Shearer}
Let $\mH = ([n],\edges)$ be a hypergraph and $\hat{\bu}$ be a fractional edge cover of $[n] \setminus A \subseteq [n]$.  Then, 
$$ \sum_{F \in \edges} \hat{u}_F  \cdot  h({F} \mid {A \cap F})  +  h(A)  \geq h([n])$$
\end{lemma}

\begin{proof}
Assume w.l.o.g. that $A = \{1, \dots, \ell-1\}$. Then we can write:
\begin{align*}
h([n]) & = \sum_{j=\ell}^n h(j \mid i: i<j) + h(A) \\
& \leq \sum_{j=\ell}^n \sum_{F \in \edges: j \in F} \hat{u}_F \cdot h(j \mid i: i<j, i \in F) + h(A) \\
& = \sum_{F \in \edges} \hat{u}_F \sum_{j \in F \setminus A} h(j \mid i: i<j, i \in F)  + h(A) \\
& = \sum_{F \in \edges} \hat{u}_F \left(h(F) - \sum_{j \in F, j < \ell} h(j \mid i: i<j, i \in F) \right) + h(A) \\
& = \sum_{F \in \edges} \hat{u}_F \left(h(F) - \sum_{j \in F \cap A} h(j \mid i: i<j, i \in F \cap A) \right) + h(A) \\ 
& = \sum_{F \in \edges} \hat{u}_F \left(h(F) - h({A \cap F}) \right) + h(A) \\ 
& = \sum_{F \in \edges} \hat{u}_F \cdot h(F \mid {A \cap F})  + h(A) \\ 
\end{align*}
This completes the proof.
\end{proof}

\begin{proof}[Proof of Theorem~\ref{thm:main1}]
To obtain the desired tradeoff, we consider two PMTDs. The first PMTD $P_1$ has one node $t$ with $\chi_1(t) = [n]$ and $M_1 = \emptyset$, while the second PMTD $P_2$ has also one bag $t$ with $\chi_2(t) = [n]$ and $M_2 = \{t\}$. $P_1$ contains the $T$-view $T_{[n]}(\bx_{[n]})$, while $P_2$ contains the $S$-view $S_A(\bx_A)$. These two PMTDs correspond to the following materialization policy: either store directly the answer of an access request, or compute the access request from scratch.
Hence, we only need to consider one disjunctive rule (we use $\bx$ to denote the tuple $\bx_{12\dots n}$):
$$ T_{[n]}(\bx) \vee S_{A}(\bx_A) \leftarrow Q_A(\bx_A) \wedge \bigwedge_{F \in \edges} R_F(\bx_F).  $$

Define $\alpha = \alpha(\bu,A)$ and $\hat{\bu} = \bu /\alpha$. We can now write the following proof:
\begin{align*}
        \sum_{F \in \edges} u_F \cdot \log N_{F|\emptyset} + \slack \cdot {\log |Q_A|}
    & \geq  \sum_{F \in \edges} u_F  \cdot \{ {\hT(F \mid {A \cap F})} + {\hS({A \cap F})} \} + \slack \cdot {\hT(A)}  \\
    & =  \sum_{F \in \edges} u_F  \cdot {\hT(F \mid {A \cap F})}  + \slack \cdot {\hT(A)}  + \sum_{F \in \edges} u_F \cdot {\hS({A \cap F})} && \\
    & \geq \sum_{F \in \edges} u_F  \cdot {\hT(F \mid {A \cap F})}  + \slack \cdot {\hT(A)}  +  {\hS({A})} &&  \textit{(Shearer's Lemma)} \\
        & = \slack \sum_{F \in \edges} \hat{u}_F  \cdot {\hT({F} \mid {A \cap F})}  + \slack \cdot {\hT(A)}  +  {\hS({A})} &&  \textit{(\autoref{lem:conditional:Shearer})} \\
  %     & = \slack \cdot {\hT(\bx \mid {A})}  + \slack \cdot {h(A)}  +  {h({A})} && \textit{Shearer's Lemma}  \\
       & \geq \slack \cdot {\hT([n])}  +  {\hS({A})} \\
\end{align*}
The second inequality is a direct application of Shearer's Lemma \highlight{on the sub-hypergraph $(A, \{A \cap F \mid F \in \edges\})$ of $H$, since $\bu$ is a fractional edge cover of $A$.} The last inequality is a direct consequence of Lemma~\ref{lem:conditional:Shearer}.
By Theorem~\ref{thm:main}, we obtain the desired tradeoff.
\end{proof}

\subsection{Tradeoffs via Tree Decompositions (Section~\ref{sec:decompostion})}
\label{sec:tree-tradeoff}
Let $\varphi(\bx_A \mid \bx_A)$ be a CQAP. Following Section~\ref{sec:decompostion}, let $\mP$ be the set of all 2-phase disjunctive rules generated by the induced set of PMTDs. We start our analysis by showing that for any disjunctive rule $\rho_a$ in this set, there is another disjunctive rule $\rho_b$ that is no easier than $\rho_a$ in terms of its intrinsic tradeoff. 
%Further, we will prove that the rule $\rho_b$ has a specific structure where the tree $\tree'$ of the tree decomposition used to generate $\rho_b$ consists of all bags of $(\tree, \chi)$ for some root-to-leaf path. This will allow us to use the proof sequence for~\autoref{eq:disjunctive} for some appropriately set value of $\ell$. 
Interestingly, despite choosing a (possibly) harder rule, we are still able to recover many state-of-the-art tradeoffs. We begin by stating two key observations.

\begin{observation} \label{obv:rule}
    For any 2-phase disjunctive rules $\rho_a$ and $\rho_b$, $\rho_a$ is said to be no harder than $\rho_b$ (or equivalently, $\rho_b$ is no easier than $\rho_a$) if 
    the $S$-targets of $\rho_b$ are a subset of the $S$-targets of $\rho_a$ and the $T$-targets of $\rho_b$ are a subset of the $T$-targets of $\rho_a$.
\end{observation} 

In other words,~\autoref{obv:rule} states that adding more targets to the head of a disjunctive rule can only make its model evaluation easier since we can always ignore the new targets. The next lemma makes use of the structure of the $T$-views in a PMTD.

\begin{lemma} \label{lem:rule}
    Let $\varphi(\bx_A \mid \bx_A)$ be a given CQAP. Let $(\tree, \chi, r)$ be a fixed free-connex tree decomposition. Let $\mP$ be the set of PMTDs induced from $(\tree, \chi, r)$. Let $\rho_a$ be a 2-phase disjunctive rule generated from $\mP$. Then, there is a 2-phase disjunctive rule $\rho_b$ that is no easier than $\rho_a$ such that for any two $T$-targets of $\rho_b$, their corresponding nodes in $\tree$ do not lie in some root-to-leaf path.
\end{lemma} 
\begin{proof}
    Let $\sfBT(\rho)$ be a set such that for every $B \in \sfBT(\rho)$, there is a $T$-targets $T_B$ picked by the 2-phase disjunctive rule $\rho$ (similarly, we define $\sfBS(\rho)$ for $S$-views). Note that by definition of views, a $T$-view for a node implies that its ancestors are all associated with $T$-views.
    
    We construct the 2-phase disjunctive rule $\rho_b$ in the following way: for each PMTD, if $\rho_a$ picks a $S$-view, $\rho_b$ follows $\rho_a$'s pick; if $\rho_a$ picks a $T$-view, $\rho_b$ looks at the path from root to this $T$-view (picked by $\rho_a$) and pick the first node $t$ such that $\nu(t) \in \sfBT(\rho_a)$, i.e. the top-most $T$-view $\rho_a$ picked on this branch. It is easy to see that by this construction, it holds that no two bags corresponding to two different $T$-targets of $\rho_b$ lie on a root-to-leaf path in $\tree$. Furthermore, $\sfBT(\rho_b) \subseteq \sfBT(\rho_a)$ and $\sfBS(\rho_b) = \sfBS(\rho_a)$. Thus, rule $\rho_b$ is no easier than rule $\rho_a$.
\eat{For any rule $\rho_a$ that does not satisfy the property already, we find a subset of its $T$-targets as follows. Let $\sfBT(\rho_a) = \{T_1, \dots, T_j \}$. Then, for each root-to-leaf path in $\tree$, we pick the first node $\chi^{\inv}(t)$ for some $t \in \sfBT(\rho_a)$ that appears in the traversal of the path starting from the root and add $t$ to the set $\sfBT(\rho_b)$. We set $\sfBS(\rho_b) = \sfBS(\rho_a)$.}
\end{proof}

% A different way to state the property in the statement of~\autoref{lem:rule} is that the $T$-targets of the rule $\rho_a$ are chosen by storing the \emph{fringe} of the tree when performing the breadth-first order traversal, starting from the root, and marking the entire subtree of a vertex as visited as soon as we find an unvisited vertex corresponding to a $T$-target of $\rho_a$. 

We are now ready to show the main property for any tree decomposition.

% \hangdong{we should note that the argument here only works for Boolean CQAP}

% \hangdong{some of the notations in A.4 should be polished if we have time}

\begin{lemma} \label{lem:tree}
      Let $\varphi(\bx_A \mid \bx_A)$ be a given CQAP. Let $(\tree, \chi, r)$ be a fixed free-connex tree decomposition. Let $\mP$ be the set of PMTDs induced from $(\tree, \chi, r)$. Any $2$-phase disjunctive rule $\rho_a$ generated from 
     $\mP$ is no easier than a 2-phase disjunctive rule of the form
     \begin{align} \label{eq:rule}
        T(\bx_{\chi(t_{\ell})}) \vee  \bigvee_{j \in [\ell]} S(\bx_{A_j}) \leftarrow Q_A(\bx_{A}) \wedge \bigwedge_{F \in \edges} R_F(\bx_F),
    \end{align}
    where $t_1 = r, t_2, \dots, t_{\ell}$ are the nodes of the tree $\tree$ that form a path starting from the root node $r$; and
     \begin{align*}
      A_j & \defeq \begin{cases} 
      A  & \text { if } j = 1  \\
      \chi(t_j) \cap \chi(t_{j-1}) & \text { if } j = 2,  \dots, \ell \\
       \end{cases}
 \end{align*}
    %  For any disjunctive rule $\rho_a$, there is a disjunctive rule $\rho_b$ that is no easier than $\rho_a$, such that $\rho_b$ is of the form
\end{lemma}
\begin{proof}
    We start by invoking~\autoref{lem:rule} with $\rho_a$ and $(\tree, \chi, r)$ as input to get the rule $\rho_b$ that satisfies that no $T$-target of $\rho_b$ is an ancestor of another. We fix the $T$-targets of $\rho_b$, i.e. fix $\sfBT(\rho_b) = \{B_1, \ldots, B_k\}$. Our goal is to show that $\rho_b$ must contain a subset of $S$-targets whose corresponding nodes in $\tree$ form a path starting from the root and ending at some node $t$ such that $\chi(t) \in \sfBT(\rho_b)$. In the following, we will use the function $\chi^{\inv}(B)$ to recover the node $t \in V(\tree)$ such that $\chi(t) = B$ and use $parent(t)$ to denote the parent node of a non-root node $t$.

    We prove that the property holds by allowing an adversary to pick targets in PMTDs, while we adaptively choose the PMTDs that the adversary must pick from. We can choose the ordering of the PMTDs because to construct a 2-phase disjunctive rule, one view must be picked from every PMTD in $\mP$ and we are only controlling the order in which they are examined. 

    Consider the PMTD $P_1$ where the set for $S$-targets is exactly $M_{1} = \{\chi^{\inv}(B_1), \dots, \chi^{\inv}(B_k)\}$. We offer the adversary to pick a target from $P_1$. We claim that the adversary must pick one $S$-view corresponding to a node in $M_{1}$. Indeed, the adversary cannot pick a $T$-view from $P_1$ since the $T$-targets of $\rho_b$ have already been fixed and cannot be changed. Suppose the adversary picks $S$-view associated with node $\chi^{\inv}(B_i)$. For this $S$-view (since its parent is associated with a $T$-view), 
    \begin{align}\label{Aj_prop}
        \nu(\chi^{\inv}(B_i)) = \chi(\chi^{\inv}(B_i)) \cap \chi(parent(\chi^{\inv}(B_i))) = B_i \cap \chi(parent(\chi^{\inv}(B_i)).
    \end{align}
    
    We will now choose the PMTD $P_2$ where $M_{2} = M_{1} \cup \parent(\chi^{\inv}(B_i))$ and give it to the adversary. Once again, the adversary cannot pick a $T$-view since that will change $\sfBT(\rho_b)$ and must choose an $S$-view associated with one of the nodes in $M_{2}$. Suppose the adversary picks $S$-view associated with a node $\chi^{\inv}(B_{i'}) \in M_{2}$. We generate the next PMTD $P_3$ where $M_{3} = M_{2} \cup \parent(\chi^{\inv}(B_{i'}))$. In general, the PMTD $P_{q+1}$ generated after $q$ rounds has $M_{q+1} = M_q \cup \parent(\chi^{\inv}(B_q))$, where $\chi^{\inv}(B_q) \in M_q$ is the node $\rho_b$ picked as a $S$-view in the last round (for PMTD $P_q$ with materialization set $M_q$). This process terminates when the $S$-view corresponding to the root is picked by the adversary. It is also guaranteed that this process will terminate after a finite number of steps as we always add nodes in the materialization set by moving up the tree and the tree is of finite size. At every step, the adversary picks an $S$-view for a bag and can only get closer to the root across all branches. Therefore, when the adversary reaches the root, there must be a subset of picked $S$-views corresponding to some path starting from the root $t_1 = r$ to some node $t_{\ell}$ in $\tree$ such that $\chi(t_{\ell}) \in \sfBT(\rho_b)$, the desired property. Also, note that these $S$-views follow \eqref{Aj_prop}, thus we complete the proof.
\end{proof}

\autoref{lem:tree} tells us that it is sufficient to analyze only the 2-phase disjunctive rules that are of the form as shown in~\eqref{eq:rule}. Now we proceed to find a proof sequence for~\eqref{eq:rule}. Let the bags for $t_1 = r, t_2, \ldots, t_{\ell}$ have corresponding fractional edge covers $\bu^{(1)}, \dots, \bu^{(\ell)}$, where $\bu^{(j)} = (u^{(j)}_F)_{F \in \edges}$, for $j \in [\ell]$. We define $\alpha_j$ to be the slack of each bag w.r.t. the variables in $A_j$ (i.e. $\alpha_j = \alpha(\bu^{(j)}, A_j)$), and introduce a factor $\beta_j = \beta^\star / \alpha_j$ where $\beta^\star = \sum_{j \in [\ell]} \alpha_j$. Next, we apply~\autoref{thm:main1} for each of the $\ell$ bags, multiply the proof sequence obtained for the $j$-{th} bag with $\beta_j$, and sum up terms as follows:

\begin{align*}
    \sum_{j \in [\ell]} \beta_j \sum_{F \in \edges} u^{(j)}_F \cdot \log |R_F| + \beta^\star \cdot \log{|Q_A|} & \geq \sum_{j \in [\ell]} \beta_j \sum_{F \in \edges} u^{(j)}_F  \cdot ( {\hT(F \mid {A_j \cap F})} + {\hS({A_j \cap F})} ) + \beta^\star \cdot {\hT({A_1})}\\
    & =  \sum_{j \in [\ell]} \beta_j \sum_{F \in \edges} u^{(j)}_F  \cdot {\hT(F \mid {A_j \cap F})} +\beta^\star \cdot {\hT({A_1})}  + \sum_{j \in [\ell]} \beta_j \sum_{F \in \edges} u_F \cdot {\hS({A_j \cap F})}  \\
    & \geq \sum_{j \in [\ell]} \beta_j \cdot \slack_j \sum_{F \in \edges} \hat{u}^{(j)}_F  \cdot {\hT({F} \mid {A_j \cap F})}  + \beta^\star \cdot {\hT({A_1})}  +  \sum_{j \in [\ell]} \beta_j \cdot {\hS({A_j})} \\
    & = \sum_{j \in [\ell]} \beta^\star \sum_{F \in \edges} \hat{u}^{(j)}_F  \cdot {\hT({F} \mid {A_j \cap F})}  + \beta^\star \cdot {\hT({A_1})}  +  \sum_{j \in [\ell]} \beta_j \cdot {\hS({A_j})} \\
    & = \sum_{j \in \{2, \dots, \ell\}} \beta^\star \sum_{F \in \edges} \hat{u}^{(j)}_F  \cdot {\hT({F} \mid {A_j \cap F})}  +  \beta^\star \left( \sum_{F \in \edges} \hat{u}^{(1)}_F  \cdot {\hT({F} \mid {A_1 \cap F})} + {\hT({A_1})}\right)  +  \sum_{j \in [\ell]} \beta_j \cdot {\hS({A_j})} \\
    & \geq \sum_{j \in \{2, \dots, \ell\}} \beta^\star \sum_{F \in \edges} \hat{u}^{(j)}_F  \cdot {\hT({F} \mid {A_j \cap F})}
    + \beta^\star {\hT({\chi(t_1)})}
     +  \sum_{j \in [\ell]} \beta_j \cdot {\hS({A_j})} \\
     & \geq \sum_{j \in \{2, \dots, \ell\}} \beta^\star \sum_{F \in \edges} \hat{u}^{(j)}_F  \cdot {\hT({F} \mid {A_j \cap F})}
    + \beta^\star {\hT({A_2})}
     +  \sum_{j \in [\ell]} \beta_j \cdot {\hS({A_j})} \\
    & \hdots \\
    & \geq \beta^\star {\hT({\chi(t_{\ell})})} +  \sum_{j \in [\ell]} \beta_j \cdot {\hS({A_j})},
\end{align*}
where the second inequality is a direct application of Shearer’s Lemma and the third inequality is a direct consequence of Lemma~\ref{lem:conditional:Shearer}. From this proof sequence, we obtain the following intrinsic tradeoff,
\begin{align*}
    |Q_A|^{\beta^\star} \cdot  |\mD|^{\sum_{j \in [\ell]} \beta_j \cdot u^*_j} \cong S^{\sum_{j \in [\ell]} \beta_j} \cdot T^{\beta^\star}
\end{align*}
where $u_j^* = \sum_{F \in \edges} u_F^{(j)}$. Equivalently, we get
\begin{align}\label{tradeoffDecomposition}
      |Q_A| \cdot  |\mD|^{\sum_{j \in [\ell]} u^*_j / \alpha_j} \cong S^{\sum_{j \in [\ell]} 1/\alpha_j} \cdot T.
\end{align}
In the above tradeoff, for a given $S$, (assume $|Q_A| = 1$), we get
\begin{align*}
      \log T & =  \sum_{j \in [\ell]} \frac{u^*_j}{\alpha_j}\cdot \log |\mD| - \sum_{j \in [\ell]} (1/\alpha_j) \log S  \\
      & = \sum_{j \in [\ell]} (1/\alpha_j) \cdot (u^*_j  \log |\mD| - \log S)
    \end{align*}
One observation is that if some bag $t_j$ on the path $t_1, \ldots, t_{\ell}$ has an AGM bound that is not greater than $S$, then the materialization of $t_j$'s corresponding $S$-view $S(\bx_{A_j})$ can be fully materialized as the model for \eqref{eq:rule}. Otherwise, for every $t_j, j \in [\ell]$, we have that $(u^*_j \log |\mD| - \log S)$ is non-negative for every $j \in [\ell]$, thus the above expression for $\log T$ monotonically increases as $\ell$ increases. Therefore, the most expensive tradeoff corresponds to the disjunctive rule of the form in~\eqref{eq:rule} that starts from the root and ends at a leaf. To obtain the final space-time tradeoffs, we simply take the worst tradeoffs across all the root-to-leaf paths in $\htree$. 

% Next, we show that the most expensive tradeoff corresponds to the disjunctive rule of the form in~\eqref{eq:rule} that starts from the root and ends at a leaf. For the next lemma, we assume that all bags starting from the root to some node $t_\ell$ have AGM bound at least as large as $S$. This is because, as soon as $\log S = \sum_{F \in \edges_p} u^p_F \cdot \log |\mD|$ for some node $t_p$, the materialization of the bag of $t_p$ can be the model for~\eqref{eq:rule}, thus costing constant time in the online phase.

% we can directly store t
% do not need to consider any nodes below in the decomposition. Thus, only $\ell \leq p$ needs to be considered.

\eat{
\begin{lemma}
    For an arbitrarily-fixed space budget $S$, the answering time monotonically increases as a function of $\ell$.
\end{lemma}
\begin{proof}
    In order to prove the monotonicity property, it is sufficient to show that the expression
    \begin{align*}
        \sum_{j \in [\ell]} (1/\alpha_j) \cdot \left(\sum_{F \in \edges_j} u^j_F \cdot \log |\mD| - \log S \right)
    \end{align*}
    is monotonically increasing as a function of $\ell$. The first term is a monotonically increasing function since the slack $\alpha_j \geq 1$ by definition. The second term is also an increasing function since $\sum_{F \in \edges_j} u^j_F \cdot \log |\mD|$ is the AGM bound and the value of $S$ can be at most the AGM bound by assumption. Thus, the answering time can only increase.
    % Otherwise, the materialization of the $j$-th bag can be the model for~\eqref{eq:rule}, which results in constant time in the online phase (and more precisely $\geq 1$). 
\end{proof}
}

% The authors of~\cite{deep2018compressed} stated in Theorem 13 (of ~\cite{deep2018compressed}) that: take a fixed tree decomposition $(\htree, \chi)$ of the access CQ of a CQAP $\varphi(\bx_A \mid \bx_A)$, associate every node $t$ with a hyper-parameter $\delta(t)$ and a parameterized width $\rho(\delta(t)) \defeq \min_{\bu} (\sum_{F \in \edges} u_F -\delta(t) \cdot \alpha)$, where ${\bu}$ is a fractional edge cover of $\chi(t)$ and $\alpha$ is the slack (on the access p of the bag).

% Then, there is a data structure of space $O(|\mD| + |\mD|^{\max_{j \in [\ell]} \rho_j(\delta(t_j))})$, where $(t_j)_{j \in [\ell]}$ is some root-to-leaf path of $\htree$, such that any access request can be answered in time $T = O(|\mD|^{\sum_{j \in [\ell]} \delta(j)})$.
% The bags $t_1, \ldots, t_{\ell}$ correspond to a root-to-leaf path in the decomposition. The best possible tradeoff of Theorem 13 in~\cite{deep2018compressed} is then obtained by tuning the hyper-parameters $(\delta_j)_{j \in [\ell]}$.

Before we conclude this section, we show that the tradeoff we obtained in \eqref{tradeoffDecomposition} across all root-to-leaf paths of a fixed tree decomposition $(\htree, \chi)$ recovers (and possibly improves over) Theorem 13 of~\cite{deep2018compressed}, without incurring extra hyper-parameters. Indeed, in \eqref{tradeoffDecomposition}, the authors set a hyper-parameter $\delta(t)$ for every $t \in V(\htree)$ and
let the online answering time to be $T = |\mD|^{\sum_{j \in [\ell]} \delta(t_j)}$, where $t_1, \ldots, t_{\ell}$ is a root-to-leaf path of $\htree$ that maximizes $\sum_{j \in [\ell]} \delta(t_j)$. Suppose we construct a 2-phase disjunctive rule of the form \eqref{tradeoffDecomposition} for this root-to-leaf path, $t_1, \ldots, t_{\ell}$, then for any fractional edge cover $u^*_j, j \in [\ell]$, it holds that (by rewriting  \eqref{tradeoffDecomposition})
\begin{align*}
    \log_{|\mD|} S & = \frac{1}{\sum_{j \in [\ell]} (1/\alpha_{j})} \cdot \left(\sum_{j \in [\ell]} \frac{u^*_j}{\alpha_j} - \log_{|\mD|} T \right) \\
    & = \frac{1}{\sum_{j \in [\ell]} (1/\alpha_j)} \cdot \left(\sum_{j \in [\ell]} \frac{u^*_j}{\alpha_j} - \sum_{j \in [\ell]} \delta(t_j) \right) \\
    & = \frac{1}{\sum_{j \in [\ell]} (1/\alpha_j)} \cdot \sum_{j \in [\ell]}  \frac{1}{\alpha_j} (u^*_j - \delta(t_j) \cdot \alpha_j ) \\
    & \leq \max_{j \in [\ell]}
      \left(\sum_{F \in \edges} u_F^{(j)} - \delta(t_j) \cdot \alpha_j \right).
\end{align*}
The last inequality holds because for all $w_i\geq 0$ such that $\sum_i w_i = 1$, $\sum_{i=1}^n \gamma_i w_i \leq \max_i \gamma_i$. Setting $w_i = \frac{1}{\alpha_i} \cdot \frac{1}{\sum_{j \in [\ell]} (1/\alpha_j)}$, we get the desired expression. Thus, it is easy to see that the tradeoff we obtained in \eqref{tradeoffDecomposition} across all root-to-leaf paths of a fixed tree decomposition indeed recovers results  in~\cite{deep2018compressed}.

%\paris{Complete the proof by showing how we recover prior result} \hangdong{probably no time}

%\hangdong{this section should go later than $\PANDA$ sections}

\subsection{Additional Examples}

In this section, we present a number of concrete examples of space-time tradeoffs obtained through our framework. Across all examples, we assume the database size to be $|\mD|$. Moreover, for brevity, we use an ordered tuple of views ($T$-views and $S$-views) to represent a PMTD that takes a path-like structure (the first entry of the tuple denotes the root). For ease of interpreting tradeoffs from proof sequences, we always carry the implied upper bound of its corresponding joint Shannon-flow inequalities at the left-hand side of the proof sequences. For the implied upper bound, we denote $n_{F} \defeq \log |R_F|$ as the log-size of the relation $R_F$ and $w_{A} \defeq \log |Q_A|$ as the log-size of the access request $Q_A$.

\begin{example}[The triangle query]
We take a triangle query with an \textit{empty} access pattern, i.e. $A = \emptyset$:
$$   \varphi(x_1, x_3 \mid \emptyset) \leftarrow R(x_1, x_2) \wedge R(x_2, x_3)\wedge R(x_3, x_1). $$

We consider two PMTDs, both with the bag $\{x_1, x_2, x_3\}$. In the first PMTD, the bag is not materialized and hence a $T$-view $T_{123}$ associated with it. In the second PMTD, the bag is materialized and has an $S$-view $S_{13}$. Thus, the two PMTDs can be denoted as 
$$ (T_{123}), \quad (S_{13}).
$$
Hence, we obtain the following disjunctive rule (without the atom $Q_A$):
$$ T_{123} \vee S_{13} \leftarrow R(x_1, x_2)\wedge R(x_2, x_3)\wedge R(x_3, x_1).
$$
One (empty) proof sequence for it is simply $\log |\mD| \geq {\hS(1 3)}$, indicating that we can store all pairs of $(x_1, x_3)$ that participate in at least one triangle in linear space. 
\end{example}

\begin{example}[The square query]
For this example, we take the following CQAP:
$$   \varphi(x_1, x_3 \mid x_1, x_3 ) \leftarrow R(x_1, x_2) \wedge R(x_2, x_3)\wedge R(x_3, x_4) \wedge R(x_4, x_1). $$
This captures the following task: given two vertices of a graph, decide whether they occur in two opposite corners of a square. We consider two PMTDs. The first PMTD has a root bag $\{1, 3, 4\}$ associated with a $T$-view $T_{134}$, and a bag $\{1, 3, 2\}$ associated with a $T$-view $T_{132}$. The second PMTD has one bag $\{1, 2, 3, 4\}$ associated with an $S$-view $S_{13}$. The two PMTDs can be denoted as
$$ (T_{134}, T_{132}), \quad (S_{13})
$$
This in turn generates two disjunctive rules:
\begin{align*}
T_{134} \vee S_{13} & \leftarrow  Q_{13}(x_1, x_3) \wedge R(x_1, x_2) \wedge R(x_2, x_3) \wedge R(x_3, x_4) \wedge R(x_4, x_1) \\
T_{132} \vee S_{13} & \leftarrow Q_{13}(x_1, x_3) \wedge R(x_1, x_2) \wedge R(x_2, x_3) \wedge R(x_3, x_4) \wedge R(x_4, x_1).
\end{align*}
We can construct the following proof sequence for the first rule:
\begin{align*}
    n_{14} + n_{34} + 2 \cdot {w_{13}}
    & \geq {\hS(1)} + {\hT( 4 | 1)} + {\hS(3)} + {\hT(4| 3)}+ 2\cdot {\hT(1 3)}\\
    & \geq {\hS(1 3)} + {\hT(4 | 1)} + {\hT(4|3)} + 2\cdot {\hT(1 3)} \\ 
    & \geq {\hS(1 3)} + {\hT(4 | 1 3)} + {\hT(1 3)}  + {\hT(4| 1 3)} + {\hT(1 3)} \\ 
    & = {\hS(1 3)} + 2 \cdot {\hT(1  3 4)}
\end{align*}
For the second rule, we symmetrically construct a proof for  $ n_{12} + n_{32} + 2 \cdot {w_{13}} \geq  {\hS(1 3)} + 2 \cdot {\hT(1  3 2)}$. Hence, we obtain a tradeoff of $S \cdot T^2 \cong |\mD|^2 \cdot |Q_A|^2 $. This tradeoff (when $|Q_A| = 1$) recovers the improved one obtained in Example 15 of \cite{deep2021space}.
\end{example}

\subsection{Tradeoffs for $k$-Reachability (Section~\ref{sec:path})}
\eat{
\begin{proof}[Proof for Path Queries]
Assume for now that $k=2\ell$.
To obtain the desired tradeoff, it suffices to consider the following set of partially materialized tree decompositions:
$$ T_{12\dots (2\ell+1)} \quad S_{1 (2\ell+1)} \quad T_{12\dots j (2\ell+2-j) \dots (2 \ell+1)} \rightarrow S_{j (2\ell+2-j)} \quad \text{for $j=2, \dots, \ell$} $$
From this, we obtain the following targets for our disjunctive rules for $m=2, \dots, \ell+1$:
\begin{align*}
T_{12\dots m (2\ell+2-m) \dots 2 \ell+1} \vee \{ S_{i (2\ell+2-i)} \}_{i=1, \dots, m-1} & \leftarrow  \dots \\
\end{align*}
We can now construct the following proof for any $m= 2, \dots, \ell$:
\begin{align*}
   &      \sum_{i=1}^{m-1} 2 \cdot h(x_i, x_{i+1}) +\sum_{i=2\ell+2-m}^{2 \ell} 2 \cdot h(x_i, x_{i+1})  + 2 \cdot {h(x_1, x_{k+1})} && \\
    & \geq   \sum_{i=1}^{m-1} 2 \cdot \{ {\hT(x_{i+1} \mid x_i)} + {\hS(x_i)}  \} +\sum_{i=2\ell-m+2}^{2 \ell} 2 \cdot \{ {\hT(x_{i} \mid x_{i+1})} + {\hS(x_{i+1})}  \}  + 2 \cdot {\hT(x_1, x_{k+1})} && \textit{bucketize} \\
    & =  2 \sum_{i=1}^{m-1} \{{\hS(x_{i})} + {\hS(x_{2\ell+2-i})} \} + 2 \cdot \left(  \sum_{i=1}^{m-1} \{ {\hT(x_{i+1} \mid x_i)}  + \sum_{i=2\ell-m+2}^{2 \ell} \{ {\hT(x_{i} \mid x_{i+1})} + {\hT(x_1, x_{k+1})} \right) && \\
    & \geq  2 \sum_{i=1}^{m-1} \{{\hS(x_{i}, x_{2\ell+2-i})} \} + 2 \cdot \left(  \sum_{i=1}^{m-1} \{ {\hT(x_{i+1} \mid x_i)}  + \sum_{i=2\ell-m+2}^{2 \ell} \{ {\hT(x_{i} \mid x_{i+1})} + {\hT(x_1, x_{k+1})} \right) && \textit{join} \\
       & \geq  2 \sum_{i=1}^{m-1} \{{\hS(x_{i}, x_{2\ell+2-i})} \} + 2 \cdot {\hT(x_1, \dots, x_m, x_{2\ell+2-m}, \dots, x_{2\ell+1})}&& \textit{submodularity, join} 
       \end{align*}
This obtains a tradeoff of $|D|^{4(m-1)} \cdot W^2 \geq S^{2(m-1)} \cdot T^2$ for $m=2, \dots, \ell$. For $m= 2\ell+1$, we have a slightly different proof structure:
\begin{align*}
   &      \sum_{i=1}^{\ell-1} 2 \cdot h(x_i, x_{i+1}) + h(x_\ell, x_{\ell+1}) + h(x_{\ell+1}, x_{\ell+2})  + \sum_{i=\ell+2}^{2 \ell} 2 \cdot h(x_i, x_{i+1})  + 2 \cdot {h(x_1, x_{k+1})} && \\
    & \geq   \sum_{i=1}^{\ell-1} 2 \cdot \{ {\hT(x_{i+1} \mid x_i)} + {\hS(x_i)}  \} + {\hT(x_{\ell+1} \mid x_{\ell})} + {\hS(x_\ell)}  + {\hT(x_{\ell+1} \mid x_{\ell+1})} + {\hS(x_{\ell+2})} +  \sum_{i=\ell+2}^{2 \ell} 2 \cdot \{ {\hT(x_{i} \mid x_{i+1})} + {\hS(x_{i+1})}  \}  + 2 \cdot {\hT(x_1, x_{k+1})} &&  \\
    & =  2\sum_{i=1}^{\ell-1} \{{\hS(x_{i})} + {\hS(x_{2\ell+2-i})} \} + {\hS(x_{\ell}) + h(x_{\ell+1})}  + 2 \cdot \left(  \sum_{i=1}^{\ell-1} \{ {\hT(x_{i+1} \mid x_i)}  + \sum_{i=\ell+2}^{2 \ell} \{ {\hT(x_{i} \mid x_{i+1})} + {\hT(x_1, x_{k+1})} \right) + {\hT(x_{\ell+1} \mid x_{\ell})} + {\hT(x_{\ell+1} \mid x_{\ell+1})}&& \\
    & = 2 \sum_{i=1}^{\ell-1} \{{\hS(x_{i}, x_{2\ell+2-i})} \} + {\hS(x_{\ell}, x_{\ell+1})}  + 2 \cdot \left(  \sum_{i=1}^{\ell-1} \{ {\hT(x_{i+1} \mid x_i)}  + \sum_{i=\ell+2}^{2 \ell} \{ {\hT(x_{i} \mid x_{i+1})} + {\hT(x_1, x_{k+1})} \right) + {\hT(x_{\ell+1} \mid x_{\ell})} + {\hT(x_{\ell+1} \mid x_{\ell+1})}&& \\
       & \geq 2 \sum_{i=1}^{\ell-1} \{{\hS(x_{i}, x_{2\ell+2-i})} \} + {\hS(x_{\ell}, x_{\ell+1})}  + 2 \cdot {\hT(x_1, \dots, x_{2\ell+1})}&& 
       \end{align*}
This obtains a tradeoff of $|D|^{2(k-1)} \cdot W^2 \geq S^{k-1} \cdot T^2$ 
\end{proof}
}

\begin{example}[$2$-reachability]
Consider the 2-reachability CQ with the following access pattern (optimizing for $|Q_{13}| = 1$):
$$   \phi_2(x_1, x_3 \mid x_1, x_3 ) \leftarrow  R_{1}(x_1, x_2) \wedge R_{2}(x_2, x_3). $$
From the tree decomposition that has one bag $\{1, 2, 3\}$, we construct the set of all non-redundant and non-dominant PMTDs,  i.e. 
$$(T_{123}), \quad (S_{13}).
$$
They generate only one 2-phase disjunctive rule,
$$ T_{123} \vee S_{13} \leftarrow Q_{13}(x_1, x_3), R_{1}(x_1, x_2), R_{2}(x_2, x_3),
$$
for which we can construct the following proof sequence:
\begin{align*}
    n_{12} + n_{23} + 2 \cdot {w_{13}}
    & \geq {\hS(1)} + {\hT( 2 | 1)} + {\hS(3)} + {\hT(2| 3)}+ 2\cdot {\hA(1 3)} \\
    & \geq {\hS(1 3)} + {\hT(2 | 1)} + {\hT(2|3)} + 2\cdot {\hA(1 3)}\\ 
    & \geq {\hS(1 3)} + {\hT(2 | 1 3)} + {\hT(2| 1 3)} + 2\cdot {\hA(1 3)} \\ 
    & \geq {\hS(1 3)} + 2 \cdot {\hT(1 2 3)} &&  (S \cdot T^2 \cong |\mD|^2 \cdot |Q_{13}|^2)
\end{align*}
\introparagraph{Discussion} 
% if $S = |\mD|, |Q_A| = |\mD|$, it costs time $T = |\mD|^{3/2}$ to answer a workload $Q_A(x_1, x_3)$, which is the best answering time for a triangular query. 
Suppose $w_{13} = \log |\mD|$, $(1, 13, N_{13|1}) \in \AC$ and $w_{3|1} \defeq \log N_{13|1}$, then $n_{12} +  {w_{3|1}} \geq   {\hT(12)} + {\hT(3|12)} = {\hT(123)}$ is a (possibly desirable) proof sequence, which implies that we can answer any access request in time $T = |\mD| \cdot N_{13|1}$ without any materializations. Note that if $S = |\mD|, N_{13|1} = |\mD|^{1/2 - \epsilon}$ for some $\epsilon > 0$, it is strictly better than the above proof sequence (which implies $T = |\mD|^{3/2}$).
\end{example}

\begin{example}[3-reachability] We study the $3$-reachability CQAP (optimizing for $|Q_A| = 1$), i.e. 
$$ 
\phi_3(x_1, x_{4} \mid x_1, x_{4}) \leftarrow R(x_1, x_2) \wedge R(x_2, x_3) \wedge R(x_3, x_{4}).
$$
Using the set of all non-redundant and non-dominant PMTDs (shown in \autoref{fig:example3}), we generate all 2-phase disjunctive rules (after discarding redundant rules/targets), and their corresponding proof sequences:
%$$T_{134} \longrightarrow T_{123}, \quad T_{134} \longrightarrow S_{13}, \quad T_{124} \longrightarrow T_{234}, \quad T_{124} \longrightarrow S_{24}, \quad S_{14} 	$$
	\begin{enumerate}
	    \item [(1)] $\rho_1: T_{134} \vee T_{124} \vee S_{14} \leftarrow Q_{14}(x_1, x_4) \wedge R(x_1, x_2) \wedge R(x_2, x_3) \wedge R(x_3, x_{4})$
	            \begin{align*}
                    n_{12} + n_{34} + {2 w_{14}} & \geq {\hS(1)} + {\hS(4)} + {\hT(2|1)} + {\hT(3|4)} + 2 {\hA(14)} && \\
                    & \geq {\hS(14)} + {\hT(2|1) + \hT(3|4)} + 2 {\hA(14)} && \\
                    & \geq {\hS(14)} + {\hT(2|14) + \hT(3|14)} + 2 {\hA(14)} && \\
                    & \geq {\hS(14)} + {\hT(124)} +  {\hT(134)} && ( S \cdot T^2 \cong |\mD|^2 \cdot |Q_A|^2)
                \end{align*}
	    \item [(2)] $\rho_2: T_{123}  \vee S_{13} \vee T_{124} \vee S_{14} \leftarrow Q_{14}(x_1, x_4) \wedge R(x_1, x_2) \wedge R(x_2, x_3) \wedge R(x_3, x_{4}) $
	        \begin{align*}
                     2 \cdot n_{12} + n_{23} + n_{34} + 3 \cdot {w_{14}} & = 2( {\hS(1)} + {\hT(2|1)}) +  {\hS(3)} + {\hT(2|3)} + {\hS(4)} + {\hT(3|4)} + 3 \cdot {\hA(14)}\\
                     & \geq {\hS(14)} + {\hS(13)} + 2{\hT(2|14)}  + {\hT(3|41)} + {\hT(2|314)} + 3 \cdot {\hA(14)} \\
                     & = {\hS(14)} + {\hS(13)} + 2 {\hT(124)} +   {\hT(1234)}  \\
                     & \geq {\hS(14)} + {\hS(13)} + 3 \cdot {\hT(124)} && (S^2 \cdot T^3 \cong |\mD|^4 \cdot |Q_A|^3)
                \end{align*}
        \item [(3)] $\rho_3: T_{134} \vee T_{234} \vee S_{24} \vee S_{14} \leftarrow Q_{14}(x_1, x_4) \wedge R(x_1, x_2) \wedge R(x_2, x_3) \wedge R(x_3, x_{4})$ \\
            The proof sequence for $\rho_3$ is omitted here because it is symmetric to rule $\rho_2$.
	    \item [(4)] $\rho_4: T_{123} \vee S_{13} \vee T_{234} \vee S_{24} \vee S_{14} \leftarrow Q_{14}(x_1, x_4) \wedge R(x_1, x_2) \wedge R(x_2, x_3) \wedge R(x_3, x_{4})$  \\
	    For $\rho_4$, we show $2$ proof sequences that do not dominate each other as follows:
	    \begin{align*}
            n_{12} + n_{34} + {w_{14}} & \geq \hS(1) + \hS(4) + \hT(2|1) + \hT(3|4) + {\hA(14)}  \\
            & \geq {\hS(14)} + {\hT(2|14) + \hT(3|214)} +  {\hA(14)}  \\
             & \geq {\hS(14)} + {\hT(1234)} \\
              & \geq {\hS(14)} + {\hT(123)} && (S \cdot T \cong |\mD|^2 \cdot |Q_A|)
        \end{align*}
        \begin{align*}
            2 n_{23} + 2 n_{12} + 2 n_{34} + {w_{14}} 
            & \geq 2\cdot {\hS(23)} +  {\hS(12)} + {\hS(34)} + {\hS(1)} + {\hT(2|1)} + {\hS(4)} + {\hT(3|4)} + {\hA(14)} \\
            & = {\hS(2)} + {\hS(3|2)} + {\hS(3)} + {\hS(2|3)} + {\hS(12)} + {\hS(34)} + {\hS(1)} + {\hT(2|1)} + {\hS(4)} + {\hT(3|4)} + {\hA(14)} \\
            & \geq {\hS(123)} + {\hS(234)} + {\hS(2)} + {\hS(3)} + {\hS(1)} + {\hT(2|1)} + {\hS(4)} + {\hT(3|4)} + {\hA(14)} \\
            & \geq {\hS(123)} + {\hS(234)} + {\hS(24)} + {\hS(13)} + {\hT(2|1)} + {\hT(3|4)} + {\hA(14)}\\
            & \geq {2\cdot \hS(24)} + {2\cdot \hS(13)} + {\hT(2|14)} + {\hT(3|124)} + {\hA(14)} \\
            & \geq 2 \cdot {\hS(24)} + 2 \cdot {\hS(13)} + {\hT(1234)} \\
             & \geq 2\cdot {\hS(24)} + 2\cdot {\hS(13)} + {\hT(123)} 
             \qquad \qquad \qquad \qquad \qquad \qquad 
             (S^4 \cdot T \cong |\mD|^6 \cdot |Q_A|)
        \end{align*}
	\end{enumerate}
	Lastly, for every rule above, there is a proof sequence that corresponds to applying breath-first search (BFS) from scratch in time $O(|\mD|)$ at the online phase. Take $\rho_1$ as an example, we get
	\begin{align*}
	    n_{23} + w_{14} & \geq \hT(23) + \hA(14) \\
	                   & \geq \hT(1234) \\
	                   & \geq \hT(134) && (T = |\mD| \cdot |Q_A|)
	\end{align*}
    The corresponding plot of tradeoff curve is included in \autoref{fig:3path}.
\end{example}

\begin{example}[$4$-reachability] We study the following CQAP for $4$-reachability (optimizing for $|Q_A| = 1$):
$$ \phi_4(x_1, x_5 \mid x_1, x_5) \leftarrow R_{12}(x_1, x_2) \wedge R_{23}(x_2, x_3) \wedge R_{34}(x_3, x_5) \wedge R_{45}(x_4, x_5) 
$$

	We fix the following set of non-redundant and non-dominant PMTDs ($11$ in total), where we use an ordered tuple of views to represents a PMTD that has a path-like structure (the first entry of the tuple denotes the root). Note that including more PMTDs could potentially obtain even improved tradeoffs.
	$$ (T_{1235}, T_{345}), \quad (T_{1235}, S_{35}), \quad  (T_{1345},  T_{123}), \quad (T_{1345}, S_{13}), \quad (T_{1245},  T_{234}), \quad (T_{1245} , S_{24})
	$$
	$$
	(T_{125} , T_{2345}), \quad (T_{125} , S_{25}), \quad (T_{145} , T_{1234}), \quad (T_{145} , S_{14}), \quad (S_{15})
	$$
	Though there are $2^{10}$ 2-phase disjunctive rules generated from the above set of PMTDs, similar to $3$-reachability, we can discard rules with strictly more targets than other rules. To start off, note that any rule has to have $S_{15}$ as a target. For any disjunctive rules picking $T_{1245}$ as a target (or $T_{125}, T_{145}$), i.e. $\rho_1 : T_{2345} \vee S_{15}$, we always have the following proof sequence:
	\begin{align*}
	    n_{12} + n_{45} + {w_{15}} & \geq {\hS(1)} + {\hT(2|1)} + {\hS(5)} + {\hT(4|5)} \\
	    & \geq {\hS(15)} + {\hT(2|15)} + {\hT(4|125)}  + {\hT(15)} \\
	     & = {\hS(15)} + {\hT(1245)} && (S \cdot T \cong |\mD|^2 \cdot |Q_A|)
	\end{align*}
    Otherwise, for the last $7$ PMTDs, the disjunctive rule must pick $\{T_{234} , S_{24} , T_{2345} , S_{25} , T_{1234} , S_{14} , S_{15}\}$, or just $\{ T_{234} , S_{24} , S_{25} , S_{14} , S_{15}\}$, by removing the redundant term $T_{2345}$ due to the presence of $T_{234}$.
    Now the disjunctive rules picks targets out of the first $4$ PMTDs, which we can break into the following cases, discarding redundant targets:
    \begin{enumerate}
        \item [(1)] $\rho_2: T_{1235} \vee T_{1345} \vee (T_{234} \vee S_{24} \vee S_{25} \vee S_{14} \vee S_{15})$
        \begin{align*}
            n_{12} + n_{23} + n_{34} + n_{45} + {2 w_{15}} & \geq {\hS(1)} + {\hT(2|1)} + {\hS(2)} + {\hT(3|2)} + {\hS(4)} + {\hT(3|4)} + {\hS(5)} + {\hT(4|5)} + {2 \hA(15)}\\
            & \geq {\hS(15)} + {\hS(24)} + {\hT(2|15)} + {\hT(3|125)} + {\hT(3|145)} + {\hT(4|15)} + {2 \hA(15)} \\
         & = {\hS(15)} + {\hS(24)} + {\hT(1235)} + {\hT(1345)} && ( \bspace^2 \cdot \btime^2 \cong |\mD|^4 \cdot |Q_A|^2)
        \end{align*}
        \item [(2)] $\rho_3: T_{345} \vee S_{35} \vee T_{123} \vee S_{13} \vee (T_{234} \vee S_{24} \vee S_{25} \vee S_{14} \vee S_{15})$  \\
        The proof sequence for $\rho_3$ is omitted here because $\rho_3$ is no harder than $\rho_2$ since $\{3, 4, 5\} \subseteq \{1, 3, 4, 5\}$ and $\{1, 2, 3\} \subseteq \{1, 2, 3, 5\}$
        \item [(3)] $\rho_4: T_{345} \vee S_{35} \vee T_{1345} \vee (T_{234} \vee S_{24} \vee S_{25} \vee S_{14} \vee S_{15}) = T_{345} \vee S_{35} \vee (T_{234} \vee S_{24} \vee S_{25} \vee S_{14} \vee S_{15})$ \\
        For $\rho_4$, we show $2$ proof sequences that do not dominate each other as follows:
        \begin{align*}
            & 2 n_{23} + 2 n_{12} + 5 n_{34} + 3 n_{45} + 5 {w_{15}} \\
            & \geq 2({\hS(2)} + {\hT(3 | 2)}) + 2({\hS(1)} + {\hT(2 | 1)}) + 2({\hS(3)} + {\hT(4 | 3)}) + \\ 
             & \quad \quad 3({\hS(4)} + {\hT(3 | 4)}) + 3({\hS(5)} + {\hT(4|5)}) + 5 {\hA(15)} \\
            & = {2(\hT(2 | 1) + \hT(3 | 2) + \hT(4 | 3)) + 3(\hT(4 | 5) + \hT(3 | 4))} + 5 {\hA(15)} + \\
            & \quad \quad {2(\hS(3) + \hS(5)) + (\hS(2) + \hS(5)) + (\hS(2) + \hS(4)) + 2(\hS(1) + \hS(4))} \\
            & \geq {2(\hT(2|15) + \hT(3|125) + \hT(4|1235)) + 3(\hT(4|5) + \hT(3 |45))} + 2 {\hT(15)} + 3 {\hT(5)} +\\
            & \quad \quad {2(\hS(3) + \hS(5 | 3)) + (\hS(2) + \hS(5 | 2)) + (\hS(2) + \hS(4 | 2)) + 2(\hS(1) + \hS(4 | 1))} \\
            & = {2 \hT(12345) + 3 \hT(345)} + {2 \hS(3 5) + \hS(2 5) + \hS(2 4) + 2 \hS(1 4)} \\
            & \geq {5 \hT(345)} + {2 \hS(3 5) + \hS(2 5) + \hS(2 4) + 2 \hS(1 4)} && (S^6 \cdot T^5 \cong |\mD|^{12} \cdot |Q_A|^5) 
       \end{align*} 
    %   \hangdong{for example, $S = D^{4/3} = D^{11/15}\cdot D^{3/5}$, we set the degree threshold to be $D^{4/15}, D^{4/15}, D^{4/15}, D^{2/5}, D^{2/5}$ to answer in time $T = D^{4/5}$}
     % Or alternatively for $\rho_4$,
       \begin{align*}
            & 3 n_{23} + 3 n_{34} + 3 n_{45} + n_{12} + 2 n_{34} + n_{23} + 3 {w_{15}}\\
            & \geq 3({\hS(3)} + {\hS(2|3)}) + 3 {\hS(34)} + 3({\hS(5)} + {\hT(4|5)}) + ({\hS(1)} + {\hT(2|1)}) + 2({\hS(4)} + {\hT(3|4)}) + ({\hS(2)} + {\hT(3|2)}) + 3 {\hT(15)} \\
            & \geq 3 {\hS(234)} + 3 {\hS(35)} + 3 {\hT(4|5)} + ({\hS(1)} + {\hT(2|1)}) + 2({\hS(4)} + {\hT(3|4)}) + ({\hS(2)} + {\hT(3|2)}) + 3 {\hT(15)} \\
             & \geq 3 {\hS(24)} + 3 {\hS(35)} + {\hS(14)} + {\hS(24)} + 3 {\hT(4|5)} +  {\hT(2|1)} + 2{\hT(3|4)} + {\hT(3|2)} + 3 {\hT(15)} \\
             & \geq 3 {\hS(24)} + 3 {\hS(35)} + {\hS(14)} + {\hS(24)} + 3 {\hT(4|5)} +  {\hT(2|1)} + 2{\hT(3|45)} + {\hT(3|2)} + 2{\hT(5)} + {\hT(15)} \\
             & \geq 3 {\hS(24)} + 3 {\hS(35)} + {\hS(14)} + {\hS(24)} + 2 {\hT(345)} +
             {\hT(4|15)} +  {\hT(2|145)}  + {\hT(3|1245)}  + {\hT(15)} \\
             & \geq 3 {\hS(24)} + 3 {\hS(35)} + {\hS(14)} + {\hS(24)} + 2 {\hT(345)} +
            {\hT(12345)} \\
            & \geq 3 {\hS(24)} + 3 {\hS(35)} + {\hS(14)} + {\hS(24)} + 3 {\hT(345)} 
            \qquad \qquad \qquad \qquad \qquad \qquad \qquad \qquad \qquad 
            (S^8 \cdot T^3 \cong |\mD|^{13} \cdot |Q_A|^3)
       \end{align*} 
       \item [(4)] $\rho_5: T_{1235} \vee T_{123} \vee S_{13} \vee (T_{234} \vee S_{24} \vee S_{25} \vee S_{14} \vee S_{15})$  \\
        The proof sequence for $\rho_5$ is omitted here because it is symmetric to rule $\rho_4$.
    \end{enumerate}
    Similar to $3$-reachability, there is a proof sequence for every rule above that corresponds to breath-first search in the online phase. The corresponding plot of tradeoff curve is included in \autoref{fig:4path}.
\end{example}

\eat{\hangdong{@Shaleen}
the disjunctive rules to range over: 
\begin{itemize}
    \item $T(Z A) \vee S(Z) \longrightarrow |D|^4 = T^3 \cdot S$
        \begin{align*}
                 4 \log |D| + 3 \log |W|  = 3 {\hT(A)} + {\hS(Z_1 Y_1 A | A)} + {\hS(Z_2 Y_1 A | A)} + {\hS(Z_3 Y_2 A | A)} + {\hS(Z_4 Y_2 A)}  \\
                 3 {\hT(A)} + {\hS(Z_1 Y_1 A | A)} + {\hS(Z_2 Y_1 A | A)} + {\hS(Z_4 Z_3 Y_2 A | Z_4 A)} + {\hS(Z_4 A)}  \\
                  3 {\hT(A)} + {\hS(Z_1 Y_1 A | A)} + {\hS(Z_2 Y_1 A | A)} + {\hS(Z_4 Z_3 Y_2 A)}  \\
                  3 {\hT(A)} + {\hS(Z_1 Y_1 A | A)} + {\hS(Z_4 Z_3 Z_2 Y_1 A | Z_4 Z_3 A)} + {\hS(Z_4 Z_3 A)}  \\
                  3 {\hT(A)} + {\hS(Z_1 Y_1 A | A)} + {\hS(Z_4 Z_3 Z_2 Y_1 A)}   \\
                  3 {\hT(A)} + {\hS(Z_4 Z_3 Z_2 Z_1 Y_1 A | Z_4 Z_3 Z_2 A)} + {\hS(Z_4 Z_3 Z_2 A)}   \\
                   3 {\hT(A)} + {\hS(Z_4 Z_3 Z_2 Z_1 Y_1 A)}   + 3 {\hT(Z)}\\
                   3 {\hT(A Z)} + {\hS(Z_4 Z_3 Z_2 Z_1)}   \\
            \end{align*}
    \item $T(Z_1 Z_2 Y_1 A) \vee S(Z_1 Z_2 A) \vee S(Z) \longrightarrow |D|^2 = T \cdot S$
            \begin{align*}
             2 \log |D| + \log |W| = {\hT(Y_1 A)}  + {\hS(Z_1 Y_1 A | Y_1 A)} + {\hS(Z_2 Y_1 A)} \\
             {\hT(Y_1 A)}  + {\hS(Z_2 Z_1 Y_1 A | Z_2 Y_1 A)} + {\hS(Z_2 Y_1 A)} \\
             {\hT(Y_1 A)}  + {\hS(Z_2 Z_1 Y_1 A)} + {\hT(Z_2 Z_1)}\\
             {\hT(Z_2 Z_1 Y_1 A)}  + {\hS(Z_2 Z_1 A)} 
        \end{align*}
    \item $T(Z_1 Z_2 Y_1 A) \vee T(Z_3 Z_4 Y_2 A) \vee S(Z) \vee S(Z_3 Z_4 A) \longrightarrow |D|^2 = T \cdot S$
        \begin{align*}
             2 \log |D| = {\hT(Y_2 A)}  + {\hS(Z_4 Y_2 A | Y_2 A)} + {\hS(Z_3 Y_2 A)} \\
             {\hT(Y_2 A)}  + {\hS(Z_4 Z_3 Y_2 A | Z_3 Y_2 A)} + {\hS(Z_3 Y_2 A)} \\
             {\hT(Y_2 A Z_2 Z_1)}  + {\hS(Z_4 Z_3 A)} 
        \end{align*}
    \item $T(Z_1 Z_2 Y_1 A) \vee T(Z_3 Z_4 Y_2 A) \vee S(Z) \vee S(Z_1 Z_2 A) \longrightarrow |D|^2 = T \cdot S$
\end{itemize}}

%% file: disjunctiveRules.tex
\newpage
 \section{Algorithms for 2-phase Disjunctive Rules} \label{sec:panda}
 
 Let $\rho$ be a 2-phase disjunctive rule taking the form \eqref{def:partitionDisjunctiveRule}, under degree constraints $\DC$ (guarded by input relations) and degree constraints $\AC$ (guarded by access request $Q_A$).
 In this section, we introduce a na\"ive algorithm that uses the $\PANDA$ algorithm to obtain a model of $\rho$ in two phases. First, we present some necessary terminologies and results.
  
  \subsection{Background}
  
  \introparagraph{Entropic Functions}
  Given a disjunctive rule $\eqref{def:disjunctiveRule}$, a set function $h: 2^{[n]} \rightarrow \bR_{+}$ is {\em entropic} if there is a joint probability distribution on $[n]$ such that $h(F)$ is the marginal entropy of $F$ for any $F \subseteq [n]$. Let $\Gamma_n^*$ be the set of all entropic functions and $h(Y|X) \defeq h(Y) - h(X)$. Under a given set of degree constraints $\DC$, any joint distribution on $[n]$ conforms to the constraints $h(Y|X) \leq n_{Y|X}$, where $n_{Y|X} \defeq \log N_{Y|X}$, for each $(X, Y, N_{Y|X}) \in \DC$.
 
 \introparagraph{Polymatroid} 
 A polymatroid is a set function $h: 2^{[n]} \rightarrow \bR_{+}$ that is non-negative, monotone, and submodular, with $h(\emptyset)=0$. To be precise, monotonicity implies that $h(Y) \geq h(X)$ for any $X \subseteq Y \subseteq [n]$ and let $h(Y|X) \defeq h(Y) - h(X)$, then submodularity implies that $h(I|I \cap J) \geq h(I \cup J| J)$ for any $I, J \subseteq [n]$. Let $\Gamma_n$ be the set of all polymatroids on $[n]$. As every entropic function is a polymatroid, it holds that $\Gamma_n^* \subseteq \Gamma_n$. 
 
 \introparagraph{Size Bound for Disjunctive Rules} Let $\rho$ be a disjunctive rule of the form \eqref{def:disjunctiveRule}. Let $\mD$ be a database instance under a given set of degree constraints $\DC$. The set
  $$
     \HDC \defeq \left\{h: 2^{[n]} \rightarrow \mathbb{R}_{+} | \bigwedge_{\left(X, Y, N_{Y | X}\right) \in \DC} h(Y|X) \leq \log N_{Y| X}\right\}
  $$
 contains all entropic functions $h$ on $[n]$ satisfying the degree constraints $\DC$. Fix a closed subset $\mF \subseteq \bR^{2^n}_+$. We define the log-size-bound with respect to $\mF$ of a disjunctive rule $\rho$ to be the quantity:
 $$ \LogSizeBound_{\mF}(\rho) \defeq \max_{h \in \mF} \min_{B \in \sfBT} h(B)
 $$
 Then for the output size of $\rho$, we have the following theorem \highlight{(see Theorem 1.5 in \cite{DBLP:conf/pods/Khamis0S17})}:
 \begin{theorem}[\cite{DBLP:conf/pods/Khamis0S17}]
  Let $\rho$ be any disjunctive rule \eqref{def:disjunctiveRule} under degree constraints $\DC$. Then for any database instance $\mD$ satisfying $\DC$,  the following holds:
    $$ \log |\rho| \leq  \underbrace{\LogSizeBound_{\Gamma_n^* \cap \HDC}(\rho)}_{\textit{entropic bound}} \leq \underbrace{\LogSizeBound_{\Gamma_n \cap \HDC}}_{\textit{polymatroid bound}}(\rho),
 $$
 \end{theorem}

 The entropic bound is tight under degree constraints in the worst case. However, its computation is often hard in general. The polymatroid bound is tight if $\rho$ is a CQ (i.e. has a single target) and there are only cardinality constraints, in which case the polymatroid bound degenerates into the AGM bound. However, it is not tight under general degree constraints.
 
   \introparagraph{The $\PANDA$ Algorithm} Given a disjunctive rule $\rho$ of the form~\eqref{def:disjunctiveRule},  the $\PANDA$ algorithm takes a database instance $\mD$, a set of degree constraint $\DC$ (guarded by $\mD$) as inputs and computes a model in time and space predicted by its polymatroid bound:
  $$ \polyO(2^{\LogSizeBound_{\Gamma_n \cap \HDC}(\rho)}). $$
  \highlight{The reader can refer to Theorem 1.7 in \cite{DBLP:conf/pods/Khamis0S17} for details.} For now, we will treat $\PANDA$ algorithm as a blackbox; we will later present how it works.
  
  \subsection{The 2-phase Framework} \label{alg:paradim}
  \highlight{In this section, we introduce a 2-phase algorithmic framework that we will follow to design 2-phase algorithms for a 2-phase disjunctive rule of the form 
  \eqref{def:partitionDisjunctiveRule}.
  }
  
  Let $\mD$ be a database instance and $Q_A$ be an arbitrary access request. We assume that \textit{hash tables on necessary index keys (of input relations and degrees of tuples in input relations)} can be pre-built at the start of the preprocessing phase as needed by the framework. That is, we assume constant-time accesses of tuples and degrees of tuples in the input relations during both phases. There are at most $O(2^{2n})$ hash tables to be pre-built per input relation $R_F$ (that is, for every $(Y, X)$-pair where $X \subset Y \subseteq F \in \edges$), so the space cost for storing all necessary hash tables is $O(|\mD|)$ in data complexity. Recall that we also assume w.l.o.g the following \textit{best constraint assumption}: we only keep at most one $(X, Y, N_{Y|X}) \in \DC$ for each $X \subset Y \subseteq [n]$ (keep the minimum $N_{Y|X}$ if there are more than one). 
  %We define a \textit{split step}, which is necessary for our paradigm. 
  
  \introparagraph{Split Steps} Let $R \in \mD$ be the guard of a cardinality constraint $(\emptyset, Z, N_{Z|\emptyset}) \in \DC$. A {\textit{split step}} on a $(Y, X)$-pair, where $\emptyset \neq X \subset Y \subseteq Z$, applies Lemma 6.1 of \cite{DBLP:conf/pods/Khamis0S17} and partitions $R_Y \defeq \Pi_Y(R)$ into $k = 2 \log N_{Z \mid \emptyset}$ sub-tables, i.e. $R_Y^{(1)}, \ldots, R_Y^{(k)}$, such that $N_{X | \emptyset}^{(j)} \cdot N_{Y | X}^{(j)} \leq  N_{Z \mid \emptyset} $, for all $j \in [k]$, where 
 \begin{align*}
     N_{X | \emptyset}^{(j)} & \defeq \left|\Pi_X (R_Y^{(j)})\right| \\
     N_{Y | X}^{(j)} & \defeq \deg_{R_Y^{(j)}}(Y|X).
 \end{align*}
 For each of these sub-tables $R_Y^{(j)}$, we create a subproblem with inputs $(\mD^{(j)}, \DC^{(j)})$, where $\mD^{(j)} \defeq \mD \cup \{R_Y^{(j)}\}$ denotes the input tables and
 $$\DC^{(j)} \defeq \DC \cup \left\{(\emptyset, X, N_{X | \emptyset}^{(j)}), (X, Y, N_{Y | X}^{(j)}) \right\}
 $$
 denotes the degree constraints guarded by $\mD^{(j)}$.
 \textit{A sequence of split steps} on $(Y_1, X_1), (Y_2, X_2), \ldots, (Y_\ell, X_\ell)$, applies the first split step on the $(Y_1, X_1)$-pair, generating $k_1 = O(\log |\mD|)$ subproblems with inputs $(\mD^{(j)}, \DC^{(j)})$, where $j \in [k_1]$. Then, for each subproblem, applies the second split step on the $(Y_2, X_2)$-pair, generating $O((\log |\mD|)^2)$ subproblems. This iterative process goes on until every split step in the sequence is applied, thus it generates $O(poly(\log |\mD|))$ subproblems in total.
 
 \introparagraph{The 2-phase Framework} \label{paradigm}
 Now we formally characterize our 2-phase algorithmic framework. Let $S$ be the given space budget. We denote the task of obtaining a model for a 2-phase disjunctive rule $\rho$ (of the form \eqref{def:partitionDisjunctiveRule}) with input relations $\mD \cup \{Q_A\}$ satisfying degree constraints $\DC \cup \AC$ as $\rho(\mD \cup \{Q_A\}, \DC \cup \AC)$. The framework starts by using a sequence of \textit{split steps} to partition $\rho(\mD \cup \{Q_A\}, \DC \cup \AC)$ into $O(poly(\log |\mD|))$ subproblems. Then, the $j$-th subproblem with input $\mD^{(j)}$ and degree constraint $\DC^{(j)} \supseteq \DC$, denoted as $\rho(\mD^{(j)} \cup \{Q_A\}, \DC^{(j)} \cup \AC)$, either
 \begin{enumerate}
     \item [(1)] generates $S$-targets $(S_B^{(j)})_{B \in \sfBS}$ as a model of the preprocessing disjunctive rule $\rho_S$ of the form~\eqref{rule:preprocess} using $\PANDA$, provided that the output size of $\rho_S$ is within $\polyO(S)$; or
     \item [(2)] generates $T$-targets $(T_B^{(j)})_{B \in \sfBT}$ as a model of the online disjunctive rule $\rho_T$ of the form~\eqref{rule:online} using $\PANDA$. 
 \end{enumerate}
%  \begin{align*}
%     \rho_S & : \qquad \bigvee_{B \in \sfBS} S_{B}(\bx_{B}) \leftarrow \bigwedge_{F \in \edges} R_F(\bx_F),
% \end{align*}
%  or $(ii)$ generates $T$-targets $(T_B^{(j)})_{B \in \sfBT}$ as a model of the online disjunctive rule $\rho_T$~\eqref{rule:online} using $\PANDA$. 
% \begin{align*}
%       \rho_T & : \qquad \bigvee_{B \in \sfBT} T_{B}(\bx_{B}) \leftarrow Q_A(\bx_A) \wedge \bigwedge_{F \in \edges} R_F(\bx_F)
% \end{align*}

%  The preprocessing disjunctive rule $\rho_S$ excludes $Q_A$ from the body to explicitly enforce the $S$-views to be universal for any instance of access request. 
 In other words, the subproblem $\rho(\mD^{(j)} \cup \{Q_A\}, \DC^{(j)} \cup \AC)$ is conquered by applying $\PANDA$ to obtain either a model $\rho_S$ with input relations $\mD^{(j)}$ under degree constraint $\DC^{(j)}$, denoted as $\rho_S(\mD^{(j)}, \DC^{(j)})$, or a model of $\rho_T$ with input relations $\mD^{(j)} \cup \{Q_A\}$ under degree constraint $\DC^{(j)} \cup \AC$, denoted as $\rho_T(\mD^{(j)} \cup \{Q_A\}, \DC \cup \AC)$. After all subproblems are computed, the model of $\rho$ is simply the union over $S$-targets and $T$-targets generated from all subproblems.

 \introparagraph{Analysis of the 2-phase Framework} Next, we analyze the intrinsic space-time tradeoff (specified in \autoref{sec:TPDR}, between $S_{\rho}$ and $T_{\rho}$) that can be obtained by the 2-phase framework introduced above. First, the split steps incur a poly-logarithmic factor on both $S_{\rho}$ and $T_{\rho}$ and spawn $O(poly(\log |\mD|))$ subproblems. Then, the $j$-th spawned subproblem $\rho(\mD^{(j)} \cup \{Q_A\}, \DC^{(j)} \cup \AC)$  is conquered by $\PANDA$ in one of the two phases. Note that $\DC^{(j)} \supseteq \DC$ contains extra degree constraints due to the split steps. For ease of analysis, we separate out the extra constraints by defining a set $\SC^{(j)} \defeq \DC^{(j)}\setminus \DC$. To apply the polymatroid bound, we use two polymatroids, $\hS \in \Gamma_n$ to represent the preprocessing phase and $\hT \in \Gamma_n$ to represent the online phase. We define 
  $$
     \HDC \defeq \left\{h: 2^{[n]} \rightarrow \mathbb{R}_{+} | \bigwedge_{\left(X, Y, N_{Y | X}\right) \in \DC} h(Y|X) \leq \log N_{Y| X}\right\}
  $$
  $$
     \HSC^{(j)} \defeq \left\{h: 2^{[n]} \rightarrow \mathbb{R}_{+} | \bigwedge_{\left(X, Y, N_{Y | X}\right) \in \SC^{(j)}} h(Y|X) \leq \log N_{Y|X}\right\}
  $$
  $$
     \HWC \defeq \left\{h: 2^{[n]} \rightarrow \mathbb{R}_{+} | \bigwedge_{\left(X, Y, N_{Y | X}\right) \in \AC} h(Y|X) \leq \log N_{Y| X}\right\}
  $$
  where $h(Y|X) = h(Y) - h(X)$, to denote collections of set functions that satisfy $\DC$, $\SC^{(j)}$ and $\AC$, respectively. The 2-phase framework enforces that $\hS$ conforms to $\HDC \cap \HSC^{(j)}$ and $\hT$ conforms to $\HDC \cap \HSC^{(j)} \cap \HWC$. Thus, the  $j$-th subproblem costs space $\polyO(S_{\rho}^{(j)})$ in the preprocessing phase, where 
  \begin{align}\label{bound:S}
      \log S_{\rho}^{(j)} \defeq \LogSizeBound_{\hS \in \Gamma_n \cap \HDC \cap \HSC^{(j)}} (\rho_S)
  \end{align}
  provided that $S_{\rho}^{(j)} \leq S$. Otherwise, $j$-th subproblem costs time (and space) $\polyO(T_{\rho}^{(j)})$, where
  \begin{align}\label{bound:T}
      \log T_{\rho}^{(j)} \defeq \LogSizeBound_{\hT \in \Gamma_n  \cap \HDC \cap \HSC^{(j)} \cap \HWC} (\rho_T) 
  \end{align}
  By conquering all subproblems, we conclude that $\bspace_{\rho} = \max_j S_{\rho}^{(j)} \leq S$ and $\btime_{\rho} = \max_j T_{\rho}^{(j)}$. 
  
%  Taking maximum over all $|S^{(j)}|$, we get the space cost for the (worst-case) subproblem, denoted as $\bspace$. Similarly, the time cost for the (worst-case) subproblem is the max over all $|T^{(j)}|$, denoted as $\btime$. In the following, we summarize the space costs in the preprocessing phase and time costs in the online phase, using the 2-phase paradigm:
% \begin{enumerate}
%     \item [(1)] \textit{Preprocessing phase: } the paradigm stores the $S$-views $(S_B)_{B \in \sfBS}$ in space $\polyO(\bspace)$. The overall space cost is $\polyO(\bspace + |\mD|)$. Note that the preprocessing phase has no knowledge of $Q_A$ except for the degree constraints $\AC$. 
%     \item [(2)] \textit{Online phase: } given any access request $Q_A$ (under $\AC$) in the online phase, the paradigm computes the $T$-views $(T_B)_{B \in \sfBT}$ in time $\polyO(\btime)$. The overall time cost online is $\polyO(\btime + |Q_A|)$, where $|Q_A|$ is unavoidable from reading the access request $Q_A$.
% \end{enumerate}
% 
% \smallskip
 
% Intuitively, a larger $\bspace$ is capable of conquering more subproblems at preprocessing, thus alleviating the task in the online phase (and vice versa). 

  \subsection{A Na\"ive 2-phase Algorithm}
  \highlight{In this section, we use the 2-phase framework to design a 2-phase algorithm for a 2-phase disjunctive rule $\rho$ that attains the smallest possible $T_{\rho}$ for a fixed space budget $S$.}
  
  Recall that for each $(\emptyset, Z, N_{Z|\emptyset}) \in \DC$, there are at most $2^{2 n}$ $(Y, X)$-pairs with $\emptyset \neq X \subset Y \subseteq Z$. Thus, the total number of distinct split steps is a constant in data complexity. To exploit the full potential of split steps, we design a na\"ive 2-phase algorithm that applies a sequence of all distinct split steps. Also, recall that a finite sequence of split steps spawns $O(poly(\log |\mD|))$ subproblems.
  
  \introparagraph{The Na\"ive Algorithm} 
  As said, the na\"ive algorithm first applies a sequence of all distinct split steps. Intuitively, this partitions $\mD$ into its most fine-grained pieces. Let $\rho(\mD^{(j)} \cup \{Q_A\}, \DC^{(j)} \cup \AC)$ be the $j$-th subproblem spawned after the sequence of all distinct split steps. Recall that $\SC^{(j)} = \DC^{(j)}\setminus \DC$. The following {{\em splitting property}} is a direct result of a sequence of all distinct split steps on $\DC$: for any $(\emptyset, Z, N_{Z|\emptyset}) \in \DC$ and $(Y, X)$-pair with $\emptyset \neq X \subset Y \subseteq Z$, there are some $(\emptyset, X,  N_{X | \emptyset}^{(j)}), (X, Y, N_{Y|X}^{(j)}) \in \SC^{(j)}$ such that $ N_{X | \emptyset}^{(j)} \cdot N_{Y | X}^{(j)} \leq N_{Z | \emptyset}$.
  Though each subproblem varies in its own $\SC^{(j)}$, the splitting property holds across all subproblems. To encode the splitting property, we define split constraints.
  
  \begin{definition}[Split Constraints] Let $\DC$ be a set of degree constraints. A \textit{split constraint} is a triple $(X, Y|X, N_{Z|\emptyset})$ where $\emptyset \neq X \subset Y \subseteq Z, (\emptyset, Z, N_{Z|\emptyset}) \in \DC$. A relation $R_F$ is said to guard a split constraint $(X, Y|X, N_{Z|\emptyset})$ if $R_F$ guards $(\emptyset, Z, N_{Z|\emptyset})$. The set of all split constraints spanned from $\DC$, is denoted as
   $$ \SC \defeq \left\{(X, Y|X, N_{Z|\emptyset}) \mid  \emptyset \neq X \subset Y \subseteq Z, (\emptyset, Z, N_{Z|\emptyset}) \in \DC \right\}.
  $$
  \end{definition}
  
  Intuitively, each triple $(X, Y|X, N_{Z|\emptyset}) \in \SC$ encodes the splitting property for the $(Y, X)$-pair on $(\emptyset, Z, N_{Z|\emptyset}) \in \DC$. Since we assume that every $(\emptyset, Z, N_{Z|\emptyset}) \in \DC$ has at least one guard, every $(X, Y|X, N_{Z|\emptyset}) \in \SC$ is guarded by some $R_F \in \mD$.

  The na\"ive algorithm stores $S$-views for
  $\rho_S(\mD^{(j)}, \DC^{(j)})$ whenever its polymatroid bound, as specified in \eqref{bound:S}, is no larger than $\log \bspace$. Otherwise, it applies $\PANDA$ algorithm as a black box for $\rho_T(\mD^{(j)} \cup \{Q_A\}, \DC^{(j)} \cup \AC)$ in the online phase in time as specified in \eqref{bound:T}.

  \introparagraph{Analysis of the Na\"ive Algorithm} \highlight{The na\"ive algorithm exploits the full potential of the framework (by exhausting all possible split steps) and gets the best possible $T_{\rho}$ when $S_{\rho} \leq S$ (up to a poly-logarithmic factor). In particular, we state the following theorem for the best possible $T_{\rho}$:}
%   $$ \LogSizeBound_{\hT \in \Gamma_n \cap \HDC \cap \HSC^{(j)} \cap \HWC}(\rho_T) \cdot \mathbbm{1}{ \{\LogSizeBound_{\hS \in \Gamma_n \cap \HDC \cap \HSC^{(j)}}(\rho_S) > \log \bspace \} },
%   $$
%   where $\mathbbm{1}\{\cdot\}$ denotes the indicator function. As $\hS$ and $\hT$ are independent in the above, it be re-written equivalently as
 \begin{theorem}
  For a 2-phase disjunctive rule \eqref{def:partitionDisjunctiveRule}, under space budget $S$, the na\"ive algorithm obtains $T_{\rho} = 2^{\OBJ(S)}$, where 
   \begin{equation}
    \begin{aligned}
       \OBJ(S) = \max_{ \substack{ \hS \in \HDC,  \hT \in \HDC \cap \HWC \\ (\hS, \hT) \in (\Gamma_n \times \Gamma_n) \cap \HSC} } \quad & \min_{B \in \sfBT} \hT(B) \\
     \textit{s.t. }  \quad  \quad & \hS(B) > \log \bspace,  & B \in \sfBS,
    \end{aligned}
    \label{OPT-maximin}
 \end{equation}
 assuming that $\OBJ(S)$ is postive and bounded.
 \end{theorem}
 \begin{proof}
  Merging \eqref{bound:S} and \eqref{bound:T}, we get that $\log T_{\rho}^{(j)}$ for the $j$-th subproblem can be expressed as
  \begin{equation}\label{OPT-subprob}
      \begin{aligned}
             \log T^{(j)}_{\rho} = \max_{ \substack{ \hS \in \HDC \cap \HSC^{(j)} \\ \hT \in \HDC \cap \HSC^{(j)} \cap \HWC  \\ \hS, \hT \in \Gamma_n} } \quad & \min_{B \in \sfBT} \hT(B) \\
     \textit{s.t. }  \quad  \quad & \hS(B) > \log \bspace,  & B \in \sfBS.
      \end{aligned}
  \end{equation}
  Note that the maximin optimization \eqref{OPT-subprob} is subproblem-dependent, since it is constrained on $\HSC^{(j)}$. To avoid this dependency, we recall that the na\"ive algorithm dictates the splitting property. Thus, we define the (subproblem-independent) set $\SC$ as follows:
   \begin{align*}
     \HSC \defeq \left\{(\hS, \hT): 2^{[n]} \times 2^{[n]} \rightarrow \bR^2_+| \bigwedge_{(X, Y|X, N_{Z|\emptyset})\in \SC} ( \hS(X) + \hT(Y|X) \leq \log N_{Z|\emptyset} ) \wedge ( \hS(Y|X) + \hT(X) \leq \log N_{Z|\emptyset} ) \right\},
 \end{align*}
 where $\hS(Y|X) \defeq \hS(Y) - \hS(X), \hT(Y|X) \defeq \hT(Y) - \hT(X)$. $\HSC$ is a universal collection of set functions pairs 
 satisfying the splitting property and thus, it correlates $\hS$ and $\hT$. Since every subproblem satisfies the splitting property, it holds that $\HSC^{(j)} \times \HSC^{(j)} \subseteq \HSC$. By relaxing $\HSC^{(j)} \times \HSC^{(j)}$ to $\HSC$, we get the desired upper bound \eqref{OPT-maximin} for $T_{\rho}$.
 \end{proof}

   By assigning $\hT$ to be always $0$, it is easy to see that the feasible region of \eqref{OPT-maximin} is empty if and only if
    $$
    \LogSizeBound_{\hS \in \Gamma_n \cap \HDC}(\rho_S) \leq \log \bspace,
    $$ 
    in which case we can simply store the $S$-views within space $\polyO(\bspace)$. Otherwise, the feasibility of \eqref{OPT-maximin} is guaranteed. However, the na\"ive algorithm has the following drawbacks in terms of practicality: $(1)$ the exhaustive splitting steps can incur a large poly-logarithmic factor; $(2)$ for every subproblem, we need to run $\PANDA$ from scratch with a new instance, $(3)$ the space-time tradeoff obtained is hard to interpret. In the following sections, we introduce the 2-phase $\PANDA$ algorithm, called $\twoPP$, that also attains the intrinsic tradeoff as specified in~\eqref{OPT-maximin}, while efficiently addressing these two drawbacks and it obtains much practical/interpretable intrinsic tradeoff(s).

\section{ The $\twoPP$ algorithm}
\label{sec:flow}
  
 The 2-phase $\PANDA$ ($\twoPP$) algorithm, similar to $\PANDA$, is built on a class of inequalities called the joint Shannon-flow inequalities. In this section, we first present some background on Shannon-flow inequalities, and then our extension to joint inequalities. Finally, we present the $\twoPP$ algorithm.
 
\subsection{Shannon-flow Inequalities} 

The following inequality
 \begin{align}\label{eq:shannon-flow-inequality}
      \sum_{X \subset Y \subseteq [n]} \delta_{Y|X} \cdot h(Y|X) \defeq \sum_{X \subset Y \subseteq [n]} \delta_{Y|X} \cdot (h(Y) - h(X)) \geq \sum_{\emptyset \neq Z \subseteq [n]} \lambda_{Z|\emptyset} \cdot h(Z|\emptyset),
 \end{align}
 is called a {\em Shannon-flow inequality} if it holds for any polymatroid function $h \in \Gamma_n$ and all $\delta_{Y|X}$ and $\lambda_{Z|\emptyset}$ are some non-negative rational coefficients, i.e., $\delta_{Y|X}, \lambda_{Z|\emptyset} \in \bQ_+$. To concisely represent Shannon-flow inequalities, we define the conditional polymatroid as in~\cite{DBLP:conf/pods/Khamis0S17}.
 
 \introparagraph{Conditional Polymatroids} \label{def:condPolymatroid}
  Let $\mC \subseteq 2^{[n]} \times 2^{[n]}$ denote the set of all pairs $(X, Y)$ such that $\emptyset \subseteq X \subset Y \subseteq [n]$. A vector $\bh \in \bR^{\mC}_+$ has coordinates indexed by pairs $(X, Y) \in \mC$ and we denote the corresponding coordinate value of $\bh$ by $h(Y|X)$. A vector $\bh$ is called a \textit{conditional polymatroid} if and only if there is a polymatroid $h$ such that $h(Y|X) = h(Y) - h(X)$; and, we say that the polymatroid $h$ defines the conditional polymatroid $\bh$. In particular, $\bh = (h(Y|X))_{(X, Y) \in \mC}$.

  If for \eqref{eq:shannon-flow-inequality}, we define $\lambda_{Z|X} = 0$ for any $\emptyset \neq X \subset Z \subseteq [n]$, then each of $\{\delta_{Y|X}\}, \{\lambda_{Y|X}\}$ can be interpreted as vectors over $(X, Y)$ pairs, where $\emptyset \subseteq X \subset Y \subseteq [n]$. We denote them as $\bdelta, \blambda$, respectively. Thus, the Shannon-flow inequality \eqref{eq:shannon-flow-inequality} can be re-written into an inequality on conditional polymatroids, i.e. the $\bQ_+^{\mC}$ space, as $\langle \bdelta, \bh \rangle \geq \langle \blambda, \bh \rangle$, where $\langle \cdot, \cdot \rangle$ denotes dot product.

 \introparagraph{Proof Sequences} One of the major contributions from \cite{DBLP:conf/pods/Khamis0S17} says that any Shannon-flow inequality $\langle \bdelta, \bh \rangle \geq \langle \blambda, \bh \rangle$ can be proved by just applying the following $4$ rules:
\begin{align*}
    \textit{(R1)} & \textit{ submodularity rule} & h(I \cup J | J) - h(I |I \cap J) & \leq 0, \quad I \perp J  \\
    \textit{(R2)} & \textit{ monotonicity rule}  & -h(Y|\emptyset) + h(X|\emptyset) & \leq 0, \quad X \subset Y \\
    \textit{(R3)} & \textit{ composition rule} & h(Y|\emptyset) - h(Y|X) - h(X|\emptyset) & \leq 0, \quad X \subset Y  \\
    \textit{(R4)} & \textit{ decomposition rule}  & - h(Y|\emptyset) + h(Y|X) + h(X|\emptyset) & \leq 0, \quad X \subset Y 
\end{align*}
 where $I \perp J$ means $I \nsubseteq J$ and $J \nsubseteq I$. $(R1)$ and $(R2)$ come exactly from the submodularity and monotonicity properties of polymatroids. $(R3)$ and $(R4)$ simply follow from the definition of $h(Y|X) = h(Y) - h(X)$. All the rules can also be vectorized over all $(X, Y)$ pairs, where $\emptyset \subseteq X \subset Y \subseteq [n]$. For every $I \perp J$, we define a vector $\bs_{I, J}$, and for every $X \subset Y$, we define three vectors $\bm_{X, Y}, \bc_{X, Y}, \bd_{Y, X}$ such that the linear rules above can be written using dot-products:
 \begin{align*}
    \textit{(R1)} & \textit{ submodularity rule} & \langle \bs_{I, J}, \bh \rangle & \leq 0, \quad I \perp J \\
    \textit{(R2)} & \textit{ monotonicity rule}  & \langle \bm_{X, Y}, \bh \rangle & \leq 0, \quad X \subset Y \\
    \textit{(R3)} & \textit{ composition rule} & \langle \bc_{X, Y}, \bh \rangle & \leq 0, \quad X \subset Y  \\
    \textit{(R4)} & \textit{ decomposition rule}  & \langle \bd_{Y, X}, \bh \rangle & \leq 0, \quad X \subset Y 
\end{align*}
 A {\em proof sequence} of a Shannon-flow inequality $\langle \bdelta, \bh \rangle \geq \langle \blambda, \bh \rangle$ is a sequence $\proofseq \defeq (w_1 \proofstep_1, w_2 \proofstep_2, \ldots, w_\ell \proofstep_\ell)$ of length $\ell$ satisfying all of the following:
 \begin{enumerate}
     \item [(1)] $\proofstep_i \in \{\bs_{I, J}, \bm_{X, Y}, \bc_{X, Y}, \bd_{Y, X}\}$ for every $i \in [\ell]$, called a {\em proof step};
     \item [(2)] $w_i \in \bQ_+$ for every $i \in [\ell]$, called the {\em weight} of the proof step $\proofstep_i$;
     \item [(3)] the vectors $\bdelta_0 \defeq \bdelta, \bdelta_1, \ldots, \bdelta_\ell$ defined by $\bdelta_i = \bdelta_{i-1} + w_i \cdot \proofstep_i$ are non-negative rational vectors;
     \item [(4)] $\bdelta_\ell \geq \blambda$ (element-wise comparison).
 \end{enumerate}
 A proof sequence $(w_1 \proofstep_1, w_2 \proofstep_2, \ldots, w_\ell \proofstep_\ell)$ implies the following inequalities for all polymatroids $h \in \Gamma_n$,
 $$ \langle \bdelta, \bh \rangle = \langle \bdelta_0, \bh \rangle \geq \cdots \geq \langle \bdelta_\ell, \bh \rangle  \geq \langle \blambda, \bh \rangle,
 $$
 which provides a step-by-step proof for the Shannon-flow inequality $\langle \bdelta, \bh \rangle \geq \langle \blambda, \bh \rangle$. The following theorem (also Theorem 2 in~\cite{DBLP:conf/pods/Wang022}) is a direct result from Theorem B.12 and Proposition B.13 in~\cite{DBLP:conf/pods/Khamis0S17}.
 \begin{theorem}[\cite{DBLP:conf/pods/Khamis0S17}] \label{thm:prooseqfinite}
  For any Shannon-flow inequality $\langle \bdelta, \bh \rangle \geq \langle \blambda, \bh \rangle$ such that $\|\blambda\|_1 = 1$, there is a proof sequence of length $O(poly(2^n))$.
 \end{theorem}
  \autoref{thm:prooseqfinite} implies that there is a proof sequence for a Shannon-flow inequality that is exponentially long in the number of variables $n$, but it is of constant length under fixed query sizes.

 \subsection{Joint Shannon-flow Inequalities}
 
 To motivate our study of joint Shannon-flow inequalities, we first give a characterization of the optimal objective value of the maximin optimization problem \eqref{OPT-maximin}, $\OBJ(S)$, through the following class of LPs
  \begin{equation}\label{OPT-maxminLPIneq}
     \bL(\blambda_{\sfBT}, \btheta_{\sfBS}, S) \defeq \max_{ \substack{ \hS \in \HDC,  \hT \in \HDC \cap \HWC \\ (\hS, \hT) \in (\Gamma_n \times \Gamma_n) \cap \HSC} } \quad  \sum_{B \in \sfBT} \lambda_B \cdot \hT(B) + \sum_{B \in \sfBS} \theta_B \cdot \hS(B) - (\log \bspace) \cdot \|\btheta_{\sfBS}\|_1, 
   \end{equation}
   parameterized by two vectors, $\blambda_{\sfBT} \defeq (\lambda_B)_{B \in \sfBT} \in {\bQ_+^{\sfBT}}$ with $\|\blambda_{\sfBT}\|_1 = 1$ and $\btheta_{\sfBS} \defeq (\theta_B)_{B \in \sfBS} \in {\bQ_+^{\sfBS}}$. Equivalently, any 
  \begin{align}\label{weightsSF}
      (\blambda_{\sfBT}, \btheta_{\sfBS}) \in \left\{(\ba_T,\ba_S ) \mid \ba_T \in \bQ_+^{\sfBT}, \ba_S \in \bQ_+^{\sfBS}, \|\ba_T\|_1 = 1 \right\} 
  \end{align}
  gives rise to an LP of the form \eqref{OPT-maxminLPIneq} with optimal objective value $\bL(\blambda_{\sfBT}, \btheta_{\sfBS}, S)$. 
  
\begin{lemma}\label{lem:maxminLP}
  Let $\bspace$ be a fixed quantity. For any $(\blambda_{\sfBT}, \btheta_{\sfBS})$ as specified in \eqref{weightsSF}, the optimal objective value of \eqref{OPT-maximin}, $\OBJ(S)$ satisfies,
  $$ \OBJ(S) \leq \bL(\blambda_{\sfBT}, \btheta_{\sfBS}, S)
  $$
  Moreover, assuming that $\OBJ(S)$ is positive and bounded, there is an optimal $(\blambda^*_{\sfBT}, \btheta^*_{\sfBS})$ satisfying \eqref{weightsSF} such that 
  $$
  \OBJ(S) = \bL(\blambda^*_{\sfBT}, \btheta^*_{\sfBS}, S).
  $$
 \end{lemma}
 \noindent Instead of proving \autoref{lem:maxminLP} directly, we prove a slightly more general lemma (Lemma~\ref{lem:maxminLPGeneral}) that may be of independent interest.
  \begin{lemma}\label{lem:maxminLPGeneral}
   Let $\bA \in \bQ^{\ell \times m}, \vectorb \in \bR^{\ell},  \bD \in \bQ_+^{m \times \highlight{q}}$, $\bC \in \bQ_+^{m \times p}$ be a matrix with columns $\bc_1, \ldots, \bc_p$ and polyhedron $P = \{\bx \in \bR^{m} \mid \bA \bx \leq \vectorb, \bx \geq \bzero\}$. Let $w^*$ be the optimal objective value of the following optimization problem (assume $\bspace$ as a fixed quantity)
   \begin{equation}
       \begin{aligned}
          \max_{\bx \in P} \quad & \min_{k \in [p]} \bc_k^{\top} \bx \\
          \quad  \text{s.t. } \quad & \bD^{\top} \bx \geq \bone_{\highlight{q}} \log \bspace.
        \end{aligned}
        \label{OPT-maximinGeneral}
   \end{equation}
 If $w^*$ is positive and bounded, then for any vectors $\bz, \bu \in \bQ_+^p$ with $\|\bz\|_1 = 1$, the following linear program:
    \begin{equation}
     \begin{aligned}
      L(\bz, \bu) \defeq \max_{\bx \in P} \quad & (\bC \bz)^{\top} \bx + (\bD^{\top} \bx - \bone_{\highlight{q}} \log \bspace)^{\top} \bu
      \label{OPT-alternativeLP}
     \end{aligned}
    \end{equation}
 satisfies $w^* \leq L(\bz, \bu)$. In particular, there is a pair of vectors $\bz^*  \in \bQ_+^p, \bu^*  \in \bQ_+^{\highlight{q}}$ with $\|\bz^*\|_1 = 1$ such that $w^* = L(\bz^*, \bu^*)$.
  \end{lemma}
 
  \begin{proof} 
  First, we introduce $\bu \in \bQ_+^{\highlight{q}}$ as the Lagrange multiplier for \eqref{OPT-maximinGeneral} and obtain
  \begin{align*}
      w^* & \leq \max_{\bx \in P} \min_{\bu \in \bQ_+^{\highlight{q}}} \quad  \min_{k \in [p]} \left( \bc_k^{\top} \bx + (\bD^{\top} \bx - \bone_{\highlight{q}} \log \bspace)^{\top} \bu \right) \\
      & \leq  \max_{\bx \in P} \min_{\bu \in \bQ_+^{\highlight{q}}, \bz \in \bQ_+^p, \|\bz\|_1 = 1} \quad (\bC \bz)^{\top} \bx + (\bD^{\top} \bx - \bone_{\highlight{q}} \log \bspace)^{\top} \bu \\
      & \leq \min_{\bu \in \bQ_+^{\highlight{q}}, \bz \in \bQ_+^p, \|\bz\|_1 = 1} \max_{\bx \in P} \quad (\bC \bz)^{\top} \bx + (\bD^{\top} \bx - \bone_{\highlight{q}} \log \bspace)^{\top} \bu \\
      & =  \min_{\bu \in \bQ_+^{\highlight{q}}, \bz \in \bQ_+^p, \|\bz\|_1 = 1} L(\bz, \bu)
  \end{align*}
  where the third equality is because of the Minimax Inequality. Therefore, we have shown that for any vectors $\bu \in \bQ_+^{\highlight{q}}, \bz \in \bQ_+^p$ with $\|\bz\|_1 = 1$, $L(\bz, \bu) \geq w^*$. Next, we show that there are vectors $\bu^* \in \bQ_+^{\highlight{q}}, \bz^* \in \bQ_+^p$ with $\|\bz^*\|_1 = 1$ such that $L(\bz^*, \bu^*) = w^*$. We first re-write \eqref{OPT-maximin} as the following equivalent linear program:
   \begin{equation}
       \begin{aligned}
          \max_{\bx, w} \quad & w\\
          \quad  \text{s.t. } \quad & \bA \bx \leq \vectorb \\
          & \bD^{\top} \bx \geq \bone_{\highlight{q}} \log \bspace \\
          & \bC^{\top} \bx \geq \bone_p w \\
          & \bx \geq \bzero, w \geq 0
       \end{aligned}
       \label{OPT-LP}
   \end{equation}
 Let $(w^*, \bx^*)$ be an optimal solution for \eqref{OPT-LP}. Then, the dual of \eqref{OPT-LP} can be written as the following linear program:
 \begin{equation}
    \begin{aligned}
          \min_{\by, \bu, \bz} \quad & \vectorb^{\top} \by - (\bone_{\highlight{q}} \log \bspace)^{\top} \bu \\
          \quad  \text{s.t. } \quad & \bA^{\top} \by - \bD \bu - \bC \bz \geq \bzero \\
          & \bone_p^{\top} \bz \geq 1 \\
          & \by, \bu, \bz \geq \bzero
     \end{aligned}   
     \label{OPT-LPDual}
 \end{equation}
 Let $(\by^*, \bu^*, \bz^*)$ be an extreme point of the (rational) dual polyhedron that attains the optimal objective value for \eqref{OPT-LPDual}, so $\bz^* \in \bQ^p_+, \bu^* \in \bQ^{\highlight{q}}_+$. The complementary slackness conditions of the \eqref{OPT-LP} and \eqref{OPT-LPDual} primal-dual pair and the assumption $w^* > 0$ imply that $\bone_p^{\top} \bz^* = \|\bz^*\|_1 = 1$, $(\bone_p w^* - \bC^{\top} \bx^*)^{\top} \bz^* = 0$ and $(\bD^{\top} \bx^* - \bone_{\highlight{q}} \log \bspace)^{\top} \bu^* = 0$. 
 We then show that $L(\bz^*, \bu^*) = w^*$. First, we note that $\bx^*$ is feasible for \eqref{OPT-alternativeLP} with objective value $(\bC \bz^*)^{\top} \bx^* + (\bD^{\top} \bx^* - \bone_{\highlight{q}} \log \bspace)^{\top} \bu^* = (\bC \bz^*)^{\top} \bx^* = (\bC^{\top} \bx^*)^{\top} \bz^* = (\bone_p w^*)^{\top} {\bz}^* = w^*$. Furthermore, for any feasible $\bx$ to \eqref{OPT-alternativeLP}, we have that
 \begin{align*}
     (\bC \bz^*)^{\top} \bx + (\bD^{\top} \bx - \bone_{\highlight{q}} \log \bspace)^{\top} \bu^* & = (\bC \bz^*)^{\top} \bx + (\bD \bu^*)^{\top} \bx -  (\bone_{\highlight{q}} \log \bspace)^{\top} \bu^* && \textit{re-arrange} \\
     & \leq (\bA^{\top} \by^*)^{\top} \bx -  (\bone_{\highlight{q}} \log \bspace)^{\top} \bu^* && \textit{dual feasibility and non-negativity} \\
     & = (\bA \bx)^{\top} \by^* -  (\bone_{\highlight{q}} \log \bspace)^{\top} \bu^* \\
      & \leq \vectorb^{\top} \by^* -  (\bone_{\highlight{q}} \log \bspace)^{\top} \bu^* && \textit{primal feasibility} \\
      & = w^* && \textit{strong duality}
 \end{align*}
 This implies that $L(\bz^*, \bu^*) \leq w^*$. Together, we get $L(\bz^*, \bu^*) = w^*$.
 \end{proof}
 \highlight{Now, \autoref{lem:maxminLP} is a direct corollary of Lemma~\ref{lem:maxminLPGeneral} by setting $\blambda_{\sfBT}$ as $\bz^*$, $\btheta_{\sfBS}$ as $\bu^*$.} 
 
 \eat{\hangdong{As a side note, this proof governs that $\blambda_{\sfBT}$ and $\btheta_{\sfBS}$ have bounded denominators. $\by$ here is exactly all dual variables $(\bmu, \bsigma, \bdelta, \bgamma)$ for the upcoming dual formulation, which also has bounded denominators using Cramer's Lemma. Thus, we can apply Theorem B.12 in PANDA to guarantee that the proof sequences we are to obtain are of finite length. Precisely, amongst all vars in $\by$, $(\bmu_S, \bsigma_S, \bdelta_S, \bgamma_S)$ and $\btheta_{\sfBS}$ will be used to construct $\proofseq(S)$; $(\bmu_T, \bsigma_T, \bdelta_T, \bgamma_T)$ and $\blambda_{\sfBT}$ will be used to construct $\proofseq(T)$. 
 }}

 In the remainder of the section, we implicitly assume that $\bspace$ is a fixed quantity. Moreover, we fix a pair  $(\blambda_{\sfBT}, \btheta_{\sfBS})$ satisfying \eqref{weightsSF}. We can now re-write \eqref{OPT-maxminLPIneq}, listing out all its constraints and ignore the constant factor $(\log \bspace) \cdot \|\btheta_{\sfBS}\|_1$ :
 \begin{equation}
     \begin{aligned}
          \ell(\blambda_{\sfBT}, \btheta_{\sfBS}) \defeq \max_{\hS, \hT} \quad & \sum_{B \in \sfBT} \lambda_B \cdot \hT(B) + \sum_{B \in \sfBS} \theta_B \cdot \hS(B) \\
         \text{s.t.} \quad & \hS(Y) - \hS(X) \leq n_{Y|X}, && (X, Y, N_{Y|X}) \in \DC \\
          \quad  \quad & \hT(Y) - \hT(X) \leq n_{Y|X}, && (X, Y, N_{Y|X}) \in \DC \cup \AC \\
         \quad  \quad &  \hS(I \cup J|J) - \hS(I | I \cap J) \leq 0, && I \perp J \\
         \quad  \quad &  \hT(I \cup J|J) - \hT(I | I \cap J) \leq 0, && I \perp J \\
         \quad  \quad &  \hS(X) - \hS(Y) \leq 0, && \emptyset \neq X \subset Y \subseteq [n] \\
         \quad  \quad &  \hT(X) - \hT(Y) \leq 0, && \emptyset \neq X \subset Y \subseteq [n] \\
         \quad  \quad & \hS(X) + \hT(Y|X) \leq n_{Z|\emptyset}, && (X, Y|X, N_{Z|\emptyset}) \in \SC \\
         \quad  \quad & \hS(Y|X) + \hT(X) \leq n_{Z|\emptyset}, && (X, Y|X, N_{Z|\emptyset}) \in \SC \\
          \quad  \quad & \hS(Z) \geq 0, \quad \hT(Z) \geq 0, && \emptyset \neq Z \subseteq [n] 
     \end{aligned}
     \label{OPT-DetailedLP}
 \end{equation}
 Recall that implicitly we have $\hS(\emptyset) = \hT(\emptyset) =0$ and that $n_{Y|X} = \log N_{Y|X}$. 
 Then we write down the dual LP for \eqref{OPT-DetailedLP}. We associate a dual variable $(\delta_S)_{Y|X}$ to $\hS(Y) - \hS(X) \leq n_{Y|X}$ for each $(X, Y, N_{Y|X}) \in \DC$ and a dual variable $(\delta_T)_{Y|X}$ to $\hT(Y) - \hT(X) \leq n_{Y|X}$ for each $(X, Y, N_{Y|X}) \in \DC \cup \AC$. For each $I \perp J$, where $I, J \subseteq [n]$, we associate a dual variable $(\sigma_S)_{I, J}$ to the submodularity constraint of $\hS$ and $(\sigma_T)_{I, J}$ to the submodularity constraint of $\hT$. For each $\emptyset \neq X \subset Y \subseteq [n]$, we associate a dual variable $(\mu_S)_{X, Y}$ to the monotonicity constraint of $\hS$ and a dual variable $(\mu_T)_{X, Y}$ to the monotonicity constraint of $\hT$. Lastly, for each $(X, Y|X, N_{Z|\emptyset}) \in \SC$, we associate a dual variable $\gamma_{X, Y|X}$ to $\hS(X) + \hT(Y) - \hT(X) \leq N_{Z|\emptyset}$ and a dual variable $\gamma_{Y|X, X}$ to $\hS(Y) - \hS(X)+ \hT(X) \leq N_{Z|\emptyset}$. Moreover, we extend vectors $(\lambda_B)_{B \in \sfBT}$ and $(\theta_B)_{B \in \sfBS}$ to every $Z \in 2^{[n]}$ in the obvious way:
 $$
 \lambda_Z \defeq \begin{cases}\lambda_{B} & \text { when } Z = B \in \sfBT \\ 0 & \text { otherwise } \end{cases} \qquad 
 \theta_Z \defeq \begin{cases} \theta_{B} & \text { when } Z = B \in \sfBS \\ 0 & \text { otherwise }\end{cases}
 $$
 Abusing notations, we write $(X, Y) \in \DC$ whenever $(X, Y, N_{Y|X}) \in \DC$ and $(X, Y|X) \in \SC$ whenever $(X, Y, N_{Z|\emptyset}) \in \SC$. Note that by maintaining the best constraint assumption, we can always recover the only $N_{Y|X}$ or $N_{Z|\emptyset}$ from a given $(Y, X)$-pair. The dual of \eqref{OPT-DetailedLP} can now be written as
 \begin{equation}
     \begin{aligned}
           \min \quad & \sum_{(X, Y) \in \DC} n_{Y|X} \cdot (\delta_S)_{Y|X} + \sum_{(X, Y) \in \DC \cup \AC} n_{Y|X} \cdot (\delta_T)_{Y|X} 
           \\ & \qquad + \sum_{(X, Y|X) \in \SC}  n_{Z|\emptyset} \cdot (\gamma_{X, Y|X} + \gamma_{Y|X, X})\\ 
         \text{s.t.} \quad & \quad \inflow_S(Z) \geq \theta_Z & \emptyset \neq Z \subseteq [n]  \\
         & \quad \inflow_T(Z) \geq \lambda_Z & \emptyset \neq Z \subseteq [n] \\
         & \quad (\delta_S)_{Y|X}, (\mu_S)_{X, Y}, (\sigma_S)_{I, J} \geq 0 & \\
         & \quad (\delta_T)_{Y|X}, (\mu_T)_{X, Y}, (\sigma_T)_{I, J} \geq 0  & \\
         & \quad \gamma_{X, Y|X}, \gamma_{Y|X, X} \geq 0 & 
     \end{aligned}
     \label{OPT-DetailedDualLP}
 \end{equation}
 where for each $\emptyset \neq Z \subseteq [n]$,
 \begin{align*}
        \textcolor{black}{\inflow_S(Z)} \defeq \quad & \left(\sum_{X:(X, Z) \in \DC} \textcolor{black}{(\delta_S)_{Z \mid X}}-\sum_{Y:(Z, Y) \in \DC} \textcolor{black}{(\delta_S)_{Y \mid Z}}\right) + 
        \left( -\sum_{X: X \subset Z} \textcolor{black}{(\mu_S)_{X, Z}}+\sum_{Y: Z \subset Y} \textcolor{black}{(\mu_S)_{Z, Y}} \right) \\
        & + \left( \sum_{\substack{I \perp J \\ I \cap J=Z}} \textcolor{black}{(\sigma_S)_{I, J}} + \sum_{\substack{I \perp J \\ I \cup J=Z}} \textcolor{black}{(\sigma_S)_{I, J}} - \sum_{J: J \perp Z} \textcolor{black}{(\sigma_S)_{Z, J}} \right)  \\
        & + \left( \sum_{\substack{Z:(Z, Y|Z) \in \SC \\ Z \subset Y}}  
        \gamma_{Z, Y|Z} - \sum_{\substack{Z:(Z, Y|Z) \in \SC \\ Z \subset Y}}   \gamma_{Y|Z, Z} + \sum_{\substack{Z:(X, Z|X) \in \SC \\ X \subset Z}}   \gamma_{Z|X, X}\right) 
 \end{align*}
 \begin{align*}
        \textcolor{black}{\inflow_T(Z)} \defeq \quad & \left(\sum_{X:(X, Z) \in \DC \cup \AC} \textcolor{black}{(\delta_T)_{Z \mid X}}
        -\sum_{Y:(Z, Y) \in \DC \cup \AC } \textcolor{black}{(\delta_T)_{Y \mid Z}}\right) +
        \left( -\sum_{X: X \subset Z} \textcolor{black}{(\mu_T)_{X, Z}}+\sum_{Y: Z \subset Y} \textcolor{black}{(\mu_T)_{Z, Y}} \right) \\
        & + \left( \sum_{\substack{I \perp J \\ I \cap J=Z}} \textcolor{black}{(\sigma_T)_{I, J}} + \sum_{\substack{I \perp J \\ I \cup J=Z}} \textcolor{black}{(\sigma_T)_{I, J}} - \sum_{J: J \perp Z} \textcolor{black}{(\sigma_T)_{Z, J}} \right)  \\
        & + \left( \sum_{\substack{Z:(Z, Y|Z) \in \SC \\ Z \subset Y}}  
        \gamma_{Y|Z, Z} - \sum_{\substack{Z:(Z, Y|Z) \in \SC \\ Z \subset Y}}   \gamma_{Z, Y|Z} + \sum_{\substack{Z:(X, Z|X) \in \SC \\ X \subset Z}}   \gamma_{X, Z|X}\right) 
\end{align*}
 
 Next, we introduce the \textit{joint Shannon-flow inequalities}.
 \begin{definition}[Joint Shannon-flow Inequality]
 The inequality
   \begin{equation}
     \begin{aligned}\label{jointSF-full}
           \sum_{(X, Y) \in \DC} \hS(Y|X) \cdot (\delta_S)_{Y|X} & +  \sum_{(X, Y) \in \DC \cup \AC} \hT(Y|X) \cdot (\delta_T)_{Y|X} + \sum_{(X, Y|X) \in \SC}  (\hS(X) + \hT(Y|X)) \cdot \gamma_{X, Y|X} \\
      & + \sum_{(X, Y|X) \in \SC}  (\hS(Y|X) + \hT(X)) \cdot \gamma_{Y|X, X} \geq \sum_{B \in \sfBS} \theta_B \cdot \hS(B) + \sum_{B \in \sfBT} \lambda_B \cdot \hT(B),
     \end{aligned}
 \end{equation}
 is called a \textit{joint Shannon-flow inequality} if it holds for all $(\hS, \hT) \in \Gamma_n \times \Gamma_n$ and all coefficients are non-negative rational numbers.
 \end{definition}
 
  By assigning either polymatroid to be always $0$, the above inequality implies the following two Shannon-flow inequalities, called  the {\textit{participating Shannon-flow inequalities}}.
 \begin{align}
     \sum_{(X, Y) \in \DC} \hS(Y|X) \cdot (\delta_S)_{Y|X} + \sum_{(X, Y|X) \in \SC}  \hS(X) \cdot \gamma_{X, Y|X} + \sum_{(X, Y|X) \in \SC}  \hS(Y|X) \cdot \gamma_{Y|X, X} & \geq \sum_{B \in \sfBS} \theta_B \cdot \hS(B) \label{SF-S} \\
     \sum_{(X, Y) \in \DC \cup \AC} \hT(Y|X) \cdot (\delta_T)_{Y|X} + \sum_{(X, Y|X) \in \SC}  \hT(Y|X) \cdot \gamma_{X, Y|X} + \sum_{(X, Y|X) \in \SC}  \hT(X) \cdot \gamma_{Y|X, X} & \geq \sum_{B \in \sfBT} \lambda_B \cdot \hT(B) \label{SF-T}
 \end{align}
 
%Now, define a collection of non-negative rational coefficients of the form:
%\begin{equation}\label{definingCollection}
%    \begin{aligned}
%          \coeff \defeq & \left\{\lambda_B \mid {B \in \sfBS}\right\} \cup \left\{\theta_B \mid {B \in \sfBT}\right\} \cup \left\{(\delta_S)_{Y|X} \mid (X, Y) \in \DC \cup \AC \right\} \cup \left\{(\delta_T)_{Y|X} \mid (X, Y) \in \DC \right\} \\
%            & \cup \left\{\gamma_{X, Y|X} \mid  (X, Y|X) \in \SC \right\} \cup \left\{\gamma_{Y|X, X} \mid  (X, Y|X) \in \SC \right\}
%    \end{aligned}
%\end{equation}
% that supplies coefficients for the following two inequalities:
 
 It will be convenient to write a joint Shannon-flow inequality as inequalities over conditional polymatroid. The polymatroid $\hS$ defines a conditional polymatroid $\bhS$ and the polymatroid $\hT$ defines a conditional polymatroid $\bhT$. More precisely, we define the vectors $\blambda, \btheta \in \bQ^\mC_+$ (extend to $(X, Y)$ pairs where $\emptyset \subseteq X \subset Y \subseteq [n]$) with coordinate values assigned as the following:
 $$
 \lambda(Y|X) \defeq \begin{cases}\lambda_{B} & \text { when } Y = B, X = \emptyset \\ 0 & \text { otherwise } \end{cases} \qquad 
 \theta(Y|X) \defeq \begin{cases} \theta_{B} & \text { when } Y = B, X = \emptyset \\ 0 & \text { otherwise }\end{cases}
 $$
 and similarly, we define the vectors $\bdelta_S, \bdelta_T \in \bQ^\mC_+$ with coordinate values:
 $$
 \delta_S(Y|X) \defeq \begin{cases} (\delta_S)_{Y|X} & \text { when } (X, Y) \in \DC \\ 0 & \text { otherwise } \end{cases} \qquad 
 \delta_T(Y|X) \defeq \begin{cases} (\delta_T)_{Y|X} & \text { when } (X, Y) \in \DC \cup \AC \\ 0 & \text { otherwise }\end{cases}
 $$
 Lastly, the coefficients $\left\{\gamma_{X, Y|X} \mid  (X, Y|X) \in \SC \right\} \cup \left\{\gamma_{Y|X, X} \mid  (X, Y|X) \in \SC \right\} $ contributes to the coefficients of both inequalities \eqref{SF-S} and \eqref{SF-T}. We define a pair of vectors $\bgamma_S, \bgamma_T \in \bQ^\mC_+$ as the following:
 \begin{align*}
      \gamma_S(Y|X) & \defeq \begin{cases} \gamma_{U, V|U} & \text { when } Y = U, X = \emptyset, (U, V|U) \in \SC  \\ \gamma_{V|U, U} & \text { when } Y = V, X = U, (U, V|U) \in \SC
      \\ 0 & \text { otherwise } \end{cases} 
      \\
      \\
      \gamma_T(Y|X) & \defeq \begin{cases} \gamma_{U, V|U} & \text { when } Y = V, X = U, (U, V|U) \in \SC \\ 
      \gamma_{V|U, U} & \text { when } Y = U, X = \emptyset, (U, V|U) \in \SC
      \\ 0 & \text { otherwise }\end{cases}
 \end{align*}
 for tracking the contributions to $\bhS, \bhT$, respectively. Let $\bg_S \defeq \bdelta_S + \bgamma_S$ and $\bg_T \defeq \bdelta_T + \bgamma_T$, then the two inequalities \eqref{SF-S} and \eqref{SF-T} can be re-written using dot-products:
 \begin{align}
     \langle \bg_S, \bhS \rangle \defeq \langle \bdelta_S, \bhS \rangle +  \langle \bgamma_S, \bhS \rangle  & \geq  \langle \btheta, \bhS \rangle \label{SF-S-Vector}\\
     \langle \bg_T, \bhT \rangle \defeq \langle \bdelta_T, \bhT \rangle +  \langle \bgamma_T, \bhT \rangle & \geq  \langle \blambda, \bhT \rangle \label{SF-T-Vector}
 \end{align}
 Moreover, the joint Shannon-flow inequality \eqref{jointSF-full} can be written as
 \begin{equation}\label{jointSF-Vector}
      \begin{aligned}
      \langle \bg_S, \bhS \rangle + \langle \bg_T, \bhT \rangle 
      & \geq  \langle \btheta, \bhS \rangle + \langle \blambda, \bhT \rangle
      \end{aligned}
  \end{equation}

%Applying Proposition 5.6 in \cite{DBLP:conf/pods/Khamis0S17}, we provide a characterization of \eqref{jointSF-Vector} being a joint Shannon-flow inequalities, stated in the following theorem:

 \begin{theorem}\label{thm:witness}
The inequality \eqref{jointSF-Vector} is a joint Shannon-flow inequality if and only if there are non-negative vectors of rationals $\bsigma_S, \bmu_S, \bsigma_T, \bmu_T$ such that all constraints of the dual \eqref{OPT-DetailedDualLP} are satisfied. In particular, we call $(\bsigma_S, \bmu_S, \bsigma_T, \bmu_T)$ a {\textit{witness}} for the joint Shannon-flow inequality.
 \end{theorem}
 
 \begin{proof}
  Proposition 5.6 in \cite{DBLP:conf/pods/Khamis0S17} states that: \eqref{SF-S-Vector} is a Shannon-flow inequality if and only if there is a $(\bsigma_S, \bmu_S) \geq \bzero$ (called a witness in \cite{DBLP:conf/pods/Khamis0S17}) such that $(\bg_S, \bsigma_S, \bmu_S)$ satisfies the set of constraints:
  $$
 \left\{\inflow_S(Z) \geq \theta_Z \mid \emptyset \neq Z \subseteq [n] \right\};
  $$
  similarly, \eqref{SF-T-Vector} is a Shannon-flow inequality if and only if there is a (witness) $(\bsigma_T, \bmu_T) \geq \bzero$ such that $(\bg_T, \bsigma_T, \bmu_T)$ satisfies the set of constraints:
 $$
 \left\{\inflow_T(Z) \geq \lambda_Z \mid \emptyset \neq Z \subseteq [n] \right\}.
 $$
 Recall the formulation of \eqref{OPT-DetailedDualLP}, these two sets of constraints form exactly all the constraints in \eqref{OPT-DetailedDualLP}. Thus, \eqref{jointSF-Vector} is a joint Shannon-flow inequality if and only if there is a $(\bsigma_S, \bmu_S, \bsigma_T, \bmu_T) \geq \bzero$ such that $(\bg_S, \bsigma_S, \bmu_S, \bg_T, \bsigma_T, \bmu_T)$ satisfies all the constraints of the dual \eqref{OPT-DetailedDualLP}.
 \end{proof}

Theorem~\ref{thm:witness} implies that a feasible solution of the dual \eqref{OPT-DetailedDualLP}, and in particular the component
$$\left\{(\delta_S)_{Y|X} \mid (X, Y) \in \DC \right\} \cup \left\{(\delta_T)_{Y|X} \mid (X, Y) \in \DC \right\} \cup \left\{\gamma_{X, Y|X} \mid  (X, Y|X) \in \SC \right\} \cup \left\{\gamma_{Y|X, X} \mid  (X, Y|X) \in \SC \right\}$$
in conjunction with $(\blambda_{\sfBT}, \btheta_{\sfBS})$, defines a joint Shannon-flow inequality~\eqref{jointSF-full}. The component $(\bsigma_S, \bmu_S, \bsigma_T, \bmu_T)$ of the dual, by \autoref{thm:witness}, is a witness for the joint Shannon-flow inequality. Note that the joint Shannon-flow inequality implies that
 $$
 \langle \btheta, \bhS \rangle + \langle \blambda, \bhT \rangle \leq \sum_{(X, Y) \in \DC} n_{Y|X} \cdot (\delta_S)_{Y|X} + \sum_{(X, Y) \in \DC \cup \AC} n_{Y|X} \cdot (\delta_T)_{Y|X} + \sum_{(X, Y|X) \in \SC}  n_{Z|\emptyset} \cdot (\gamma_{X, Y|X} + \gamma_{Y|X, X})
 $$
 where the right-hand side is exactly $\ell(\blambda_{\sfBT}, \btheta_{\sfBS})$ by taking the optimal solution of the dual \eqref{OPT-DetailedDualLP} (by strong duality).
 We have established that: for an arbitrary $(\blambda_{\sfBT}, \btheta_{\sfBS})$ satisfying \eqref{weightsSF}, we can construct a joint Shannon-flow inequality,
  $$
  \langle \bg_S, \bhS \rangle + \langle \bg_T, \bhT \rangle \geq  \langle \btheta, \bhS \rangle + \langle \blambda, \bhT \rangle,
  $$
 having a witness $(\bsigma_S, \bmu_S, \bsigma_T, \bmu_T)$, such that its implied upper bound coincides with $\ell(\blambda_{\sfBT}, \btheta_{\sfBS})$. 
 
 \subsection{A Brief Review/Augmentation of the $\PANDA$ Algorithm}
 This section provides a brief review of the $\PANDA$ algorithm. Let $\rho$ be a disjunctive rule~\eqref{def:disjunctiveRule} under degree constraints $\DC$. At a high level, $\PANDA$ does the following:
  \begin{enumerate}
      \item [Step 1.] find a vector of non-negative rationals $\blambda_{\sfBT} = (\lambda_B)_{B \in {\sfBT}}$ with $\|\blambda_{\sfBT}\|_1 = 1$, extend it to a vector (over conditional polymatroid) $\blambda \in \bQ^\mC_+$ and  $\LogSizeBound_{\Gamma_n \cap \HDC}(\rho) = \max_{h \in \Gamma_n \cap \HDC} \langle \blambda, \bh \rangle$, where the right-hand side is a linear program;
      \item [Step 2.] find an optimal dual solution $(\bdelta, \bsigma, \bmu)$ to the linear program, $\max_{h \in \Gamma_n \cap \HDC} \langle \blambda, \bh \rangle$, so that 
      $$
      \OBJ \defeq \sum_{(X, Y) \in \DC} \log N_{Y \mid X} \cdot \delta_{Y \mid X} = \LogSizeBound_{\Gamma_n \cap \HDC}(\rho)
      $$
      and $\langle \bdelta, \bh \rangle \geq \langle \blambda, \bh \rangle$ forms a Shannon-flow inequality;
      \item [Step 3.] construct a proof sequence $\proofseq$ of length $O(poly(2^n))$ for the Shannon-flow inequality $\langle \bdelta, \bh \rangle \geq \langle \blambda, \bh \rangle$ (see \autoref{thm:prooseqfinite});
      \item [Step 4.] run a $\PANDA$ instance, denoted as $\PANDA(\mD, \DC, (\blambda, \bdelta), \proofseq)$, which interprets each proof step of $\proofseq$ as a relational operation on the input relation, where each operation is guaranteed to take time $\polyO (2^{\OBJ})$. Overall, the $\PANDA$ instance runs in time $\polyO (2^{\OBJ})$ and computes a model of size $\polyO (2^{\OBJ})$.
  \end{enumerate}
  
  In particular, the $\PANDA$ instance maintains the following $4$ invariants on its inputs, i.e. $(\mD, \DC, (\blambda, \bdelta), \proofseq)$:

  \introparagraph{$\PANDA$ invariants}  

%  \introparagraph{Join Size Bound} Let $\mD$ be a fixed database instance respecting a given set of degree constraints $\DC$. We denote by $\varphi(\mD)$ the output of the CQ $\varphi$ evaluated on $\mD$ with output size $|\varphi(\mD)|$. The set
%   $$
%      \HDC \defeq \left\{h: 2^{[n]} \rightarrow \mathbb{R}_{+} | \bigwedge_{\left(X, Y, N_{Y | X}\right) \in \DC} h(Y|X) \leq \log N_{Y| X}\right\}
%   $$
%  contains all entropic functions $h$ on $[n]$ satisfying the degree constraints $\DC$. Then, the following inequalities on the join size holds:
%  $$ \log |\varphi(\mD)| \leq \max_{h \in \Gamma_n^* \cap \HDC} h([n]) \leq \max_{h \in \Gamma_n \cap \HDC} h([n]),
%  $$
%  \\ \\

  \begin{enumerate}
      \item [(1)] \textit{Degree-support invariant:} For every $\delta_{Y|X} > 0$, there exist $Z \subseteq X, W \subseteq Y$ such that $W - Z = Y - X$ and $(Z, W, N_{W|Z}) \in \DC$. The degree constraint is said to support the positive $\delta_{Y|X}$. If there are more than one $(Z, W, N_{W|Z}) \in \DC$ supporting $\delta_{Y|X}$, we choose the one with minimum $N_{W|Z}$ and call it the supporting constraint of $\delta_{Y|X}$.
      \item [(2)] $0 < \|\blambda\|_1 \leq 1$
      \item [(3)] The Shannon ﬂow inequality along with the supporting degree constraints satisfy $\sum_{(X, Y)} n(\delta_{Y \mid X}) \leq\|\blambda\|_1 \cdot \OBJ$ where 
      $$
      n(\delta_{Y \mid X}) \defeq \begin{cases}\delta_{Y \mid X} \cdot n_{W \mid Z} & \text { if } \delta_{Y \mid X}>0 \text { and } \\  & (Z, W, N_{W \mid Z}) \text { supports it } \\ 0 & \text { if } \delta_{Y \mid X}=0 .\end{cases}
      $$
      and we call the quantity $\sum_{(X, Y)} n(\delta_{Y \mid X})$ the \textit{potential}.
      \item [(4)] For every $\delta_{Y|\emptyset} > 0$, the supporting degree constraint $(\emptyset, Y, N_{Y| \emptyset})$ satisfies $n_{Y|\emptyset} \leq \OBJ$.
  \end{enumerate}
 
  \introparagraph{A slight augmentation on $\PANDA$}
  
  For our purposes, we augment $\PANDA$ slightly to take non-optimal proof sequences. More precisely, suppose instead of Step 1 and 2, we are given vectors $(\blambda_{\sfBT}, \bdelta_{\DC})$ and witness $(\bsigma, \bmu)$ such that by extending both vectors as $\blambda, \bdelta \in \bQ^\mC_+$, 
  $\langle \bdelta, \bh \rangle \geq \langle \blambda, \bh \rangle$ forms a (not necessarily optimal) Shannon-flow inequality with witness $(\bsigma, \bmu)$. Note that the implied upper bound $\OBJ$ satisfies 
   $$
     \LogSizeBound_{\Gamma_n \cap \HDC}(\rho) \leq \OBJ \defeq \sum_{(X, Y) \in \DC} \log N_{Y \mid X} \cdot \delta_{Y \mid X},
$$
  but not necessarily coincides with $\LogSizeBound_{\Gamma_n \cap \HDC}(\rho)$. 
   In such cases, following Step 3 and 4, we show that the $\PANDA$ instance (taking the non-optimal proof sequence for the given Shannon-flow inequality as input), runs in time as predicated by the non-optimal Shannon-flow inequality, i.e. $\polyO (2^{\OBJ})$.
  
   The augmentation is as follows. We observe that in the proof of correctness of $\PANDA$ \cite{DBLP:conf/pods/Khamis0S17}, only invariant $(4)$ (at the beginning of the $\PANDA$ instance) relies on the optimality of proof sequences. However, as argued in Proposition 6.2, if initially for some $\delta_{Y|\emptyset} > 0$, $n_{Y|\emptyset} > \OBJ$, then we could replace the original Shannon-flow inequality with a new Shannon-flow inequality $\langle \bdelta', \bh \rangle \geq \langle \blambda', \bh \rangle$ along with witness $(\bsigma', \bmu')$ such that invariants (1)-(3) hold, and the length of the proof sequence decreases by at least $1$. We can repeat this replacement step iteratively until invariant $(4)$ is satisfied. If invariant $(4)$ is never satisfied, then we will end up with a proof sequence of length $0$, in which case \highlight{the Shannon-flow inequality becomes a trivial one, $\langle \bdelta_0, \bh \rangle \geq \langle \blambda_0, \bh \rangle$, for some $\blambda_0, \bdelta_0$ with $0 < \|\blambda_0\|_1 \leq 1$ (by invariant (2)) and $\bdelta_0 \geq \blambda_0$ (element-wise comparison).
   } This implies that any input relation $R_F$ where $(\lambda_0)_{F|\emptyset} > 0$ can be a model and there is an $n_{F|\emptyset} \leq \OBJ$ that can be appointed as \textit{the} output model (so the model size is $\polyO(2^{\OBJ})$), because otherwise, 
   $$\sum_{F: (\lambda_0)_{F|\emptyset} > 0} \frac{(\lambda_0)_{F|\emptyset}}{\|\blambda_0\|_1} \cdot n_{F|\emptyset} > \sum_{F: (\lambda_0)_{F|\emptyset} > 0} \frac{(\lambda_0)_{F|\emptyset}}{\|\blambda_0\|_1} \cdot \OBJ = \OBJ,
   $$
   and this contradicts invariant (3). For the $\twoPP$ algorithm, we implicitly assume that $\PANDA$ is equipped with this minor augmentation.

 \subsection{The Algorithm}
 
While $\twoPP$ follows a similar structure as the na\"ive algorithm (and uses the 2-phase algorithmic framework), it is guided by a joint Shannon-flow inequality to execute only the necessary split steps and $\PANDA$ instances, which provides more practicality and interpretability. In particular, we will show the following theorem for $\twoPP$.
 
 \begin{theorem}\label{thm:main}
  Let $\rho$ be a 2-phase disjunctive rule of the form \eqref{def:partitionDisjunctiveRule} satisfying degree constraints $\DC$ (guarded by the input relations) and $\AC$ (guarded by the access request). Let $(\blambda_{\sfBT}, \btheta_{\sfBS}) \in \left\{(\ba_T,\ba_S) \mid \ba_T \in \bQ_+^{\sfBT}, \ba_S \in \bQ_+^{\sfBS}, \|\ba_T\|_1 = 1 \right\}$. 
  The $\twoPP$ algorithm obtains a model of $\rho$ in two phases and attains the following (smooth) intrinsic trade-off:
% \begin{equation}\label{jointSF}
%     \begin{aligned}
%           \sum_{(X, Y) \in \DC} \hS(Y|X) \cdot (\delta_S)_{Y|X} & +  \sum_{(X, Y) \in \DC \cup \AC} \hT(Y|X) \cdot (\delta_T)_{Y|X} + \sum_{(X, Y|X) \in \SC}  (\hS(X) + \hT(Y|X)) \cdot \gamma_{X, Y|X} \\
%       & + \sum_{(X, Y|X) \in \SC}  (\hS(Y|X) + \hT(X)) \cdot \gamma_{Y|X, X} \geq \sum_{B \in \sfBT} \theta_B \cdot \hS(B) + \sum_{B \in \sfBS} \lambda_B \cdot \hT(B)
%     \end{aligned}
% \end{equation}
%   and obtains a model of $\rho$ in two phases, attaining the trade-off as predicted by \eqref{jointSF}:
%   \begin{align}\label{tradeOff:disjunctive}
%       \bspace^{\|\btheta_{\sfBS}\|_1} \cdot \btime \cong 2^{\ell(\blambda_{\sfBT}, \btheta_{\sfBS})} = \prod_{(X, Y) \in \DC} N_{Y|X}^{(\delta_S)_{Y|X}} \cdot \prod_{(X, Y) \in \DC \cup \AC} N_{Y|X}^{(\delta_S)_{Y|X}} \cdot \prod_{(X, Y|X) \in \SC} N_{Z|\emptyset}^{\gamma_{X, Y|X} + \gamma_{Y|X, X}} 
%   \end{align}
  \begin{align}\label{tradeOff:disjunctive}
      \bspace_{\rho}^{\|\btheta_{\sfBS}\|_1} \cdot \btime_{\rho} \cong 2^{\ell(\blambda_{\sfBT}, \btheta_{\sfBS})}
  \end{align}
  where $\cong$ hides a poly-logarithmic factor at the right-hand side.
 \end{theorem} 
 
% Note that by strong duality,
% $$
% \ell(\blambda_{\sfBT}, \btheta_{\sfBS}) = \sum_{(X, Y) \in \DC} \log N_{Y|X} \cdot (\delta_S)_{Y|X} + \sum_{(X, Y) \in \DC \cup \AC} \log N_{Y|X} \cdot (\delta_T)_{Y|X} + \sum_{(X, Y|X) \in \SC}  \log N_{Z|\emptyset} \cdot (\gamma_{X, Y|X} + \gamma_{Y|X, X})
% $$
% where $\{(\delta_S)_{Y|X}\}, \{(\delta_T)_{Y|X}\},  \{\gamma_{X, Y|X}\}, \{\gamma_{Y|X, X}\}$ denotes the dual optimal solution for $L(\blambda_{\sfBT}, \btheta_{\sfBS})$. Thus, the intrinsic trade-off can be equivalently written as 
Note that the intrinsic trade-off can be equivalently written as
   \begin{align}\label{tradeOff:disjunctive}
      \bspace_{\rho}^{\|\btheta_{\sfBS}\|_1} \cdot \btime_{\rho} \cong  \prod_{(X, Y) \in \DC} N_{Y|X}^{(\delta_S)_{Y|X}} \cdot \prod_{(X, Y) \in \DC \cup \AC} N_{Y|X}^{(\delta_S)_{Y|X}} \cdot \prod_{(X, Y|X) \in \SC} N_{Z|\emptyset}^{\gamma_{X, Y|X} + \gamma_{Y|X, X}} 
  \end{align}
  In particular, given a fixed $\bspace$ for $S_{\rho}$, we can construct from \autoref{lem:maxminLP} (using complementary slackness) $(\blambda_{\sfBT}^*, \btheta_{\sfBS}^*)$ such that $\OBJ(S) = \bL(\blambda_{\sfBT}^*, \btheta_{\sfBS}^*, S)$. Since $\ell(\blambda_{\sfBT}^*, \btheta_{\sfBS}^*) - (\log S) \cdot \|\btheta_{\sfBS}^*\|_1 = \bL(\blambda_{\sfBT}^*, \btheta_{\sfBS}^*, S)$, \autoref{thm:main} recovers the best possible instrinsic trade-off as specified in \eqref{OPT-maximin}. In the rest of this section, we present the $\twoPP$ algorithm, and defer the full proof of~\autoref{thm:main} to the next section.
  
  %  $\twoPP$ follows the algorithmic recipe: 
%  it starts by partitioning the original problem $\rho(\mD, W)$ as guided by the joint Shannon-flow inequality and for the $j$-th subproblem $\rho(\mD^{(j)}, W)$, either
%  \begin{enumerate}
%      \item [(1)] builds a $\PANDA$ instance using $\proofseq(S)$ for $\rho_S(\mD^{(j)}, W)$ and store the output $(S^{(i)}_B)_{B \in \sfBT}$ in the preprocessing phase; or
%      \item [(2)] stores its input $\mD^{(j)}$ in the preprocessing phase and in the online phase, builds an $\PANDA$ instance using $\proofseq(T)$ for $\rho_T(\mD^{(j)}, W)$ and obtain its output $(T^{(i)}_B)_{B \in \sfBS}$
%  \end{enumerate}

\introparagraph{Preparation phase}
 Similar to $\PANDA$, $\twoPP$ has a preparation phase to construct the necessary inputs for a $\twoPP$ instance. First, we construct from the given $(\blambda_{\sfBT}, \btheta_{\sfBS})$ a joint Shannon-flow inequality, $\langle \bg_S, \bhS \rangle + \langle \bg_T, \bhT \rangle \geq  \langle \btheta, \bhS \rangle + \langle \blambda, \bhT \rangle$ and a witness for it, $(\bsigma_S, \bmu_S, \bsigma_T, \bmu_T)$. Recall that the joint Shannon-flow inequality implies an upper bound coincides with $\ell(\blambda_{\sfBT}, \btheta_{\sfBS})$.
 
 Second, we construct a proof sequence for the joint Shannon-flow inequality. The idea is to build two parallel proof sequences for its two participating Shannon-flow inequalities and stitch them together. Recall that the two participating Shannon-flow inequalities are 
 \begin{align*}
     \langle \bg_S, \bhS \rangle & \geq \langle \btheta, \bhS \rangle \\
 \langle \bg_T, \bhT \rangle & \geq \langle \blambda, \bhT \rangle
 \end{align*}

From the proof of \autoref{thm:witness}, $(\bsigma_S, \bmu_S)$ is a witness for $\langle \bg_S, \bhS \rangle \geq \langle \btheta, \bhS \rangle$ and $(\bsigma_T, \bmu_T)$ is a witness for $\langle \bg_T, \bhT \rangle \geq \langle \blambda, \bhT \rangle$. We  normalize the Shannon-flow inequality $\langle \bg_S, \bhS \rangle \geq \langle \btheta, \bhS \rangle$ into
 $\langle \widetilde{\bg}_S, \bhS \rangle \geq \langle \widetilde{\btheta}, \bhS \rangle$, where $\widetilde{\bg}_S \defeq \bg_S/\|\btheta\|_1$ and $\widetilde{\btheta} \defeq \btheta / \|\btheta\|_1$, with a new witness $(\widetilde{\bsigma}_S, \widetilde{\bmu}_S) \defeq (\bsigma_S/\|\btheta\|_1, \bmu_S/\|\btheta\|_1)$. From \autoref{thm:prooseqfinite}, we can construct a proof sequence $\proofseq(S)$ for the Shannon-flow inequality $\langle \widetilde{\bg}_S, \bhS \rangle \geq \langle \widetilde{\btheta}, \bhS \rangle$ and a proof sequence $\proofseq(T)$ for $\langle \bg_T, \bhT \rangle \geq \langle \blambda, \bhT \rangle$, where both proof sequences have length $O(poly(2^n))$, $n$ being the number of variables. We say that $\proofseq(S)$ and $\proofseq(T)$ are {\textit{the participating proof sequences}} for the joint Shannon-flow inequality. As a brief summary, in the preparation phase, $\twoPP$:
 \begin{enumerate}
     \item [(1)] (see \autoref{thm:witness}) constructs a joint Shannon-flow inequality $\langle \bg_S, \bhS \rangle + \langle \bg_T, \bhT \rangle \geq  \langle \btheta, \bhS \rangle + \langle \blambda, \bhT \rangle$ with a witness $(\bsigma_S, \bmu_S, \bsigma_T, \bmu_T)$; and
    \item [(2)] (see \autoref{thm:prooseqfinite}) constructs a $\proofseq(S)$ for the participating Shannon-flow inequality $\langle \widetilde{\bg}_S, \bhS \rangle \geq \langle \widetilde{\btheta}, \bhS \rangle$ and a $\proofseq(T)$ for the participating Shannon-flow inequality $\langle \bg_T, \bhT \rangle \geq \langle \blambda, \bhT \rangle$ 
 \end{enumerate}
   
Now, we walk through both phases of $\twoPP$. In particular, we denote $\twoPPp$ as the preprocessing phase of $\twoPP$ and  $\twoPPo$ as the online phase of $\twoPP$. The sketches of $\twoPPp$ and $\twoPPo$ are in the box of \autoref{2PP-1} and \autoref{2PP-2}.
 
 \begin{algorithm}[t]
 \SetKwInOut{Input}{Input}
 \SetKwInOut{Output}{Output}
    \Input{ a database instance $\mD$ and degree constraints $\DC$ guarded by $\mD$}
    \Input{ the degree constraints $\AC$ guarded by any possible access request $Q_A$}
    \Input{ the participating Shannon-flow inequality $\langle \widetilde{\bg}_S, \bhS \rangle \geq \langle \widetilde{\btheta}, \bhS \rangle$ and its proof sequence $\proofseq(S)$}
    Let $\SC$ be the set of split constraints spanned from $\DC$ \\
    $\SC^+ \leftarrow \left\{(Y, X) \mid (X, Y|X, N_{Z|\emptyset}) \in \SC, \gamma_{X, Y|X} > 0 \vee \gamma_{Y|X, X} > 0 \right\}$ \\
    Apply a sequence of split steps, one for every $(Y, X) \in \SC^+$ and spawn $k$ subproblems with inputs $(\mD^{(j)}, \DC^{j})$, $j \in [k]$ \\ \tcp*[f]{$k = O(poly(\log |\mD|))$}
    \\ 
    $\mJ \leftarrow \emptyset$ \\
    \ForAll{$j \in [k]$} {
    Create a $\PANDA$ instance $\PANDA(\mD^{(j)}, \DC^{(j)}, (\widetilde{\btheta}, \widetilde{\bg}_S), \proofseq(S))$ \\
        \If{the potential satisfies \eqref{alg:abortCond}}{
        $(S^{(j)}_B)_{B \in \sfBS} \leftarrow \PANDA(\mD^{(j)}, \DC^{(j)}, (\widetilde{\btheta}, \widetilde{\bg}_S), \proofseq(S))$ \\
        } \Else{
            \textbf{abort} $\PANDA(\mD^{(j)}, \DC^{(j)}, (\widetilde{\btheta}, \widetilde{\bg}_S), \proofseq(S))$ \\
            \textbf{insert} $(\mD^{(j)}, \DC^{(j)})$ to $\mJ$ \\
        }   
    }
    \Return $\left(\mJ, (\bigcup_j S^{(j)}_B)_{B \in \sfBS}\right)$ \\
    
    \caption{$\twoPPp (\mD, \DC, \AC, (\widetilde{\btheta}, \widetilde{\bg}_S), \proofseq(S))$}\label{2PP-1}
\end{algorithm}

\begin{algorithm}[t]
 \SetKwInOut{Input}{Input}
 \SetKwInOut{Output}{Output}
    \Input{ an index $\mJ$ containing $O(poly(\log |\mD|))$ entries where each entry contains input relations $\mD^{(j)}$ and degree constraints $\DC^{(j)}$ guarded by $\mD^{(j)}$}
    \Input{ an access request $Q_A$ and the degree constraints $\AC$ guarded by $Q_A$}
    \Input{ the participating Shannon-flow inequality $\langle \bg_T, \bhT \rangle \geq \langle \blambda, \bhT \rangle$ and its proof sequence $\proofseq(T)$}

    \ForAll{$(\mD^{(j)}, \DC^{(j)}) \in \mJ$} {
    Create a $\PANDA$ instance $\PANDA(\mD^{(j)}\cup \{Q_A\}, \DC^{(j)} \cup \AC, (\blambda, \bg_T), \proofseq(T))$ \\
        $(T^{(j)}_B)_{B \in \sfBT} \leftarrow \PANDA(\mD^{(j)}\cup \{Q_A\}, \DC^{(j)} \cup \AC, (\blambda, \bg_T), \proofseq(T))$ \\
    }
    \Return $(\bigcup_j T^{(j)}_B)_{B \in \sfBT}$ \\
    
    \caption{$\twoPPo (\mJ, \{Q_A\}, \AC, (\blambda, \bg_{T}), \proofseq(T))$}\label{2PP-2}
\end{algorithm}
 
 \introparagraph{The preprocessing phase} We call this phase $\twoPPp$ and it is sketched in the box of \autoref{2PP-1}.
We first scan over the joint Shannon-flow inequality, $\langle \bg_S, \bhS \rangle + \langle \bg_T, \bhT \rangle \geq  \langle \btheta, \bhS \rangle + \langle \blambda, \bhT \rangle$ and $(\bsigma_S, \bmu_S, \bsigma_T, \bmu_T)$ and apply a sequence of split steps that consists of one $(Y, X)$-pair for every $(Y, X)$ satisfying $(X, Y|X, N_{Z|\emptyset}) \in \SC$ and either $\gamma_{X, Y|X} > 0$ or $\gamma_{Y|X, X} > 0$. The sequence of split steps spawns $\poly(\log |\mD|)$ subproblems. Let $\rho(\mD^{(j)} \cup \{Q_A\}, \DC^{(j)} \cup \AC)$ be the $j$-th subproblem after the sequence of split steps, having degree constraints $\DC^{(j)}$ guarded by $\mD^{(j)}$. Next, we create for it a $\PANDA$ instance $\PANDA(\mD^{(j)}, \DC^{(j)}, (\widetilde{\btheta}, \widetilde{\bg}_S), \proofseq(S))$ and look at its initial \textit{potential}
 \begin{align}\label{potential-s}
     \sum_{(X, Y)} n^{(j)}_{W|Z} \cdot (\widetilde{g}_S)_{Y|X} 
    %  & \defeq \sum_{(X, Y)} n^{(j)}_{W|Z} \cdot \frac{(g_S)_{Y|X}}{\|\btheta\|_1}  
     \defeq \sum_{(X, Y) \in \DC} n^{(j)}_{W|Z} \cdot \frac{(\delta_S)_{Y|X}}{\|\btheta\|_1}   + \sum_{(X, Y|X) \in \SC}  n^{(j)}_{X|\emptyset} \cdot \frac{\gamma_{X, Y|X}}{\|\btheta\|_1} + \sum_{(X, Y|X) \in \SC}  n^{(j)}_{W|Z}  \cdot \frac{\gamma_{Y|X, X}}{\|\btheta\|_1}
 \end{align}
 where $n^{(j)}_{W|Z} \defeq \log N^{(j)}_{W|Z}$ and $(Z, W, N^{(j)}_{W|Z}) \in \DC^{(j)}$ is the constraint that supports a positive $(\widetilde{g}_S)_{Y|X}$. If the potential satisfies
 \begin{align}\label{alg:abortCond}
     \sum_{(X, Y)} n^{(j)}_{W|Z} \cdot (\widetilde{g}_S)_{Y|X} \leq \log \bspace,
 \end{align}
 then $\twoPP$ allows this $\PANDA$ instance to run and stores its output $(S^{(j)}_B)_{B \in \sfBS}$. Otherwise, $\twoPP$ aborts this $\PANDA$ instance and keeps track of the input $(\mD^{(j)}, \DC^{(j)})$ for the $j$-th subproblem using an index $\mJ$. The data structure(s) stored in the preprocessing phase are the index $\mJ$ that tracks inputs (and its degree constraints) from aborted instances, and a set of tables $S_B = \bigcup_j S^{(j)}_B, B \in \sfBS$ from all succeeded $\PANDA$ instances.
 
 \introparagraph{The online phase} We call this phase $\twoPPo$ and it is sketched in the box of \autoref{2PP-2}.
The algorithm scans over the index $\mJ$ built in the preprocessing phase, and for each $(\mD^{(j)}, \DC^{(j)}) \in \mJ$, it creates and runs a $\PANDA$ instance $\PANDA(\mD^{(j)}\cup \{Q_A\}, \DC^{(j)} \cup \AC, (\blambda, \bg_T), \proofseq(T))$. At the termination of each $\PANDA$ instance, $\twoPP$ collects outputs $(T^{(j)}_B)_{B \in \sfBT}$. The overall output in the online phase is the set of tables $T_B = \bigcup_j T^{(j)}_B, B \in \sfBT$ from all $\PANDA$ instances created and executed online.

\subsection{Analysis of the $\twoPP$ algorithm}
\label{sec:proof}

 In the rest of the section, we formally prove Theorem~\ref{thm:main} for the $\twoPP$ algorithm. 

\begin{proof}[Proof of~\autoref{thm:main}]
 
  Recall that the split steps at the initial stage of $\twoPP$ spawn $O(poly(\log |\mD|))$ subproblems. We prove by analyzing the space and time usage for the $j$-th subproblem. 
 
  First, in the preprocessing phase, if the potential of the $\PANDA$ instance, $\PANDA(\mD^{(j)}, \DC^{(j)}, (\widetilde{\btheta}, \widetilde{\bg}_S), \proofseq(S))$, is no larger than $\log \bspace$, then by invariant $(4)$ of $\PANDA$, the output tables $(S^{(j)}_B)_{B \in \sfBT}$ (and every intermediate view produced by this $\PANDA$ instance) have size $\polyO(\bspace)$. Otherwise, the $j$-th subproblem consumes space $O(|\mD|)$ for storing its input $(\mD^{(j)}, \DC^{(j)})$ in the index. Thus, the overall space consumption for the data structure is $O(poly(\log |\mD|)) \cdot \polyO(|\mD| + \bspace) = \polyO(\bspace)$.
  
  Next, we are left to show that for any subproblem delegated to the online phase, the $\PANDA$ instance created by
  $\twoPP$, $\PANDA(\mD^{(j)}\cup \{Q_A\}, \DC^{(j)} \cup \AC, (\blambda, \bg_T), \proofseq(T))$, terminates in time $\polyO(\btime_{\rho})$, where
 $$ 
      \log \btime_{\rho} = \ell(\blambda_{\sfBT}, \btheta_{\sfBS}) - \log \bspace \cdot {\|\btheta_{\sfBS}\|_1}.
 $$
 Recall that $n_{Y|X} \defeq \log N_{Y|X}$ and by strong duality,
 $$\ell(\blambda_{\sfBT}, \btheta_{\sfBS})  = \sum_{(X, Y) \in \DC} n_{Y|X} \cdot {(\delta_S)_{Y|X}} + \sum_{(X, Y) \in \DC \cup \AC} n_{Y|X} \cdot {(\delta_S)_{Y|X}} + \sum_{(X, Y|X) \in \SC} n_{Z|\emptyset} \cdot ({\gamma_{X, Y|X} + \gamma_{Y|X, X}}).$$
 To show this, we look at both $\PANDA$ instances of the $j$-th subproblem, $\rho(\mD^{(j)} \cup \{Q_A\}, \DC^{(j)} \cup \AC)$, i.e.
 \begin{align*}
     \textit{(preprocessing instance)} \qquad & \PANDA(\mD^{(j)}, \DC^{(j)}, (\widetilde{\btheta}, \widetilde{\bg}_S), \proofseq(S)) \\
     \textit{(online instance)} \qquad & \PANDA(\mD^{(j)}\cup \{Q_A\}, \DC^{(j)} \cup \AC, (\blambda, \bg_T), \proofseq(T))
 \end{align*}
 The \textit{preprocessing instance} has the potential specified in \eqref{potential-s}, while the \textit{online instance} has the following \textit{potential}:
 $$ \sum_{(X, Y)} n^{(j)}_{W|Z} \cdot (g_T)_{Y|X} \defeq \sum_{(X, Y) \in \DC \cup \AC} n^{(j)}_{W|Z} \cdot (\delta_T)_{Y|X}  + \sum_{(X, Y|X) \in \SC}  n^{(j)}_{W|Z} \cdot \gamma_{X, Y|X} + \sum_{(X, Y|X) \in \SC}  n^{(j)}_{X|\emptyset}  \cdot \gamma_{Y|X, X}
 $$
 where $n^{(j)}_{W|Z} \defeq \log N^{(j)}_{W|Z}$ and $(Z, W, N^{(j)}_{W|Z}) \in \DC^{(j)}$ is the constraint that supports a positive $(g_T)_{Y|X}$. Now we have, 
 \begin{align*}
      \sum_{(X, Y)} n^{(j)}_{W|Z} \cdot (g_T)_{Y|X} + \|\btheta\|_1 \cdot \sum_{(X, Y)} n^{(j)}_{W|Z} \cdot & (\widetilde{g}_S)_{Y|X}  =
      \sum_{(X, Y) \in \DC} n^{(j)}_{W|Z} \cdot (\delta_S)_{Y|X} + \sum_{(X, Y) \in \DC \cup \AC}   n^{(j)}_{W|Z} \cdot (\delta_T)_{Y|X} \\ 
       & + \sum_{(X, Y|X) \in \SC}  (n^{(j)}_{X|\emptyset} + n^{(j)}_{W|Z}) \cdot \gamma_{X, Y|X} + \sum_{(X, Y|X) \in \SC}  (n^{(j)}_{W|Z} + n^{(j)}_{X|\emptyset})  \cdot \gamma_{Y|X, X}
    %   + \sum_{(X, Y) \in \DC \cup \AC}   n^{(j)}_{W|Z} \cdot (\delta_T)_{Y|X}  + \sum_{(X, Y|X) \in \SC}  n^{(j)}_{W|Z} \cdot \gamma_{X, Y|X} + \sum_{(X, Y|X) \in \SC}  n^{(j)}_{X|\emptyset}  \cdot \gamma_{Y|X, X} \\
    %   = \sum_{(X, Y) \in \DC} n_{Y|X} \cdot {(\delta_S)_{Y|X}} + \sum_{(X, Y) \in \DC \cup \AC} n_{Y|X} \cdot {(\delta_S)_{Y|X}} + \sum_{(X, Y|X) \in \SC} n_{Z|\emptyset} \cdot ({\gamma_{X, Y|X} + \gamma_{Y|X, X}}) 
 \end{align*}
 Recall that $\twoPP$ executes a split step for every $(Y, X)$-pair satisfying $(X, Y|X, N_{Z|\emptyset}) \in \SC$ and $\gamma_{X, Y|X} > 0 \vee \gamma_{Y|X, X} > 0$. So for every such $(X, Y)$-pair, there are some $(\emptyset, X,  N_{X | \emptyset}^{(j)}), (X, Y, N_{Y|X}^{(j)}) \in \DC^{(j)}$ such that $ N_{X | \emptyset}^{(j)} \cdot N_{Y | X}^{(j)} \leq N_{Z | \emptyset}$. It implies that $n^{(j)}_{X|\emptyset} + n^{(j)}_{W|Z} \leq \log N_{Z|\emptyset}$. Moreover, for every $(X, Y) \in \DC$ where $(\delta_S)_{Y|X} > 0$ or $(\delta_T)_{Y|X} > 0$, it holds that $n^{(j)}_{W|Z} \leq \log N_{Y|X}$. Therefore, we get
 $$ \sum_{(X, Y)} n^{(j)}_{W|Z} \cdot (g_T)_{Y|X} + \|\btheta\|_1 \cdot \sum_{(X, Y)} n^{(j)}_{W|Z} \cdot (\widetilde{g}_S)_{Y|X}  \leq \ell(\blambda_{\sfBT}, \btheta_{\sfBS})
 $$
 This implies that for the $j$-th subproblem (thus for any subproblem), either in the preprocessing phase, 
 $$ \sum_{(X, Y)} n^{(j)}_{W|Z} \cdot (g_T)_{Y|X} \leq \log \bspace, 
 $$
 or in the online phase
 \begin{align*}
     \sum_{(X, Y)} n^{(j)}_{W|Z} \cdot (g_T)_{Y|X} & \leq \ell(\blambda_{\sfBT}, \btheta_{\sfBS}) - \|\btheta\|_1 \cdot \sum_{(X, Y)} n^{(j)}_{W|Z} \cdot (\widetilde{g}_S)_{Y|X} \\
     & \leq \ell(\blambda_{\sfBT}, \btheta_{\sfBS}) - \|\btheta_{\sfBS}\|_1 \cdot \log \bspace \\
     & = \log \btime_{\rho}.
 \end{align*}
 So, all online $\PANDA$ instances terminate in time as predicated by \eqref{tradeOff:disjunctive}.
 
 \end{proof}
 
 \eat{
 \section{Example of The $\twoPP$ Algorithm}
 \begin{example} Consider the $2$-Set Disjointness example (boolean): 
 The joint Shannon-flow inequality
  \begin{align*}
    \textcolor{black}{\hS(1)} + \textcolor{black}{\hT( 2 | 1)} + \textcolor{black}{\hS(3)} + \textcolor{black}{\hT(2| 3)}+ 2\cdot \textcolor{black}{\hT(1 3)} \geq \textcolor{black}{\hS(1 3)} + 2 \cdot \textcolor{black}{\hT(1 2 3)} 
\end{align*}
%  The participating proof sequences
%  \begin{align*}
%      \textcolor{black}{\hS(1)} + \textcolor{black}{\hS(3)} & \geq \textcolor{black}{\hS(13|3)} + \textcolor{black}{\hS(3)} \\
%      & = \textcolor{black}{\hS(13)} 
%  \end{align*}
%  \begin{align*}
%       \textcolor{black}{\hT(2 | 1)} + \textcolor{black}{\hT(2| 3)}+ 2 \textcolor{black}{\hT(1 3)} & \geq  2\textcolor{black}{\hT(2|13)} + 2 \textcolor{black}{\hT(1 3)}  \\
%       &  \geq 2 \textcolor{black}{\hT(123)}
%  \end{align*}
 $\twoPP$ algorithm (let $|Q_{13}| = 1$, $S$ be an arbitrarily-fixed space budget):
 \begin{enumerate}
     \item [(1)] $\twoPP$ scans the joint Shannon-flow inequality and partitions $R_{12}(x_1, x_2)$ into $R'_{12}(x_1)$ and $R'_{12}(x_1, x_2)$, where $R'_1(x_1)$ contains all values of $x_1$ that are heavy-hitters, meaning that $|\sigma_{x_1 = t}(R_{12})| \geq N/\sqrt{S}$, and $R'_{12}(x_1, x_2)$ contains all pairs $(x_1, x_2)$ where $x_1$ is the light-hitters, i.e.  $R'_{12}(x_1, x_2)$ satisfies $\deg_{12}(x_2|x_1) \leq N/\sqrt{S}$. Similarly, $\twoPP$ partitions $R_{23}(x_2, x_3)$ into $R'_{23}(x_3)$ and $R'_{23}(x_2, x_3)$ where 
     $R'_3(x_3)$ contains all values of $x_1$ such that $|\sigma_{x_3 = t}(R_{23})| \geq N/ \sqrt{S}$, and $R'_{23}(x_2, x_3)$ contains all pairs $(x_2, x_3)$ where $\deg_{23}(x_2|x_3) \leq N/\sqrt{S}$. So there are $4$ subproblems, i.e. $R'_1(x_1) \bowtie R'_3(x_3)$, $R'_1(x_1) \bowtie R'_{23}(x_2, x_3)$, $R'_{12}(x_1, x_2) \bowtie R'_3(x_3)$ and $R'_{12}(x_1, x_2) \bowtie R'_{34}(x_3, x_4)$.
     \item [(2)] in the preprocessing phase, $\twoPPp$, for every subproblem, follows the first participating proof sequence
      \begin{align*}
     \textcolor{black}{\hS(1)} + \textcolor{black}{\hS(3)} & \geq \textcolor{black}{\hS(13|3)} + \textcolor{black}{\hS(3)} && (submodularity) \\
     & = \textcolor{black}{\hS(13)} && (composition)
    \end{align*}
    to attempt joining relations and storing output in space $N^{3/2}$. Because $R'_1(x_1)$ and $R'_3(x_3)$ have sizes at most $N / (N/\sqrt{S}) = \sqrt{S}$, the first subproblem $R'_1(x_1) \bowtie R'_3(x_3)$ can be computed and stored (in space $\sqrt{S} \cdot \sqrt{S} = S$)
    \item [(3)] in the online phase, $\twoPPo$ receives an access request $Q_{13}(x_1, x_3)$ that contains one tuple. Now, $\twoPPo$ follows the second participating proof sequence
    \begin{align*}
      \textcolor{black}{\hT(2 | 1)} + \textcolor{black}{\hT(2| 3)}+ 2 \textcolor{black}{\hT(1 3)} & \geq  2\textcolor{black}{\hT(2|13)} + 2 \textcolor{black}{\hT(1 3)}   && (submodularity)\\
      &  = 2 \textcolor{black}{\hT(123)}  && (composition) 
 \end{align*}
    In particular, for the remaining $3$ subproblems, $\twoPPo$ creates
    the following $3$ subproblems: $Q_{13}(x_1, x_3)  \bowtie R'_{23}(x_2, x_3)$, $Q_{13}(x_1, x_3) \bowtie R'_{12}(x_1, x_2)$ and $Q_{13}(x_1, x_3) \bowtie R'_{12}(x_1, x_2)$. Then, in the submodularity step, $\twoPPo$ identifies that for the subproblem $Q_{13}(x_1, x_3) \bowtie R'_{23}(x_2, x_3)$, $\deg_{23}(x_2|x_3) \leq N/\sqrt{S}$, so the join takes time $|Q_{13}| \cdot N/\sqrt{S} \leq N/\sqrt{S}$; and similarly for the last two subproblems, $\twoPPo$ identifies that $\deg_{12}(x_2|x_1) \leq N/\sqrt{S}$, so both joins take time $|Q_{13}| \cdot \deg_{12}(x_2|x_1) \leq  N/\sqrt{S}$. Therefore, the overall online computing time is $N/\sqrt{S}$.
 \end{enumerate}
%  \begin{align*}
%     n_{12} + n_{23} + 2 \cdot \textcolor{pink}{w_{13}}
%     & \geq \textcolor{black}{\hS(1)} + \textcolor{black}{\hT( 2 | 1)} + \textcolor{black}{\hS(3)} + \textcolor{black}{\hT(2| 3)}+ 2\cdot \textcolor{black}{\hT(1 3)} && \textit{bucketize} \\
%     & \geq \textcolor{black}{\hS(1 3)} + \textcolor{black}{\hT(2 | 1)} + \textcolor{black}{\hT(2|3)} + 2\cdot \textcolor{black}{\hT(1 3)} && \textit{join}\\ 
%     & \geq \textcolor{black}{\hS(1 3)} + \textcolor{black}{\hT(2 | 1 3)} + \textcolor{black}{\hT(2| 1 3)} + 2\cdot \textcolor{black}{\hT(1 3)} && \textit{submodularity}\\ 
%     & \geq \textcolor{black}{\hS(1 3)} + 2 \cdot \textcolor{black}{\hT(1 2 3)} &&  S \cdot T^2 \cong |\mD|^2 \cdot |Q_A|^2
% \end{align*}
 \end{example}
 }

%% file: hierarchical.tex
\section{Space-Time Tradeoffs for Boolean Hierarchical Queries} \label{sec:hierarchicalCQAP}

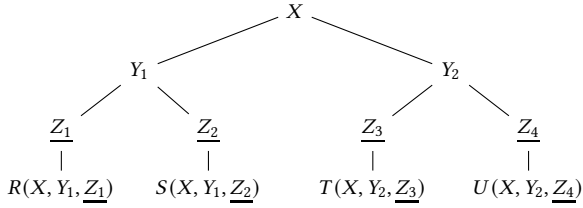
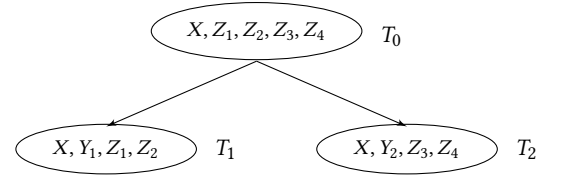
\begin{figure}
    \centering
    \begin{subfigure}[b]{0.5\linewidth}
    \begin{tikzpicture}[xscale=1.15, yscale=1]
      \node at (0.0, 0.0) (A) {{\small \color{black} $X$}};
      \node at (-1.8, -0.8) (B) {{\small\color{black} $Y_1$}} edge[-] (A);
      \node at (1.8, -0.8) (C) {{\small\color{black} $Y_2$}} edge[-] (A);
      \node at (-2.7, -1.6) (D) {{\small\color{black} $\underline{Z_1}$}} edge[-] (B);
      \node at (-1, -1.6) (E) {{\small\color{black} $\underline{Z_2}$}} edge[-] (B);
      \node at (0.9, -1.6) (F) {{\small\color{black} $\underline{Z_3}$}} edge[-] (C);
      \node at (2.7, -1.6) (G) {{\small\color{black} $\underline{Z_4}$}} edge[-] (C);
      \node at (-2.7, -2.4) (R) {{\small \color{black} $R(X,Y_1,\underline{Z_1})$}} edge[-] (D);
      \node at (-1, -2.4) (S) {{\small \color{black} $S(X,Y_1,\underline{Z_2})$}} edge[-] (E);
      \node at (0.9,  -2.4)(T) {{\small \color{black} $T(X,Y_2,\underline{Z_3})$}} edge[-] (F);
      \node at (2.7,  -2.4)(U) {{\small \color{black} $U(X,Y_2,\underline{Z_4})$}} edge[-] (G);
    \end{tikzpicture} 
    \caption{Hierarchical query represented by a complete binary tree. The leaves of the tree form $\bx_A$. Every root-to-leaf path forms a relation (labeled).} 
    \label{fig:hierarchy}
    \end{subfigure} \hspace{1cm}
    \begin{subfigure}[b]{0.3\linewidth}
	{\begin{tikzpicture}
				\tikzset{edge/.style = {->,> = latex'},
					vertex/.style={circle, thick, minimum size=5mm}}
				\def\x{0.25}
				
				\begin{scope}[fill opacity=1]
				
				\draw[] (7,-1.93) ellipse (1.4cm and 0.38cm) node {\small ${X, Z_1, Z_2, Z_3, Z_4}$};
				\node[vertex]  at (8.8,-2) {$T_{0}$};
				\draw[] (5,-3.5) ellipse (1.2cm and 0.33cm) node {\small \small ${X, Y_1,Z_1,Z_2}$};
				\node[vertex]  at (6.6,-3.5) {$T_{1}$};				
				\draw[] (9,-3.5) ellipse (1.2cm and 0.33cm) node {\small \small ${X, Y_2,Z_3,Z_4}$};						
				\node[vertex]  at (10.6,-3.5) {$T_{2}$};								
				\draw[edge] (7,-2.33) -- (5,-3.2);
				\draw[edge] (7,-2.33) -- (9,-3.2);			
				\end{scope}	
				\end{tikzpicture} 
				}
				\caption{Tree decomposition for the hierarchical query on the left.} \label{fig:decomp}
		\end{subfigure} 
        \caption{Example hierarchical CQAP and its tree decomposition.}
\end{figure}

In this section, we show some applications of our framework to CQ Boolean hierarchical queries, which are CQAPs defined over queries whos body is hierarchical. A query said to be {\em hierarchical} if for any two of its variables, either their sets of atoms are disjoint or one is contained in the other. An alternative interpretation is that each hierarchical query admits a \emph{canonical ordering}\footnote{The original definition in~\cite{kara19} also has a dependency function but we omit that since we do not use it for the decomposition construction.}, which is a rooted tree where the variables of each atom in the query lie along the same root-to-leaf path in the tree and each atom is a child of its lowest variable. For example, the query shown in~\autoref{fig:hierarchy} is hierarchical. Variable $Y_1$ is present in atoms $\{R,S\}$, which is a subset of the atoms in which variable $X$ is present, i.e. $\{R,S,T,U\}$. Note that every root-to-leaf path forms a relation.

Hierarchical queries are an interesting class of queries that captures tractability for a variety of problems such as query evaluation in probabilistic databases~\cite{dalvi2009probabilistic}, dynamic query evaluation~\cite{idris2017dynamic, berkholz2017answering}, etc. Seminal work by~\cite{kara19} provided tradeoffs between preprocessing time and delay guarantees for enumerating the result of any (not necessarily full) hierarchical query (i.e. $A = \emptyset$ for their setting). However,~\cite{kara19} does not deal with the setting where access patterns can be specified on certain variables of the hierarchical queries, which is the main focus of our work. Despite the difference in the settings, an adaptation of the main algorithm from~\cite{kara19} is able to provide a non-trivial baseline for tradeoffs between space usage and answering time that we highlight in this section. The adapted algorithm may be of independent interest as well. We begin by describing the main algorithm for the static enumeration of hierarchical queries.

Given a hierarchical query $\varphi(\by)$, the preprocessing phase in~\cite{kara19} takes the free variables\footnote{The terminology in~\cite{kara19} calls variables in the head as free and all other variables of the query as bound (not to be confused with $\bx_A$, which we refer to as bound in this paper).} $\by$ and constructs a list of skew-aware view trees and heavy/light indicator views. Starting from the root, for any bound variable that \emph{violates} the free-connex property, two evaluation strategies are used. The first strategy materializes a subset of the query result obtained for the light values over the set of variables $\mathsf{anc}(W) \cup \{W\}$\footnote{\textsf{anc}($W$) is defined as the variables on the path from $W$ to the root excluding $W$.} in the variable order. It also aggregates away the bound variables in the subtree rooted at $W$. Since the light values have a bounded degree, this materialization is inexpensive. The second strategy computes a compact representation of the rest of the query result obtained for those values over  $\mathsf{anc}(W) \cup \{W\}$ that are heavy (i.e., have high degree) in at least one relation. This second strategy treats $W$ as a free variable and proceeds recursively to resolve further bound variables located below $W$.

\begin{example}
    For the query in~\autoref{fig:hierarchy}, the following views are constructed: $V_X(Z_1, Z_2, Z_3, Z_4), V_{Y_1}(X,Z_1, Z_2),$ $ V_{Y_2}(X, Z_3, Z_4)$ for light degree threshold $N^\epsilon, 0 \leq \epsilon \leq 1$ of $\{X\}, \{X, Y_1\}, \{X,Y_2\}$ respectively (assume relation size $N$). The indicator views are $H_X(X)$ and $L_X(X)$ that store which $X$ values are heavy and light. Proceeding recursively, we treat $A$ as a free variable and construct the indicator views $H_{Y_1}(X,Y_1)$ and $L_{Y_1}(X,Y_1)$ (and similarly for the right subtree). The heavy indicator views have size at most $O(N^{1-\epsilon})$.
\end{example}

The enumeration phase then uses the views as an input to the so called \texttt{UNION} and \texttt{PRODUCT} algorithm that combines the output of $O(1)$ number of view trees and uses the observation that the heavy bound variables themselves form a hierarchical structure. 

\begin{example}
    Continuing the example, in the enumeration phase, the output from $V_X(Z_1, Z_2, Z_3, Z_4)$ is available with constant delay. For 
    each heavy $X$ in $H_X(X)$, we proceed recursively and either use the view $V_{Y_1}(X,Z_1, Z_2)$, whose output can be accessed with constant delay for all light $\{X,Y_1\}$ or use the observation that the only remaining case for the left subtree is when $\{X,Y_1\}$ is heavy, allowing us to enumerate the $(Z_1,Z_2)$ answers from the subqueries $Q_{Z_1}(X,Y_1,Z_1) = R(X, Y_1, Z_1)$ and $Q_{Z_2}(X,Y_1,Z_2) = S(X, Y_1, Z_2)$, both of which allow constant delay enumeration for a given fixing of $\{X,Y_1\}$. The overall delay is $O(N^{1 - \epsilon})$.
\end{example}

\begin{theorem}[\cite{kara19}]\label{thm:main_static}
Given a hierarchical query with static width $\mw$, a database  
of size $N$, and $\epsilon \in [0,1]$, the query result can be enumerated with $O(N^{1-\epsilon})$ delay after $O(N^{1 + (\mw -1)\epsilon})$ preprocessing time and space.
\end{theorem}

The static width $\mw$ is a generalization of fractional hypertree width that takes the structure of the free variables in the query into account. For the query in~\autoref{fig:hierarchy}, $\mw = 4$.

\smallskip
\introparagraph{Adapted Algorithm} In order to adapt the algorithm for enumerating hierarchical query results to apply to our setting, we make two key observations. First, we apply the preprocessing phase of~\cite{kara19} with the free variables as $\bx_A$ and construct the data structures. Then, we can construct an indexed representation of the base relations, and the view trees where the indexing variables are the set $\bx_A$. In other words, for any view $V$ and each $v \in \Pi_{\vars{V} \cap A} (V)$, we store  $V \ltimes v$ in a hashtable $H_V$ with key as $v$. Second, given a fixing (say $z$) for $\bx_A$, we provide the indexed base relations, view indicators, and the views (i.e. $H_V[\Pi_{\vars{V} \cap A}(z)]$) to the enumeration algorithm that enumerates the query result with $O(N^{1-\epsilon})$ delay. Since each base relation $R' = R \ltimes z$ provided as input has only one possible value for $\vars{R} \cap A$ (which is $\Pi_{\vars{R} \cap A}(z)$), the only answer the enumeration algorithm can output is either $z$ or declare that there is no answer. This allows us to answer the Boolean hierarchical CQAP successfully. Let $\mw$ denote the static width of the hierarchical query obtained from~\cite{kara19} with free variables as $\bx_A$. We can obtain the following tradeoff.

% \hangdong{Some thoughts: Dan's results are  very weak for us, because we do Boolean
% \begin{enumerate}
%     \item [(1)] Perhaps we don't need all the details, just directly use Theorem H.4 and justify that we can recover and improve it, for the special type of hierarchical
%     \item [(2)] Apply our result of single decomposition here, and say it gets $|\mD|^2 = T \cdot S^{3/4}$. However, a fine-grained analysis leads to $|\mD|^4 = T^4 \cdot S$ and $|\mD|^2 = T \cdot S$, strictly better
% \end{enumerate}
% }
\begin{theorem} \label{thm:static:adapted}
Given a Boolean hierarchical CQAP (i.e. $H = A$), a database  
of size $N$, $\mw$ as static width with free variables as $\bx_A$, and $\epsilon \in [0,1]$, the query can be answered in time $O(N^{1-\epsilon})$ using a preprocessed data structure that takes space $S = O(N^{1 + (\mw -1)\epsilon})$. 
\end{theorem}

\begin{example} \label{ex:static}
    Applying~\autoref{thm:static:adapted} to the CQAP (we use $\bZ$ as a shorthand for the set $\{Z_1, Z_2, Z_3, Z_4\}$)
    $$ \varphi(\bZ \mid \bZ) \leftarrow Q_A(\bZ)  \wedge R(X,Y_1, Z_1) \wedge S(X,Y_1, Z_2) \wedge T(X,Y_2, Z_3) \wedge U(X, Y_2, Z_4)
    $$
    that corresponds to~\autoref{fig:hierarchy} gives us the tradeoff $S \cdot T^3  \cong |\mD|^4$ for any instantiation of the query. Here, $\bZ$ is the access pattern.
\end{example}

Next, we show how the tradeoff in~\autoref{ex:static} can be recovered by our framework and improved to also takes the access request $Q_A$ into account. \autoref{fig:decomp} shows the (free-connex) tree decomposition we consider for the query. The set of PMTDs ($5$ of them) induced from this (free-connex) tree decomposition contains $5$ PMTDs, by choosing the materialization set $M$ to be one of the following: $\{T_0\}$, $\{T_1\}$, $\{T_2\}$, $\{T_1, T_2\}$ or $\emptyset$. Let 
$$ \textsf{body} = Q_A(\bZ)  \wedge R(X,Y_1, Z_1) \wedge S(X,Y_1, Z_2) \wedge T(X,Y_2, Z_3) \wedge U(X, Y_2, Z_4), 
$$
we generate from the $5$ PMTDs the following $2$-phase disjunctive rules:
\begin{align*}
 T_{0}(\bZ, X) \vee S_{1234}(\bZ)  & \leftarrow \textsf{body}  \\
 T_{1}(Z_1, Z_2, Y_1, X)  \vee S_{12}(Z_1, Z_2, X) \vee S_{1234}(\bZ) & \leftarrow \textsf{body}  \\
 T_{1}(Z_1, Z_2, Y_1, X)  \vee T_{2}(Z_3, Z_4, Y_2, X) \vee S_{34}(Z_3, Z_4, X) \vee S_{1234}(\bZ)  &\leftarrow \textsf{body}  \\
 T_{1}(Z_1, Z_2, Y_1, X) \vee  T_{2}(Z_3, Z_4, Y_2, X) \vee S_{12}(Z_1, Z_2, X) \vee S_{1234}(\bZ) & \leftarrow \textsf{body} 
\end{align*}

We now construct the proof sequence for each of the four rules. For the first rule, we get the tradeoff $S \cdot T^3  \cong  |\mD|^4 \cdot |Q_A|^3$ as shown below
\begin{align*}
     4 \log |\mD| + 3 \log |Q_A(\bZ)|  & \geq 3 \textcolor{black}{\hT(X)} + \textcolor{black}{\hS(Z_1 Y_1 X | X)} + \textcolor{black}{\hS(Z_2 Y_1 X | X)} + \textcolor{black}{\hS(Z_3 Y_2 X | X)} + \textcolor{black}{\hS(Z_4 Y_2 X)}  + 3 \textcolor{black}{\hT(\bZ)}  \\
     & \geq 3 \textcolor{black}{\hT(X)} + \textcolor{black}{\hS(Z_1 Y_1 X | X)} + \textcolor{black}{\hS(Z_2 Y_1 X | X)} + \textcolor{black}{\hS(Z_4 Z_3 Y_2 X | Z_4 X)} + \textcolor{black}{\hS(Z_4 X)}  + 3 \textcolor{black}{\hT(\bZ)} \\
      & = 3 \textcolor{black}{\hT(X)} + \textcolor{black}{\hS(Z_1 Y_1 X | X)} + \textcolor{black}{\hS(Z_2 Y_1 X | X)} + \textcolor{black}{\hS(Z_4 Z_3 Y_2 X)}  + 3 \textcolor{black}{\hT(\bZ)} \\
      & \geq3 \textcolor{black}{\hT(X)} + \textcolor{black}{\hS(Z_1 Y_1 X | X)} + \textcolor{black}{\hS(Z_4 Z_3 Z_2 Y_1 X | Z_4 Z_3 X)} + \textcolor{black}{\hS(Z_4 Z_3 X)}  + 3 \textcolor{black}{\hT(\bZ)} \\
      & = 3 \textcolor{black}{\hT(X)} + \textcolor{black}{\hS(Z_1 Y_1 X | X)} + \textcolor{black}{\hS(Z_4 Z_3 Z_2 Y_1 X)}  + 3 \textcolor{black}{\hT(\bZ)}  \\
      & \geq 3 \textcolor{black}{\hT(X)} + \textcolor{black}{\hS(Z_4 Z_3 Z_2 Z_1 Y_1 X | Z_4 Z_3 Z_2 X)} + \textcolor{black}{\hS(Z_4 Z_3 Z_2 X)}  + 3 \textcolor{black}{\hT(\bZ)}   \\
       & = 3 \textcolor{black}{\hT(X)} + \textcolor{black}{\hS(Z_4 Z_3 Z_2 Z_1 Y_1 X)}   + 3 \textcolor{black}{\hT(\bZ)}\\
       & \geq 3 \textcolor{black}{\hT(X \bZ)} + \textcolor{black}{\hS(\bZ)} && (S \cdot T^3  \cong  |\mD|^4 \cdot |Q_A|^3)
\end{align*}
For the rest of the rules, we get the tradeoff $S \cdot T  \cong  |\mD|^2 \cdot |Q_A|$. We show the proof sequence for the second rule below. The proof sequences for the third and fourth rules are very similar and thus omitted.
\begin{align*}
             2 \log |\mD| + \log |Q_A(\bZ)| & \geq \textcolor{black}{\hT(Y_1 X)}  + \textcolor{black}{\hS(Z_1 Y_1 X | Y_1 X)} + \textcolor{black}{\hS(Z_2 Y_1 X)} + \textcolor{black}{\hT(Z_2 Z_1)} \\
             & \geq \textcolor{black}{\hT(Y_1 X)}  + \textcolor{black}{\hS(Z_2 Z_1 Y_1 X | Z_2 Y_1 X)} + \textcolor{black}{\hS(Z_2 Y_1 X)} + \textcolor{black}{\hT(Z_2 Z_1)} \\
             & =  \textcolor{black}{\hT(Y_1 X)}  + \textcolor{black}{\hS(Z_2 Z_1 Y_1 X)} + \textcolor{black}{\hT(Z_2 Z_1)}\\
             & \geq \textcolor{black}{\hT(Z_2 Z_1 Y_1 X)}  + \textcolor{black}{\hS(Z_2 Z_1 X)} && (S \cdot T  \cong  |\mD|^2 \cdot |Q_A|)
\end{align*}
Overall, we obtain the tradeoff $S \cdot T^3  \cong  |\mD|^4 \cdot |Q_A|^3$ since $S \cdot T  \cong  |\mD|^2 \cdot |Q_A|$ is dominated by it.

\smallskip\introparagraph{Improved Tradeoffs} We now show an alternative proof sequence that can improve upon the tradeoff $S \cdot T^3  \cong  |\mD|^4 \cdot |Q_A|^3$ for our running example query. The key insight is to bucketize on the bound variables rather than the free variables, an idea also used by~\cite{deep2021enumeration}. We fix the same disjunctive rules from before. For the first rule, we can obtain an improved tradeoff as follow:
\begin{equation}\label{improvetradeoff}
    \begin{aligned}
    4 \log |\mD| + 4 \log |Q_A(\bZ)|  &= \textcolor{black}{\hT(Y_1 X Z_1 | Z_1)} + \textcolor{black}{\hS(Z_1)} + \textcolor{black}{\hT(Y_1 X Z_2 | Z_2)} + \textcolor{black}{\hS(Z_2)} + \textcolor{black}{\hT(Y_2 X Z_3| Z_3)} \\
     & \qquad + \textcolor{black}{\hS(Z_3)} + \textcolor{black}{\hT(Y_2 X Z_4| Z_4)} + \textcolor{black}{\hS(Z_4)}  + 4 \textcolor{black}{\hT(\bZ)}  \\
     & \geq \textcolor{black}{\hS(\bZ)} + \textcolor{black}{\hT(Y_1 X Z_1 | Z_1)} + \textcolor{black}{\hT(\bZ)} + \textcolor{black}{\hT(Y_1 X Z_2 | Z_2)} + \textcolor{black}{\hT(\bZ)} \\
     & \qquad + \textcolor{black}{\hT(Y_2 X Z_3| Z_3)} + \textcolor{black}{\hT(\bZ)} + \textcolor{black}{\hT(Y_2 X Z_4| Z_4)} + \textcolor{black}{\hT(\bZ)} \\
     & \geq \textcolor{black}{\hS(\bZ)} + 4\textcolor{black}{\hT(X \bZ)} && (S \cdot T^4  \cong  |\mD|^4 \cdot |Q_A|^4)
    \end{aligned}
\end{equation}
The above proof sequence generates the tradeoff $S \cdot T^4  \cong  |\mD|^4 \cdot |Q_A|^4$, a clear improvement for $|Q_A| = 1$. For the rest of the rules, we keep the tradeoff derived above, i.e. $S \cdot T  \cong  |\mD|^2 \cdot |Q_A|^2$. 

Note that the tradeoff $S \cdot T^3  \cong  |\mD|^4 \cdot |Q_A|^3$ dominates both $S \cdot T^4  \cong  |\mD|^4 \cdot |Q_A|^4$ for the first rule and $S \cdot T  \cong  |\mD|^2 \cdot |Q_A|^2$ for the rest of the rules. So we get a strictly improved tradeoff across all regimes.

\introparagraph{Capturing~\autoref{thm:static:adapted} in Our Framework} Before we conclude this section, we present a general strategy to capture the tradeoff from~\autoref{thm:static:adapted} for a subset of hierarchical queries. In particular, we show that for any Boolean hierarchical CQAP that contains $\bx_A$ only in the leaf variables, there is a proof sequence that recovers the tradeoff obtained from~\autoref{thm:static:adapted}. We will use $\bZ$ to denote the leaf variables. Recall that each hierarchical query admits a \emph{canonical ordering}\footnote{The original definition in~\cite{kara19} also has a dependency function but we omit that since we do not use it for the decomposition construction.}, which is a rooted tree where the variables of each atom in the query lie along the same root-to-leaf path in the tree and each atom is a child of its lowest variable.

We begin by describing the query decomposition that we will use. Consider the canonical variable ordering of the hierarchical query. The root bag of the decomposition $T_0$ consists of $\bx_A$ and the variable at the root of the variable ordering (say $X$). For each child $Y_i$ of $X$, we add a child bag of $T_i$ containing the subset of $\bx_A$ in the subtree rooted at $Y_i$ and $\anc{(Y_i)} \cup Y_i$. We continue this procedure by traversing the variable ordering in a top-down fashion and processing all non-bound variables. It is easy to see that the tree obtained is indeed a valid decomposition.

\begin{example}
\autoref{fig:decomp} shows the query decomposition generated from the canonical ordering in~\autoref{fig:hierarchy}. The root bag contains all bound variables and $X$. $X$ contains two children $Y_1, Y_2$ so the decomposition contains two children of the root node. The left child contains $X, Y_1$ and the subtree rooted at $Y_1$ contains $Z_1, Z_2$ as the bound variables, which are added to the left node in the decomposition. Similarly, the right node contains $X,Y_2, Z_3, Z_4$.
\end{example}

Similar to the running example, we now construct the set of PMTDs induced from this decomposition and generate the corresponding 2-phase disjunctive rules. It is easy to see that every disjunctive rule contains $S(\bZ)$. 

\eat{First, we show that the disjunctive rule $T_0(\bZ X) \vee S(\bZ)$ achieves the tradeoff $S 
\cdot T^{\mw-1} \cong |\mD|^\mw \cdot |Q_A(\bZ)|^{\mw-1}$. Indeed,

\begin{align*}
     \mw \log |\mD| + (\mw-1) \log |Q_A(\bZ)|  &= (\mw-1) \textcolor{black}{\hT(X)} + \sum_{i=1}^{\mw-1} \textcolor{black}{\hS(\{\anc{(Z_i)} \cup Z_i \} \mid X)} + \textcolor{black}{\hS(\anc{(Z_\mw)} \cup Z_\mw)} + (\mw-1) \textcolor{black}{\hT(\bZ)}\\
     &\geq (\mw-1) \textcolor{black}{\hT(X \bZ)} + \sum_{i=1}^{\mw-2} \textcolor{black}{\hS(\{\anc{(Z_i)} \cup Z_i \} \mid X)} + \textcolor{black}{\hS(\{\anc{(Z_{\mw-1})} \cup Z_{\mw-1} \cup Z_{\mw} \} \mid Z_{\mw} X)} + \textcolor{black}{\hS(Z_\mw X)} \\
     &\geq (\mw-1) \textcolor{black}{\hT(X \bZ)} + \sum_{i=1}^{\mw-2} \textcolor{black}{\hS(\{\anc{(Z_i)} \cup Z_i \} \mid X)} + \textcolor{black}{\hS(\anc{(Z_{\mw-1})} \cup Z_{\mw-1} \cup Z_{\mw})} \\
     &  \hspace{0.5em}  \vdots \\
     &\geq (\mw-1) \textcolor{black}{\hT(X \bZ)} + \textcolor{black}{\hS(\anc{(Z_{1})} \cup Z_1 \cup \dots \cup Z_{\mw-1} \cup Z_{\mw})} \\
     &\geq (\mw-1) \textcolor{black}{\hT(X \bZ)} + \textcolor{black}{\hS(X \bZ)} \\
     &\geq (\mw-1) \textcolor{black}{\hT(X \bZ)} + \textcolor{black}{\hS(\bZ)}
\end{align*}}

Let $\bU_1 \subseteq \bZ$ and $\mw_1 = |\bU_1|$. For any disjunctive rule, it must contain a term of the form $T(\bU_1 p)$ and $S(\bZ)$. Here, $p$ denotes the set of variables other than the bound variables in the bag corresponding to the $T$-view that is picked. Note that the induced PMTD where $M = \emptyset$ forces any disjunctive rule to have at least one $T$-view in the head of the rule. Consider some $v^\star \in U_1$.
\begin{align*}
     \mw_1 \log |\mD| + (\mw_1 - 1) \log |Q_A(\bZ)|  &=   \textcolor{black}{\hS(\anc(v^\star) \cup v^\star)} +  \sum_{v \in \bZ \setminus v^\star} (\textcolor{black}{\hT(p)} + \textcolor{black}{\hS(\anc(v) \cup v  \mid p)}) + (\mw_1 - 1) \textcolor{black}{\hT(\bZ)} \\
     &\geq (\mw_1 - 1) \textcolor{black}{\hT(\bU_1 p)} + \textcolor{black}{\hS(\anc(v^\star) \cup v^\star)} +  \sum_{v \in \bZ \setminus v^\star} \textcolor{black}{\hS(\anc(v) \cup v  \mid p)}  \\
     &= (\mw_1 - 1) \textcolor{black}{\hT(\bU_1 p)} + \textcolor{black}{\hS(\anc(v_1) \cup v_1  \mid p)} + \dots + \textcolor{black}{\hS(\anc(v_k) \cup v_k \cup v^\star  \mid v^\star \cup p)} + \textcolor{black}{\hS(\anc(v^\star) \cup v^\star)} \\
     &\geq (\mw_1 - 1) \textcolor{black}{\hT(\bU_1 p)} + \textcolor{black}{\hS(\anc(v_1) \cup v_1  \mid p)} + \dots + \textcolor{black}{\hS(\anc(v_k) \cup v_k \cup v^\star  \mid v^\star \cup p)} + \textcolor{black}{\hS(p \cup v^\star)} \\
     &\geq (\mw_1 - 1) \textcolor{black}{\hT(\bU_1 p)} + \textcolor{black}{\hS(\anc(v_1) \cup v_1  \mid p)} + \dots + \textcolor{black}{\hS(\anc(v_k) \cup v_k \cup v^\star )} \\
     &\geq (\mw_1 - 1) \textcolor{black}{\hT(\bU_1 p)} + \textcolor{black}{\hS(\anc(v_1) \cup v_1 \dots \cup v_k \cup v^\star )} \\
     &\geq (\mw_1 - 1) \textcolor{black}{\hT(\bU_1 p)} + \textcolor{black}{\hS(\bZ)} \hspace{15em} (S \cdot T^{\mw_1-1} \cong |\mD|^{\mw_1} \cdot |Q_A|^{\mw_1 - 1})
\end{align*}

The tradeoff is the most expensive when $\mw_1$ is as large as possible. Thus, for $\mw_1 = \mw$, which corresponds to $U_1 = \bZ$, we achieve the tradeoff $S 
\cdot T^{\mw-1} \cong |\mD|^\mw \cdot |Q_A(\bZ)|^{\mw-1}$. However, for the same disjunctive rule $T_{0}(\bZ, A) \vee S_{\bx_A}(\bZ)$ that gives the dominating tradeoff, we can also obtain a different proof sequence that provides an improvement, similar to \eqref{improvetradeoff}
\eat{\begin{align*}
     \mw \log |\mD| + \mw \log |Q_A(\bZ)|  &= \mw\ \textcolor{black}{\hT(X)} + \sum_{i=1}^{\mw} \textcolor{black}{\hT(\{\anc{(Z_i)} \cup Z_i \} \mid Z_i)} + \mw\ \textcolor{black}{\hT(\bZ)} \\
     &\geq \mw\ \textcolor{black}{\hS(\bZ)} + \sum_{i=1}^{\mw}(\textcolor{black}{\hT(X)} + \textcolor{black}{\hT(\bZ)}) \\
     &\geq \mw\ \textcolor{black}{\hS(\bZ)} + \textcolor{black}{\hT(A\bZ)} \hspace{15em} (S \cdot T^{\mw} \cong |\mD|^{\mw} \cdot |Q_A|^{\mw})
\end{align*}}

\begin{align*}
     \mw \log |\mD| + \mw \log |Q_A(\bZ)|  &= \textcolor{black}{\hS(\bZ)} + \sum_{i=1}^{\mw} \textcolor{black}{\hT(\anc{(Z_i)} \cup Z_i  \mid Z_i)} + \mw\ \textcolor{black}{\hT(\bZ)} \\
     &=\textcolor{black}{\hS(\bZ)} + \sum_{i=1}^{\mw} (\textcolor{black}{\hT(X \cup Z_i  \mid Z_i)} + \textcolor{black}{\hT(\bZ)}) \\
     &\geq \textcolor{black}{\hS(\bZ)} + \mw  \cdot \textcolor{black}{\hT(X\bZ)}\hspace{15em} (S \cdot T^{\mw} \cong |\mD|^{\mw} \cdot |Q_A|^{\mw})
\end{align*}

\eat{
The only case that remains is when $\bU_1 \cap \bU_2 = \emptyset$. For this case, we get the following:

\begin{align*}
     (\mw_1 + \mw_2) \log |\mD| + (\mw_1 - 1) \log |Q_A(\bZ)|  &=   \sum_{v \in \bU_1 } (\textcolor{black}{\hT(q)} + \textcolor{black}{\hS(\anc(v) \cup v \mid q)}) \\
     &+ \sum_{v \in \bU_2  \setminus v_2^\star} (\textcolor{black}{\hT(p)} + \textcolor{black}{\hS(\anc(v) \cup v \mid p)}) + \textcolor{black}{\hS(\anc(v_2^\star) \cup v_2^\star)} + (\mw_1 -1 ) \textcolor{black}{\hT(\bZ)} \\
     &\geq (\mw_1 - 1 ) \textcolor{black}{\hT(\bU_1 q)} + \sum_{v \in \bU_2  \setminus v_2^\star} (\textcolor{black}{\hT(p)} + \textcolor{black}{\hS(\anc(v) \cup v \mid p)}) + \textcolor{black}{\hS(\anc(v_2^\star) \cup v_2^\star)}\\
     &\geq (\mw_1 - 1 ) \textcolor{black}{\hT(\bU_1 q )} +\textcolor{black}{\hS(\anc(v_1) \cup v_1 \mid p)} + \dots + \textcolor{black}{\hS(\anc(v_k) \cup v_k \cup v_2^\star \mid v_2^\star p)} + \textcolor{black}{\hS(p \cup v_2^\star)} \\
     &\geq (\mw_1 - 1 ) \textcolor{black}{\hT(\bU_1 q )} + \textcolor{black}{\hS(\bU_2 p)}
\end{align*}

Note that $\mw_1 + \mw_2 \leq \mw$ since the $S$ and $T$ bags do not overlap and together, they can cover all variables $\bZ$. Further, $\mw_1 \leq \mw - 1$ as explained before.}